\documentclass[preprint,11pt]{article}

\usepackage{amssymb,amsfonts,amsmath,amsthm,amscd,dsfont,mathrsfs,bm}
\usepackage{graphicx,float,psfrag,epsfig}
\usepackage{wrapfig}
\usepackage{relsize}
\usepackage{hyperref}
\usepackage{color}
\usepackage{pict2e}
\usepackage[tight]{subfigure}
\usepackage{multirow}
\usepackage{array}
\usepackage{caption}

\DeclareMathAlphabet{\mathpzc}{OT1}{pzc}{m}{it}

\newtheorem{propo}{Proposition}[section]
\newtheorem{lemma}[propo]{Lemma}
\newtheorem{definition}[propo]{Definition}
\newtheorem{coro}[propo]{Corollary}
\newtheorem{thm}[propo]{Theorem}

\newtheorem{remark}[propo]{Remark}
\newtheorem{claim}[propo]{Claim}
\newtheorem{rclaim}[propo]{Replica Method Claim}

\newcommand{\mtx}[1]{\bm{#1}}

\def\ed{\stackrel{\rm d}{=}}
\def\Det{{\rm Det}}

\def\tx{\widetilde{x}}

\def\cZ{{\cal Z}}
\def\cF{{\cal F}}

\def\cR{{\cal R}}

\def\cS{{\cal S}}

\def\naturals{{\mathbb N}}

\def\reals{{\mathbb R}}

\def\ve{{\varepsilon}}
\def\eps{{\varepsilon}}
\def\prob{{\mathbb P}}
\def\E{{\mathbb E}}

\def\scale{{\sf d}}
\def\tscale{\tilde{\sf d}}
\def\UB{{\tiny \rm UB}}
\def\LB{{\tiny \rm LB}}

\def\normal{{\sf N}}

\def\Map{{\sf F}}

\def\Energy{{\sf E}}

\def\L0{{L_0}}

\def\de{{\rm d}}
\def\<{\langle}
\def\>{\rangle}

\def \div{{\rm div}\,}

\def\htheta{\widehat{\theta}}
\def\hSigma{\widehat{\Sigma}}

\def\supp{{\rm supp}}

\def\ind{{\mathbb I}}

\def \Tr{{\rm Trace}}

\def\normal{{\sf N}}

\def\P{{\mathbb{P}}}

\def\sT{{\sf T}}
\def\Free{{\mathfrak F}}
\def\tg{\widetilde{g}}
\def\hxi{\widehat{\xi}}
\def\tJ{\widetilde{J}}
\def\ttJ{\widetilde{\widetilde{J}}}
\def\Energy{{\mathfrak{E}}}
\def\Ena{{\sf E}_1}
\def\Enb{{\sf E}_2}
\def\D{{\rm D}}
\def\id{{\rm I}}
\def\rank{{\rm rank}}
\def \sdl{\ensuremath{\text{\textsc{SDL-test}}}}
\def \ldpe{\ensuremath{\text{\textsc{LDPE}}}}
\def \Subroutine {\ensuremath{\text{\textsc{Subroutine}}}}

\def\bbin{\beta^{\rm bin}}
\def\bopt{\beta^{\rm opt}}
\def\bora{\beta^{\rm oracle}}
\def\proj{{\rm{\bf P}}}
\def\projp{{\rm{\bf P}}^{\perp}}

\def\bu{{\mtx{u}}}
\def\bv{{\mtx{v}}}
\def\by{{\mtx{y}}}
\def\bx{{\mtx{x}}}
\def\bX{{\mtx{X}}}
\def\bR{{\mtx{R}}}
\def\bz{{\mtx{z}}}

\def\bSigma{\mtx{\Sigma}}
\def\bOmega{\mtx{K}}

\def\hbSigma{\mtx{\hSigma}}
\def\cN{{\cal N}}

\def\bh{{\mtx{h}}}

\definecolor{olivegreen}{rgb}{0,0.6,0.4}
\definecolor{darkblue}{rgb}{0,0.2,0.8}
\newcommand{\ajcomment}[1]{}

\title{Hypothesis Testing in High-Dimensional Regression
under the Gaussian Random Design Model: Asymptotic Theory}

\author{Adel Javanmard${}^{*}$ and Andrea~Montanari 
            \footnote{Department of Electrical Engineering, Stanford University}
            \footnote{Department of Statistics, Stanford University}}

\begin{document}

\maketitle

\begin{abstract}
We consider linear regression in the high-dimensional regime where
the number of observations $n$ is smaller than the number of parameters
$p$. A very successful approach in this setting uses $\ell_1$-penalized
least squares (a.k.a. the Lasso) to search for a subset of $s_0< n$  parameters
that best explain the data, while setting the other parameters to zero.
Considerable amount of work has been devoted to characterizing the
estimation and model selection problems within this approach. 

In this paper we consider instead the fundamental, but  far less
understood, question of \emph{statistical significance}. 
More precisely, we address the problem of computing p-values for single regression
coefficients.

On one hand, we develop  a general  upper bound on the minimax
power of tests with a given significance level. We show that
rigorous guarantees for earlier methods do not allow to achieve this
bound, except in special cases.
On the other, we prove that this upper bound is (nearly) achievable through a
practical procedure in the case of random design matrices with
independent entries. Our approach is based on a debiasing of the Lasso
estimator. The analysis builds on a rigorous
characterization of the asymptotic distribution of the Lasso estimator
and its debiased version. Our result holds for optimal sample size, i.e.,
when $n$ is at least on the order of $s_0 \log(p/s_0)$.

We generalize our approach to  random design matrices with
i.i.d. Gaussian rows  
$\mtx{x}_i\sim \normal(0,\mtx{\Sigma})$. In this case we prove that a
similar distributional characterization 
(termed `standard distributional limit') holds 
for $n$ much larger than $s_0(\log p)^2$. 
Our analysis assumes $\bSigma$ is known.
%Our analysis  leaves open two important problems:
%$(i)$  It  assumes $\bSigma$ is known; $(ii)$ It does not imply a lower
%bound on the statistical power. (Earlier work assumes $\bSigma$
%unknown or deterministic designs.)
To cope with unknown $\bSigma$, we suggest a plug-in estimator for sparse covariances $\mtx{\Sigma}$
and validate the method through numerical simulations.

Finally, we show that for optimal sample size, $n$ being at least of order $s_0 \log(p/s_0)$,
the standard distributional limit for general Gaussian
designs can be derived from the replica heuristics  in statistical physics.
This derivation suggests a stronger conjecture than the result we
prove, and near-optimality of the statistical power for a large
class of Gaussian designs.
\end{abstract}

%==============================================================

\section{Introduction}
\label{sec:Introduction}

The Gaussian random design model for linear regression is defined as
follows. We are given $n$ i.i.d. pairs $(y_1,\mtx{x}_1)$, $(y_2,\mtx{x}_2)$, $\cdots$,
$(y_n,\mtx{x}_n)$ with $y_i\in \reals$ and $\mtx{x}_i\in\reals^p$, $\mtx{x}_i\sim
\normal(0,\mtx{\Sigma})$ for some covariance matrix $\mtx{\Sigma}\succ
0$.  Further, $y_i$ is a linear function of  $\mtx{x}_i$, plus noise
\begin{eqnarray}\label{eqn:regression}
y_i \,=\, \<\mtx{\theta}_0,\mtx{x}_i\> + w_i\, ,\;\;\;\;\;\;\;\; w_i\sim
\normal(0,\sigma^2)\, .
\end{eqnarray}
Here $\mtx{\theta}_0\in\reals^p$ is a vector of parameters to be estimated and
$\<\,\cdot\,,\,\cdot\,\>$ is the standard scalar product. 
The special case $\mtx{\Sigma}= \id_{p\times p}$ is usually referred to as
`standard' Gaussian design model.

In matrix form,
letting  $\mtx{y} = (y_1,\dots,y_n)^\sT$ and denoting by $\bX$ the matrix with
rows $\mtx{x}_1^\sT$,$\cdots$, $\mtx{x}_n^\sT$ we have
\begin{eqnarray}\label{eq:NoisyModel}
\mtx{y}\, =\, \bX\,\mtx{\theta}_0+ \mtx{w}\, ,\;\;\;\;\;\;\;\; \mtx{w}\sim
\normal(0,\sigma^2 \id_{n\times n})\, .
\end{eqnarray}
We are interested in high-dimensional settings where the number of
parameters exceeds the sample size, i.e., $p > n$, but the number of
non-zero entries of $\mtx{\theta}_0$ (to be denoted by $s_0$) is smaller than $p$.
In this situation, a recurring problem is to select the non-zero entries
of $\mtx{\theta}_0$ that hence can provide a succinct
explanation of the data. The vast literature on this topic is
briefly overviewed in Section \ref{sec:Related}. 

The Gaussian design assumption arises naturally in
some important applications. Consider
for instance the problem of learning a high-dimensional Gaussian
graphical model from data. In this case we are given i.i.d. samples
$\mtx{z}_1,\mtx{z}_2,\dots,\mtx{z}_n\sim \normal(0,\mtx{K}^{-1})$, with
$\mtx{K}$ a sparse positive definite matrix whose non-zero entries
encode the underlying graph structure. As  first shown by Meinshausen
and B{\"u}hlmann \cite{MeinshausenBuhlmann}, the $i$-th row of
$\mtx{K}$ can be estimated by performing linear
regression of the $i$-th entry of the samples $\bz_{1},\bz_{2},\dots,\bz_{n}$ onto the other
entries \cite{ren2013asymptotic}. This reduces the problem to a
high-dimensional regression model under Gaussian designs.
Standard Gaussian designs were also shown to provide useful insights
for compressed sensing applications \cite{Do,DoTa05,DoTa08,DoTa10}.

\vspace{0.3cm}  

In  statistics and signal processing applications, it is unrealistic to assume that the set
of nonzero entries  of $\mtx{\theta}_0$ can be determined with absolute
certainty. 
The present paper focuses on the problem of quantifying the
\emph{uncertainty} associated  to the entries of $\mtx{\theta}_0$. 
More specifically, we are interested in testing null-hypotheses of the form:
\begin{eqnarray}
H_{0,i} :\;\;\; \theta_{0,i} = 0,\label{eq:Hypothesis}
\end{eqnarray}   
for $i\in [p]\equiv \{1,2,\dots,p\}$
and assigning p-values for these tests. Rejecting $H_{0,i}$ is
equivalent to stating that $\theta_{0,i}\neq 0$. 

Any hypothesis testing procedure faces two types of errors:
false positives or type I errors (incorrectly rejecting $H_{0,i}$, while
$\theta_{0,i}=0$), and false negatives or type II errors (failing to reject $H_{0,i}$, while
$\theta_{0,i}\neq 0$). The probabilities of these two
types of errors will be denoted, respectively, by  $\alpha$ and $\beta$
(see Section~\ref{subsec:guaranteed} for a more precise definition).  The quantity
$1-\beta$ is also referred to as the power of the test, and $\alpha$ as its
significance level. It is trivial to achieve $\alpha$ arbitrarily
small if we allow for $\beta=1$ (never reject $H_{0,i}$) or $\beta$
arbitrarily small if we allow for $\alpha=1$ (always reject
$H_{0,i}$). This paper aims at optimizing the trade-off between power $1-\beta$
and significance $\alpha$.

Without further assumptions on the problem structure, the trade-off is
trivial and no non-trivial lower bound on $1-\beta$ can be established. Indeed we can take $\theta_{0,i}\neq 0$
arbitrarily close to $0$, thus making $H_{0,i}$ in practice
indistinguishable from its complement. We will therefore assume that,
whenever $\theta_{0,i}\neq 0$, we have $|\theta_{0,i}|>\mu$ as well. 
The smallest value of $\mu$ such that the power and significance reach
some fixed non-trivial value (e.g., $\alpha = 0.05$ and $1-\beta\ge
0.9$) has a particularly compelling interpretation, and provides an
answer to the following question:
What is the minimum magnitude of $\theta_{0,i}$ to be able to
distinguish it from the noise level, with a given degree of
confidence?

More precisely, we are interested in establishing necessary and sufficient
conditions on $n$, $p$, $s_0$, $\sigma$ and $\mu$ such that a given
significance level $\alpha$, and power $1-\beta$ can be achieved in testing $H_{0,i}$ for
all coefficient vectors $\mtx{\theta}_0$ that are $s_0$-sparse and $|\theta_{0,i}| > \mu$.
%(under suitable assumptions on $\mtx{X}$).
Some intuition can be gained by considering special cases
(for the sake of comparison, we assume that the columns of $\mtx{X}$
are normalized to have $\ell_2$ norm of order $\sqrt{n}$):
\begin{itemize}
\item In the case of orthogonal designs we have $n=p$ and 
$\mtx{X}^{\sT}\mtx{X} = n\id_{n\times n}$. By an orthogonal
transformation,
we can limit ourselves to $\bX=\sqrt{n}\,\id_{n\times n}$,
i.e., $y_i = \sqrt{n}\,\theta_{0,i} + w_i$. Hence testing hypothesis
$H_{0,i}$ reduces to testing for the mean of a univariate Gaussian.

It is easy to see  that we can distinguish the $i$-th entry
from noise only if its size is at least of order $\sigma/\sqrt{n}$. More precisely, for any $\alpha\in (0,1)$, 
$\beta\in (0,\alpha)$, we can achieve significance $\alpha$ and power $1-\beta$ if 
and only if $|\theta_{0,i}|\ge c(\alpha,\beta)\,
\sigma/\sqrt{n}$ for some constant $c(\alpha,\beta)$ \cite[Section 3.9]{lehmann2005testing}.
\item To move away from the orthogonal case, consider standard
  Gaussian designs. Several papers studied the estimation problem in
  this setting
  \cite{CandesRechtSimple,CandesPlanRIPless,BayatiMontanariLASSO,DMM-NSPT-11}.
 The conclusion is that there exist computationally efficient
 estimators $\mtx{\htheta}$ that are consistent (in high-dimensional
 sense) for $n\ge c_1 s_0\log(p/s_0)$, with $c_1$ a numerical constant. By far the most popular such
 estimator is the Lasso or Basis Pursuit Denoiser \cite{BP95,Tibs96}. 

On the other hand, no practical estimator is known that is consistent
under a significantly smaller sample size (impossibility results
have been proven in this direction, see
e.g. \cite{IndykLower,CandesDavenport}).
We expect hypothesis testing to require at least as large sample size as
point estimation, i.e. $n\ge c_0 s_0\log(p/s_0)$ for some $c_0 = c_0(\alpha,\beta)$.
\end{itemize}
These simple remarks motivate the following seemingly simple question:
\begin{itemize}
\item[{\sf Q:}]\emph{Assume standard Gaussian design $\mtx{X}$, and fix
  $\alpha,\beta\in(0,1)$.
Are there constants $c = c(\alpha,\beta)$, $c_1=c_1(\alpha,\beta)$
and a hypothesis testing procedure achieving the desired significance
and power for all $\mu\ge c\sigma/\sqrt{n}$, $n\ge c_1
s_0\log(p/s_0)$?}
\end{itemize}

\vspace{0.3cm}

Despite the seemingly idealized setting, the answer to this question is highly non-trivial.
To document this point, we consider in Appendix \ref{app:Alternative}
two hypothesis testing methods that were recently proposed by Zhang
and Zhang \cite{ZhangZhangSignificance}, and by B\"uhlmann
\cite{BuhlmannSignificance}.
These approaches apply to a broader class of design matrices $\bX$
that satisfy the restricted eigenvalue property
\cite{BickelEtAl}. 
We show that, when specialized to the case of standard Gaussian designs
$\mtx{x}_i\sim\normal(0,\id_{p\times p})$, these
methods require $|\theta_{0,i}|\ge \mu = c\,\max\{\sigma s_0 \log p/\, n,\sigma/\sqrt{n}\}$ to reject
hypothesis $H_{0,i}$ with a given degree of confidence (with $c$ being a
constant independent of the problem dimensions). 
In other words,  these methods are guaranteed to succeed only if the
coefficient to be tested is larger than the ideal scale
$\sigma/\sqrt{n}$, by a diverging factor  of order $s_0 \log p/\sqrt{n}$.
In particular, the results of
\cite{ZhangZhangSignificance,BuhlmannSignificance} do not allow to answer the
above question.

In this paper, we answer positively to this question.
As in \cite{ZhangZhangSignificance,BuhlmannSignificance}, our approach
is based on the Lasso estimator \cite{BP95,Tibs96}
\begin{eqnarray}
\mtx{\htheta} (\mtx{y},\bX) = \arg\min_{\mtx{\theta}\in\reals^{p}}
\Big\{
\frac{1}{2n} \|\mtx{y} - \bX\mtx{\theta}\|^2 + \lambda\,\|\mtx{\theta}\|_1\Big\}\,. \label{eqn:Lasso_cost}
\end{eqnarray}
We use the solution to this problem to construct a debiased estimator
of the form
\begin{eqnarray}
\mtx{\htheta}^u =  \mtx{\htheta} + \frac{1}{n} \mtx{M} \bX^\sT(\mtx{y}-\bX\mtx{\htheta}),\label{eq:DebiasedDef}
\end{eqnarray}
with $\mtx{M}\in\reals^{p\times p}$ a properly constructed matrix.
We then use its $i$-th component $\htheta_i^u$ as a test statistics for hypothesis
$H_{0,i}$. (We refer to Sections \ref{sec:std} and \ref{sec:nstd} for
a detailed description of our procedure.)

A similar approach was developed independently in
\cite{ZhangZhangSignificance}
and (after a a preprint version of the present paper became available online) in
\cite{GBR-hypothesis}.
Apart from differences in the construction of $\mtx{M}$, the three
papers differ crucially in the assumptions and the regime analyzed, and establish results
that are not directly comparable. In the
present paper we assume a specific (random) model for the design matrix
$\mtx{X}$. In contrast \cite{ZhangZhangSignificance} and
\cite{GBR-hypothesis} assume deterministic designs, or random 
designs with general unknown covariance.

On the other hand, we are able to analyze a regime that is
significantly beyond reach of the mathematical  techniques of
\cite{ZhangZhangSignificance,GBR-hypothesis}, even for the very
special case of standard
Gaussian designs.
Namely, for standard designs, we consider $\mu$  of order $\sigma/\sqrt{n}$, and $n$ of order
$s_0\log(p/s_0)$. 

This regime is both challenging and interesting
because $\theta_{0,i}$ (when non-vanishing)
is of the same order as the noise level.
Indeed our  analysis requires an exact asymptotic distributional characterization 
of the problem (\ref{eqn:Lasso_cost}). 

The contributions of this paper are organized as follows:
\begin{description}
 \item[Section \ref{sec:minimax}: Upper bound on the minimax power.] 
   We state the problem formally, by taking a
   minimax point of view. Based on this formulation, we prove a general upper bound on the minimax power
 of tests with a given significance level $\alpha$. We then specialize
 this bound to the case of standard Gaussian design matrices, showing
 formally that no test can achieve non-trivial significance $\alpha$,
 and power $1-\beta$, unless
 $|\theta_{0,i}|\ge \mu_{\UB} = c\sigma/\sqrt{n}$, with $c$ a
 dimension-independent constant.
\item[Section \ref{sec:std}: Hypothesis testing for standard Gaussian
  designs.] We define a hypothesis testing procedure that is well-suited for the case of
  standard Gaussian  designs, $\mtx{\Sigma}=\id_{p\times p}$. We prove that
  this test achieves a `nearly-optimal' power-significance trade-off 
  in a properly defined asymptotic sense. Here `nearly optimal' means
  that the trade-off has the same form as the previous upper bound, except that $\mu_{\UB}$ is
  replaced by $\mu=C\mu_{\UB}$ with $C$ a universal constant. In
  particular, we provide a positive answer to the open question
  discussed above.

Our analysis builds on an exact asymptotic characterization of the
Lasso estimator, first developed in
\cite{BayatiMontanariLASSO}.
\item[Section \ref{sec:nstd}: Hypothesis testing for nonstandard
  Gaussian designs.] 
We introduce a generalization of the previous hypothesis testing
method to Gaussian designs with general covariance matrix
$\mtx{\Sigma}$.
In this case we cannot establish validity in the regime
$n\ge c_1s_0\log(p/s_0)$, since a rigorous generalization of the distributional
result of \cite{BayatiMontanariLASSO} is not available. 

However: $(1)$ We prove that such a generalized distributional limit 
holds under the stronger assumption that $n$ is much larger
than $s_0(\log p)^2$ (see Theorem \ref{thm:Replica-Rigor}). 
$(2)$ We show that this distributional limit can be derived
from the powerful replica heuristics in statistical physics for the regime
$n\ge c_1s_0\log(p/s_0)$. (See
Section \ref{sec:nstd} for further discussion of the validity of this heuristics.)

 Conditional on this \emph{standard distributional limit}  holding, we prove
that the proposed procedure is nearly optimal in this case as well. 
\item[Numerical validation.] We validate our approach on both synthetic and real
  data in Sections \ref{sec:experiment-std},~\ref{sec:experiment-nstd} and Section~\ref{sec:crime},
  comparing it with the methods of
 \cite{ZhangZhangSignificance,BuhlmannSignificance}. Simulations suggest that the latter are
  indeed overly conservative in the present setting, resulting in suboptimal
  statistical power. (As emphasized above, the methods of
  \cite{ZhangZhangSignificance,BuhlmannSignificance} apply to a
  broader class of design matrices $\bX$.)
\end{description}
Proofs are deferred to Section \ref{sec:Proofs}.

Let us stress that the present treatment has two important
limitations. First, it is asymptotic: it would be important to develop
non-asymptotic bounds.
Second, for the case of general designs, it requires to know or estimate the design covariance
$\mtx{\Sigma}$. In Section \ref{sec:gen-covariance} we discuss a
simple approach to this problem for sparse $\mtx{\Sigma}$. A full study
of this issue is however beyond the scope of the
present paper.  

\vspace{0.3cm}

After a a preprint version of the present paper became available online,
several papers appeared that partially address these limitations.
In particular \cite{GBR-hypothesis,confidenceJM} make use of debiased
estimators of the form (\ref{eq:DebiasedDef}), and have much weaker
assumptions on the design $\bX$. Note however that these papers
require a significantly larger sample size, namely $n\ge (s_0\log
p)^2$. Hence, even limiting ourselves to standard designs, the results presented
here are not comparable to the ones of
\cite{GBR-hypothesis,confidenceJM}, and instead complement them.
We refer to Section \ref{sec:discussion} for further discussion of the relation.

%{\bf [AM: Move this somewhere]}
%This paper focuses on the asymptotic regime introduced in
%\cite{Do,DoTa05,DoTa08,DoTa10}  and studied in
%\cite{DMM09,DMM-NSPT-11,BayatiMontanariLASSO}.
%The advantage of this approach is that the asymptotic characterization
%of \cite{BayatiMontanariLASSO} is sharp and appears to be accurate
%already at moderate sizes. 

%A forthcoming paper \cite{JM-Support} will address the same questions
%in a non-asymptotic setting.

\subsection{Further related work}
\label{sec:Related}

High-dimensional regression and $\ell_1$-regularized least squares estimation,
a.k.a. the Lasso~\eqref{eqn:Lasso_cost}, were the object of much theoretical investigation over the last few years.
The focus  has been so far on establishing order optimal
guarantees on: $(1)$ The prediction error $\|\bX(\mtx{\htheta}-\mtx{\theta}_0)\|_2$, see e.g. \cite{GreenshteinRitov}; $(2)$ The estimation error, typically
quantified through $\|\mtx{\htheta}-\mtx{\theta}_0\|_q$, with $q\in [1,2]$, see e.g. \cite{Dantzig,BickelEtAl,WainwrightEllP}; $(3)$ The model selection
(or support recovery)
properties typically by bounding  $\prob\{\supp(\mtx{\htheta})\neq \supp(\mtx{\theta}_0)\}$, see e.g. \cite{MeinshausenBuhlmann,zhao,Wainwright2009LASSO}.
For estimation and support recovery guarantees, it is
necessary to make specific assumptions on the design matrix $\bX$,
such as the restricted eigenvalue property of \cite{BickelEtAl} or the
compatibility condition of \cite{BuhlmannVanDeGeer}. 
Both \cite{ZhangZhangSignificance} and \cite{BuhlmannSignificance}
assume conditions of this type for developing hypothesis testing
procedures. 

In contrast we work within the Gaussian random design
model, and focus on the asymptotics   $s_0,p,n\to\infty$ with
$s_0/p\to\eps\in(0,1)$ and $n/p\to \delta\in (0,1)$.
The study of this type of high-dimensional asymptotics  
was pioneered by Donoho and Tanner \cite{Do,DoTa05,DoTa08,DoTa10} ,
who assumed standard Gaussian designs and focused on exact recovery in absence of noise.
The estimation error in presence of noise was characterized in
\cite{DMM-NSPT-11,BayatiMontanariLASSO}. Further work in the same or
related setting includes
\cite{Wainwright2010Gaussian,CandesRechtSimple,CandesPlanRIPless}.

Wainwright \cite{Wainwright2009LASSO} also considered the Gaussian design
model and established upper and lower thresholds
$n_{\UB}(p,s_0;\mtx{\Sigma})$, $n_{\LB}(p,s_0;\mtx{\Sigma})$ 
for correct  recovery
of $\supp(\mtx{\theta}_0)$ in noise $\sigma>0$, under an additional
condition on $\mu\equiv \min_{i\in
  \supp(\mtx{\theta}_0)}|\theta_{0,i}|$. 
The thresholds $n_{\UB}(p,s_0;\mtx{\Sigma} )$,
$n_{\LB}(p,s_0;\mtx{\Sigma})$ are of order $s_0\log p$ for many
covariance structures $\mtx{\Sigma}$, provided $\mu\ge C\sqrt{(\log
  p)/n}$ for some constant $C>0$.
Correct support recovery depends,  in a crucial way, on the
irrepresentability condition of  \cite{zhao}.

Let us stress that the results on support recovery offer limited insight into optimal hypothesis 
testing procedures.
Under the conditions that guarantee exact support recovery,
both type I and type II error rates tend to $0$ 
rapidly as $n,p,s_0\to\infty$, thus making it difficult to study the
trade-off between statistical significance and power.
Here we are interested in triples $n,p,s_0$ for which $\alpha$ and
$\beta$ stay bounded.
As discussed in the previous section, the regime
of interest (for standard Gaussian designs) is $c_1s_0\log(p/s_0)\le
n\le c_2s_0\log(p)$. At the lower end the number of observations $n$ is so small
that essentially nothing can be inferred about $\supp(\mtx{\theta}_0)$ using optimally tuned Lasso estimator, 
and therefore a nontrivial
 power $1-\beta >\alpha$ cannot be achieved. At the upper end, the number of samples
is sufficient enough to recover $\supp(\mtx{\theta}_0)$ with high probability, leading to arbitrary small errors $\alpha, \beta$

Let us finally mention that resampling methods provide an alternative  path to assess statistical
significance. A general framework to implement  this idea is provided by the stability
selection method of \cite{MeinshausenBuhlmannStability}.
However, specializing the approach and analysis of
\cite{MeinshausenBuhlmannStability} to the present 
context does not provide guarantees superior to
\cite{ZhangZhangSignificance,BuhlmannSignificance},
that are more directly comparable to the present work.

\subsection{Notations}
We provide a brief summary of the notations used throughout the paper.
We denote by $[p] = \{1,\cdots,p\}$ the set of first $p$ integers.
For a subset $\mathcal{J}\subseteq [p]$, we let
$|\mathcal{J}|$ denote its cardinality.
Bold upper (resp. lower) case letters denote
matrices (resp. vectors), and the same letter in normal typeface represents its coefficients, e.g. $a_j$ denotes
the $j$th entry of $\mtx{a}$.
For an $n\times p$ matrix $\mtx{M}$ and set of indices $I\subseteq [n], J\subseteq [p]$, we let $\mtx{M}_J$ denote the
$n \times |J|$ submatrix containing just
the columns in $J$ and use $\mtx{M}_{I,J}$ to denote the $|I| \times |J|$ submatrix formed by rows in $I$ and columns in $J$.
Likewise, for a vector $\mtx{\theta}\in \reals^p$,
$\mtx{\theta}_S$ is the restriction of $\mtx{\theta}$ to indices in $S$.
We denote the rows of the design matrix $\bX$ by
$\mtx{x}_1,\cdots,\mtx{x}_n\in\reals^p$. We also denote its columns by
$\mtx{\tx}_1,\cdots,\mtx{\tx}_p\in\reals^n$.
 The support of a vector $\mtx{\theta}\in \reals^p$ is denoted by $\supp(\mtx{\theta})$, i.e., $\supp(\mtx{\theta}) = \{i\in[p], \theta_i \neq 0\}$.
 We use $\id$  to denote the identity matrix in any dimension, and 
$\id_{d\times d}$ whenever is useful to specify the dimension $d$. 

Throughout, $\phi(x) = e^{-x^2/2}/\sqrt{2 \pi}$ is the Gaussian density and $\Phi(x) \equiv\int_{-\infty}^x
\phi(u) \de u$ is the Gaussian distribution. For two functions $f(n)$ and $g(n)$,
with $g(n)\ge 0$, the notation $f(n) = \Omega(g(n))$
means that $f$ is bounded below by $g$ asymptotically, namely, there exists constant $C>0$ and integer $n_0>0$,
such that $f(n) \ge C g(n)$ for $n > n_0$. Further, $f(n) = O(g(n))$ means that $f$ is bounded
above  by $g$ asymptotically, namely, for some constants $C< \infty$ and integer $n_0>0$,
$ f(n) \le C|g(n)|$ for all $n > n_0$. Finally  $f(n) = \Theta(g(n))$
if both $f(n) = \Omega(g(n))$ and $f(n) = O(g(n))$.

%
%====================================================
%
\section{Minimax formulation}\label{sec:minimax}

In this section we define the hypothesis testing problem, and
introduce a minimax criterion for evaluating hypothesis testing
procedures.
In subsection \ref{subsec:UpperBound} we state our upper bound on the
minimax power and, in subsection \ref{subsec:ProofUpperBound}, we
outline the prof argument, that is based on a reduction to binary
hypothesis testing.

\subsection{Tests with guaranteed power}\label{subsec:guaranteed}

We consider the minimax criterion to measure the quality of a testing
procedure. In order to define it formally, we first need to establish some notations.
 
A testing procedure for the family of hypotheses $H_{0,i}$,
cf. Eq.~(\ref{eq:Hypothesis}), is given by a family of measurable functions 
\begin{eqnarray}
\begin{split}
T_{i}:&&\reals^n\times\reals^{n\times p}\to \{0,1\}\, .\\
&& (\by,\bX) \mapsto T_{i,\bX}(\by)\, .
\end{split}
\end{eqnarray}
Here $T_{i,\bX}(\by) =1$ has the interpretation  that hypothesis
$H_{0,i}$ is rejected when the observation is $\by\in\reals^n$ and the
design matrix is $\bX$. We will hereafter drop the subscript $\bX$
whenever clear from the context.

As mentioned above, we will measure the quality of a test $T$ in terms of
its significance level $\alpha$ (probability of type I errors) and power
$1- \beta$ ($\beta$ is the probability of type II errors). 
A type I error (false rejection of the
null) leads one to conclude that a relationship between the
response vector $\by$ and a column of the design matrix $\bX$ exists
when in reality  it does not. 
On the other hand, a type II error (the failure to reject a false null
hypothesis) leads one to miss an existing  relationship.

Adopting a minimax point of view, we require that these metrics are achieved uniformly over $s_0$-sparse
vectors.
Formally, for $\mu>0$, we let
\begin{eqnarray}
\alpha_{i}(T) &\equiv& \sup \Big\{\prob_{\mtx{\theta}}\big(T_{i,\bX}(\by) = 1\big)\,
:\;\;  \mtx{\theta}\in\reals^p,\; \|\mtx{\theta}\|_0\le s_0,\; \theta_i=0\Big\}\, ,\\
\beta_{i}(T;\mu) & \equiv & \sup\Big\{\prob_{\mtx{\theta}}\big(T_{i,\bX}(\by) = 0\big)\,
:\;\;  \mtx{\theta}\in\reals^p,\; \|\mtx{\theta}\|_0\le s_0, \; |\theta_i|\ge \mu\Big\}\, .
\end{eqnarray}
In words, for any $s_0$-sparse vector with $\theta_{i}=0$, the
probability of false alarm is upper bounded by $\alpha_i(T)$. On the other
hand, if $\mtx{\theta}$ is $s_0$-sparse with $|\theta_i|\ge \mu$, the
probability of misdetection is upper bounded by
$\beta_i(T;\mu)$. 
Note that $\prob_{\mtx{\theta}}(\cdot)$ is the induced probability
distribution on $(\by,\bX)$ for random  design $\bX$ and  noise
realization $w$, given the fixed parameter
vector $\mtx{\theta}$.
Throughout we will accept randomized testing procedures
as well\footnote{Formally, this corresponds to assuming $T_i(\by) = T_i(\by;U)$
with $U$ uniform in $[0,1]$ and independent of the other random variables.}.
\begin{definition}\label{def:minimaxpower}
The \emph{minimax power} for testing hypothesis $H_{0,i}$ against the
alternative $|\theta_i|\ge \mu$ is given by the function 
$1- \bopt_{i}(\,\cdot\,;\mu): [0,1]\to [0,1]$ where, for $\alpha\in
[0,1]$
\begin{eqnarray}
1- \bopt_{i}(\alpha;\mu) \equiv \sup_{T}\Big\{1- \beta_{i}(T;\mu) :\;\;
\alpha_{i}(T)\le \alpha\Big\}\, . \label{eq:MinimaxPower}
\end{eqnarray}
\end{definition}
Note that for standard Gaussian designs (and more generally for
designs with exchangeable columns), $\alpha_i(T)$, $\beta_i(T;\mu)$ do
not depend on the index $i\in [p]$. We shall therefore omit the subscript
$i$ in this case. 

The following are straightforward yet useful properties.
\begin{remark}
The optimal power $\alpha\mapsto 1-\bopt_{i}(\alpha;\mu)$
is non-decreasing. Further, by using a test such that $T_{i,\bX}(\by) =1$ with probability
$\alpha$ independently of $\by$, $\bX$, we conclude that $ 1- \bopt_{i}(\alpha;\mu)\ge \alpha$.
\end{remark}
\begin{proof}
To prove the first property, notice that, for any $\alpha\le
\alpha'$ we have $1-\beta_i(\alpha;\mu)\le 1-\beta_i(\alpha';\mu)$.
Indeed $1-\beta_i(\alpha';\mu)$ is obtained by taking the supremum
in Eq.~(\ref{eq:MinimaxPower}) over a family of tests that includes
those over which the supremum is taken for $1-\beta_i(\alpha;\mu)$.

Next, a completely randomized test outputs $T_{i,\bX}(\by) =1$
with probability $\alpha$ independently of $\bX,\by$.  We then have
$\prob_{\mtx{\theta}}\big(T_{i,\bX}(\by) = 0\big) = 1-\alpha$ for any
$\mtx{\theta}$, whence 
$\beta_{i}(T;\mu) = 1-\alpha$. Since
this test offers --by definition-- the prescribed control on type I
errors,  we have, by Eq.~(\ref{eq:MinimaxPower}), $1-
\bopt_{i}(\alpha;\mu) \ge 1-\beta_{i}(T;\mu)= \alpha$.
\end{proof}

\subsection{Upper bound on the minimax power}\label{subsec:UpperBound}

Our upper bound on the minimax power is stated in terms of the
function $G:[0,1]\times \reals_+\to [0,1]$, $(\alpha,u)\mapsto
G(\alpha,u)$, defined as follows.
\begin{align}\label{eqn:G}
G(\alpha,u)\equiv 2-\Phi\Big(\Phi^{-1}(1-\frac{\alpha}{2})+u\Big) -
\Phi\Big(\Phi^{-1}(1-\frac{\alpha}{2})-u\Big)\, .
\end{align}
%
%In Fig.~\ref{fig:AlphaBeta1d}, the values of $G(\alpha,u)$ are plotted versus $\alpha$
%for several values of $u$.
 It is easy to check that, for any $\alpha>0$, $u\mapsto
G(\alpha,u)$ is  continuous and monotone increasing.
For $u$ fixed $\alpha\mapsto G(\alpha,u)$ is continuous and monotone
increasing. 
Finally $G(\alpha,0)
= \alpha$ and $\lim_{u\to\infty}G(\alpha,u)=1$.

%
%\begin{figure}
%%\includegraphics*[viewport = -10 140 600 640, width =
%%4in]
%\centering
%\includegraphics*[viewport = -10 160 600 640, width = 4.3in]{figs/AlphaBeta1d.pdf}
%\put(-180,10){{\small{\sf Type I error }($\alpha$)}}
%\put(-285,105){\rotatebox{90}{$G(\alpha,u)$}}
%\caption{Function $G(\alpha,u)$ versus $\alpha$ for several values of $u$.}\label{fig:AlphaBeta1d}
%\end{figure}
%%
We then have the following upper bound on the optimal power of random
Gaussian designs.  (We refer to Section
\ref{sec:ProofGeneralUpperBound} for the proof.)
\begin{thm}\label{thm:GeneralUpperBound}
For $i\in [p]$, let $1- \bopt_{i}(\alpha;\mu)$ be the minimax power of a
Gaussian random design $\bX$ with covariance matrix $\mtx{\Sigma}\in\reals^{p\times
  p}$, as per Definition
\ref{def:minimaxpower}. For $S\subseteq [p]\setminus \{i\}$, define 
${\Sigma}_{i|S}
\equiv{\Sigma}_{ii}-\mtx{\Sigma}_{i,S}\mtx{\Sigma}_{S,S}^{-1}\mtx{\Sigma}_{S,i}\in\reals$.  
Then, for any $\ell\in\reals$ and $|S| < s_0$,
\begin{align}
1 - \bopt_{i}(\alpha;\mu) & \le 
G\Big(\alpha,\frac{\mu}{\sigma_{\rm eff}(\ell)}\Big) 
+F_{n-s_0+1}(n-s_0+\ell)\,,\\
\sigma_{\rm eff}(\ell) & \equiv
\frac{\sigma}{\sqrt{{\Sigma}_{i|S}(n-s_0+\ell)}}\, ,
\end{align}
where $F_k(x) = \prob(Z_k \ge x)$, and $Z_k$ is a chi-squared
random variable with $k$ degrees of freedom.
\end{thm}
In other words, the statistical power is upper bounded by the one of
testing the mean of a scalar Gaussian random variable, with
effective noise variance $\sigma_{\rm eff}^2\approx
\sigma^2/[{\Sigma}_{i|S}(n-s_0)]$. (Note indeed that by concentration of
a chi-squared random variable around their mean, $\ell$ can be taken
small as compared to $n-s_0$.)  

The next corollary specializes the above result to the case of
standard Gaussian designs.  (The proof is immediate and hence we omit it.)
\begin{coro}\label{coro:UBStandard}
For $i\in [p]$, let $1- \bopt_{i}(\alpha;\mu)$ be the minimax power of
a standard Gaussian design $\bX$ with covariance matrix $\mtx{\Sigma} =
\id_{p\times p}$, cf. Definition
\ref{def:minimaxpower}. 
Then, for any $\xi\in[ 0,(3/2)\sqrt{n-s_{0}+1}]$
we have
\begin{align}
1 - \bopt_{i}(\alpha;\mu)  \le 
G\Big(\alpha,\frac{\mu(\sqrt{n-s_{0}+1}+\xi)}{\sigma}\Big) + e^{-\xi^2/8}\,.
\end{align}
\end{coro}

It is instructive to look at the last result from a slightly different
point of view. Given $\alpha\in(0,1)$ and $1-\beta\in (\alpha,1)$, how
big does the entry $\mu$ need to be so that $1 - \bopt_{i}(\alpha;\mu)
\ge 1-\beta$? 
It follows 
from Corollary \ref{coro:UBStandard} that to achieve a pair $(\alpha,\beta)$ as above
we require $\mu \ge \mu_{\UB} = c\sigma/\sqrt{n}$ for some
$c=c(\alpha,\beta)$. 

Previous work \cite{ZhangZhangSignificance,BuhlmannSignificance} 
requires $\mu \ge c\,\max\{\sigma s_0 \log p/\, n, \sigma/\sqrt{n}\}$ to achieve the same goal
although for deterministic designs $\bX$ (see
Appendix~\ref{app:Alternative}).  This motivates the central question
of the present paper (already stated in the introduction): Can
hypothesis testing be
performed in the ideal regime $\mu\ge c\sigma/\sqrt{n}$?

As further clarified in the next section and in Section
\ref{sec:ProofBinary},  Theorem~\ref{thm:GeneralUpperBound} by an
oracle-based argument. Namely,  we upper bound the power of any
hypothesis testing method, by the power of an oracle that  knows, for
each coordinates 
$j \in [p]\setminus j$, whether $\theta_{0,j} \in
\supp(\mtx{\theta}_0)$ or not. In other words the 
procedure has access to $\supp(\mtx{\theta}_0) \backslash\{i\}$.
At first sight, this oracle appears exceedingly powerful,  and hence
the bound might be loose.
Surprisingly, the bound turns out to be tight, at least in an
asymptotic sense, as demonstrated in  Section \ref{sec:std}.

Let us finally mention that a bound similar to the present one was
announced independently --and from a different viewpoint-- in \cite{Report-Semi}.

\subsection{Proof outline}\label{subsec:ProofUpperBound}

The proof of Theorem \ref{thm:GeneralUpperBound} is based on a simple
reduction to the  binary hypothesis testing
problem. We first introduce the binary testing problem,
in which the vector of coefficients $\mtx{\theta}$ is chosen randomly
according to one of two distributions.
\begin{definition}\label{definition:RandomizedTheta}
Let $Q_0$ be a probability distribution on $\reals^p$ supported on
$\cR_0\equiv\{\mtx{\theta}\in\reals^p:\; \|\mtx{\theta}\|_0\le s_0, \;
\theta_i=0\}$, and $Q_1$ a probability distribution supported on $\cR_1\equiv\{\mtx{\theta}\in\reals^p:\; \|\mtx{\theta}\|_0\le s_0, \;
|\theta_i|\ge \mu\}$. For fixed design matrix $\bX\in\reals^{n\times p}$, and $z\in\{0,1\}$, let $\prob_{Q,z,\bX}$
denote the  law of $y$ as per model (\ref{eq:NoisyModel}) when
$\mtx{\theta}_0$ is chosen randomly with  $\mtx{\theta}_0\sim Q_z$. 

We denote by $1-\bbin_{i,\bX}(\,\cdot\,;Q)$ the optimal
power  for the binary hypothesis testing problem $\mtx{\theta}_0\sim Q_0$
versus $\mtx{\theta}_0\sim Q_1$, namely:
\begin{align}
\bbin_{i,\bX}(\alpha_{\bX};Q) \equiv
\inf_{T}\Big\{\,\prob_{Q,1,\bX}(T_{i,\bX}(\by) =0):\; \;
\;\;\prob_{Q,0,\bX}(T_{i,\bX}(\by) =1)\le \alpha_{\bX}\,\Big\}\, .
\end{align}
\end{definition}

The reduction is stated in the next lemma.
\begin{lemma}\label{lemma:Binary}
Let $Q_0$, $Q_1$ be any two probability measures supported,
respectively, on $\cR_0$ and $\cR_1$ as per Definition \ref{definition:RandomizedTheta}.
Then, the minimax power for testing hypothesis $H_{0,i}$ under the
random design model, cf. Definition \ref{def:minimaxpower}, is bounded
as
\begin{align}
\bopt_{i}(\alpha;\mu)\ge
\inf\Big\{\E\bbin_{i,\bX}(\alpha_{\bX};Q):\;\;\;\; \E(\alpha_{\bX})\le
\alpha \;\Big\}\, .
\end{align}
Here expectation is taken with respect to the law of $\bX$ and the $\inf$
is over all measurable functions $\bX\mapsto \alpha_{\bX}$.
\end{lemma}
For the proof we refer to Section \ref{sec:ProofBinary}.

The binary hypothesis testing problem is characterized  in the next lemma by
reducing it to a simple  regression problem.
For $S\subseteq [p]$, we denote by $\proj_S$ the orthogonal projector
on the linear space spanned by the columns $\{\mtx{\tx}_i\}_{i\in S}$. We also
let $\projp_S=\id_{n\times n}-\proj_S$ be the projector on the orthogonal subspace. 
\begin{lemma}\label{lemma:PerXUpperBound}
Let $\bX\in\reals^{n\times p}$ and $i\in [p]$. For $S\subset
[p]\setminus \{i\}$,
$\alpha\in[0,1]$, define
\begin{eqnarray}
1 - \bora_{i,\bX}(\alpha;S,\mu) & = &  G\Big(\alpha,\frac{\mu\|\projp_{S}\mtx{\tx}_i\|_2}{\sigma}\Big)\,.\quad\quad  \label{eq:OrableUB2}
\end{eqnarray}
If $|S|< s_0$ then for any $\xi>0$ there exists distributions $Q_0$, $Q_1$ as per
Definition \ref{definition:RandomizedTheta}, depending on $i$, $S$,
$\mu$ but not on $\bX$, such that  $\bbin_{i,\bX}(\alpha;Q)\ge \bora_{i,\bX}(\alpha;S,\mu)-\xi$.
\end{lemma}
The proof of this Lemma is presented in Section \ref{sec:ProofPerXUpperBound}.

The proof of Theorem \ref{thm:GeneralUpperBound} follows from 
Lemmas \ref{lemma:Binary} and \ref{lemma:PerXUpperBound}, cf. Section \ref{sec:ProofGeneralUpperBound}.
%===============================================================
\section{Hypothesis testing for standard Gaussian designs }
\label{sec:std}

In this section we describe our hypothesis testing procedure (that we
refer to as \sdl) in the
case of standard Gaussian designs, see subsection  \ref{subsec:std-proc}.
In subsection \ref{subsec:std-analysis}, we develop asymptotic bounds on
the probability of type I and type II errors. The test is shown to
nearly achieve  the ideal tradeoff between significance level $\alpha$
and  power $1-\beta$, using the upper bound stated in the previous section. 

Our results are based on a characterization of the high-dimensional behavior of the Lasso estimator, 
developed in~\cite{BayatiMontanariLASSO}. For the reader's
convenience, and to provide further context, we recall this result in
subsection \ref{subsec:std-Gaussian}. Finally, subsection
\ref{sec:experiment-std}
discusses some numerical experiments.

\subsection{Hypothesis testing procedure}\label{subsec:std-proc}

Our $\sdl$ procedure for standard Gaussian designs is described in Table~\ref{tbl:procedure1}.
     \begin{table}[!t]
     \caption{\sdl\, for testing $H_{0,i}$ under standard Gaussian design model.}\label{tbl:procedure1}
     \noindent\rule{\textwidth}{1pt}\\
     \sdl\,: Testing hypothesis $H_{0,i}$ under standard Gaussian design model.\\
     \noindent\rule{\textwidth}{1pt}\\
     \textbf{Input:} regularization parameter $\lambda$, significance level $\alpha$\\
     \textbf{Output:} p-values $P_i$, test statistics $T_{i,\bX}(\by)$\\
     1:\quad Let
     \vspace{-0.5cm}
     \begin{align*}
     \mtx{\htheta}(\lambda) = \text{argmin}_{\mtx{\theta}\in\reals^p} \,\,\Big\{ \frac{1}{2n} \|\by-\bX\mtx{\theta}\|^2 + \lambda \|\mtx{\theta}\|_1\Big\}.
     \end{align*}
      2:\quad Let 
      \vspace{-0.3cm}
      \begin{align}\label{eq:Step2-std} 
      \scale = \left(1-\frac{1}{n} \|\mtx{\htheta}(\lambda)\|_0\right)^{-1},\;\;\;\;\;
      \tau = \frac{1}{\Phi^{-1}(0.75)} \frac{\scale}{\sqrt{n}} |(\mtx{y} - \bX\mtx{\htheta}(\lambda))|_{(n/2)},
      \end{align}
      where for $\mtx{v} \in \reals^K$, $|\mtx{v}|_{\ell}$ is the $\ell$-th largest entry in the vector $(|v_1|,\cdots,|v_n|)$. \\
       3:\quad Let
       \vspace{-0cm}
       \begin{align*}
        \mtx{\htheta}^u =\mtx{\htheta}(\lambda) + \frac{\scale}{n} \bX^\sT(\mtx{y} - \bX \mtx{\htheta}(\lambda)).
        \end{align*}
       4:\quad Assign the p-values $P_i$ for the test $H_{0,i}$ as follows.
       \vspace{-0.2cm}
       \begin{align*}
        P_i = 2\Big(1- \Phi\big(\big|\frac{\htheta^u_i}{\tau}\big|\big)\Big).
        \end{align*}
        5:\quad The decision rule is then based on the p-values:
	\begin{align*}
	T_{i,\bX}(\by)=
	\begin{cases}
	1&\text{if  $P_i \le \alpha$} \quad \quad \; (\text{reject the null hypothesis }H_{0,i}),\\
	0& \text{otherwise}  \quad \quad  \text{(accept the null hypothesis)}.
	\end{cases}
	\end{align*}
      \noindent\rule{\textwidth}{1pt}
      \end{table}

The key is the construction of the \emph{unbiased estimator}
$\mtx{\htheta}^u$ in step 3. The asymptotic analysis developed in 
\cite{BayatiMontanariLASSO} and in the next section establishes that 
\emph{$\mtx{\htheta}^u$ is an asymptotically unbiased estimator of $\mtx{\theta}_0$}, and the empirical distribution of 
$\{\htheta^u_i - \theta_{0,i}\}_{i=1}^p$ is asymptotically normal with
variance $\tau^2$.
Further, the variance $\tau^2$ can be consistently estimated using the
residual vector $\mtx{r}$. 
These results establish that (in a sense that will be made precise next)
the regression model (\ref{eq:NoisyModel}) is asymptotically equivalent to a
simpler sequence model
\begin{align}\label{eq:noise-unbiased}
\mtx{\htheta}^u = \mtx{\theta}_0 + \,{\sf noise}
\end{align}
with {\sf noise} having zero mean. 
In particular, under the null hypothesis $H_{0,i}$, $\htheta^u_i$ is
asymptotically gaussian with mean $0$ and variance $\tau^2$. This
motivates rejecting the null if $|\htheta^u_i|\ge \tau \Phi^{-1}(1-\alpha/2)$.

The construction of
$\mtx{\htheta}^u$ has an appealing geometric interpretation. Notice that $\mtx{\htheta}$ is
necessarily biased towards small $\ell_1$ norm. 
 The minimizer in Eq.~(\ref{eqn:Lasso_cost}) must
satisfy $(1/n)\,\bX^\sT(\mtx{y}-\bX\mtx{\htheta}) = \lambda \mtx{g} $, with $\mtx{g}$ a subgradient of $\ell_1$ norm at $\mtx{\htheta}$.
 Hence, we can rewrite
$\mtx{\htheta}^u = \mtx{\htheta} + \scale \lambda \mtx{g}$. The bias is
eliminated by modifying the estimator in the direction of increasing
$\ell_1$ norm. See Fig.~\ref{fig:geometry-2} for an illustration.

\begin{figure}[t]
\centering
\includegraphics*[width =4in]{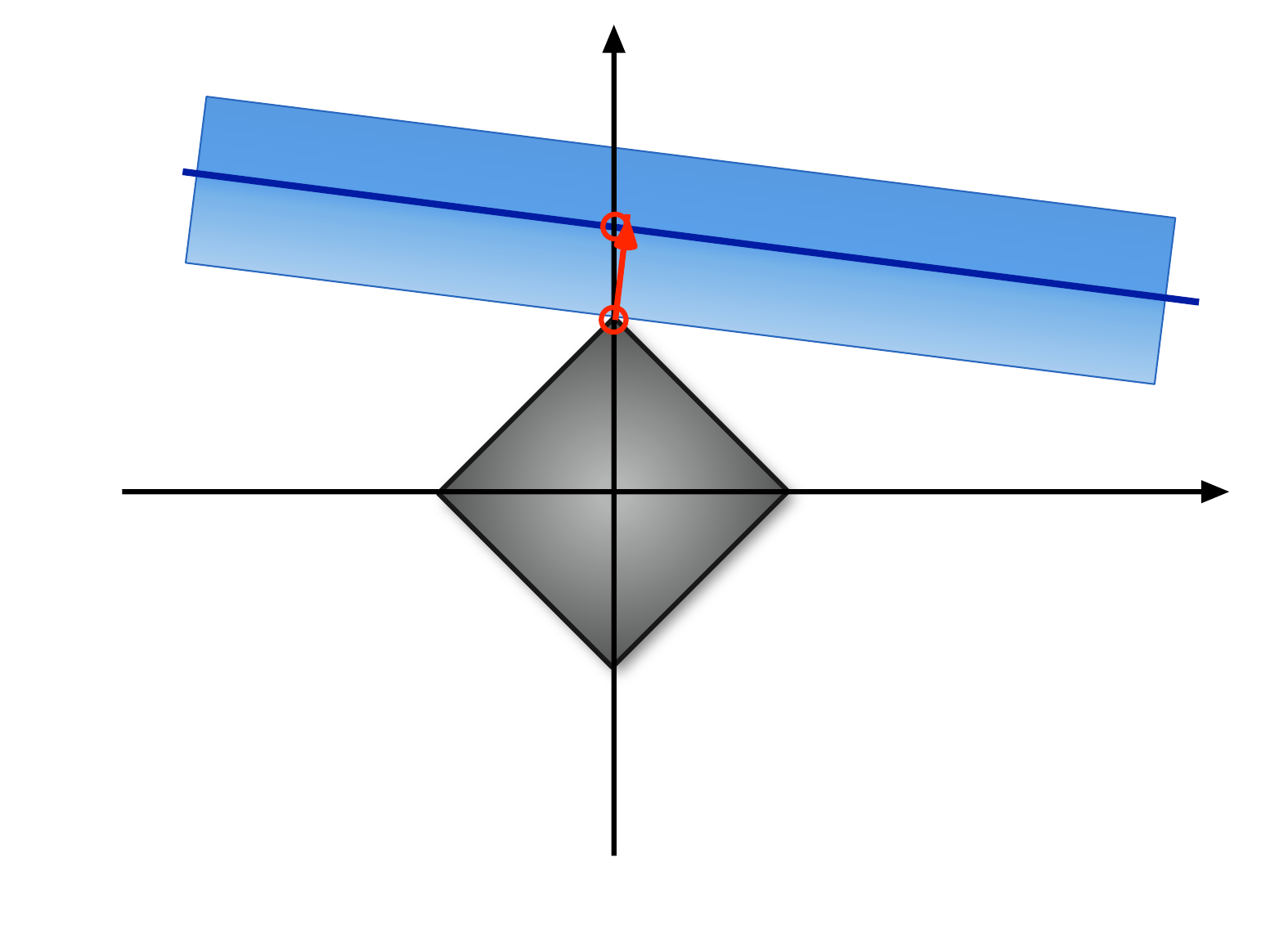}
\put(-310,180){$\frac{1}{2n}\|\by-\bX\mtx{\theta}\|^2$}
\put(-260,120){Ball of $\ell_1$ norm}
\put(-145,140){$\mtx{\htheta}$}
\put(-140,160){$\frac{1}{n}\,\bX^\sT(\by-\bX\mtx{\theta})$}
\caption{Geometric interpretation for construction of $\mtx{\htheta}^u$. The bias in $\mtx{\htheta}$ is eliminated by modifying the estimator in the direction of increasing its $\ell_1$ norm}\label{fig:geometry-2}
\end{figure}
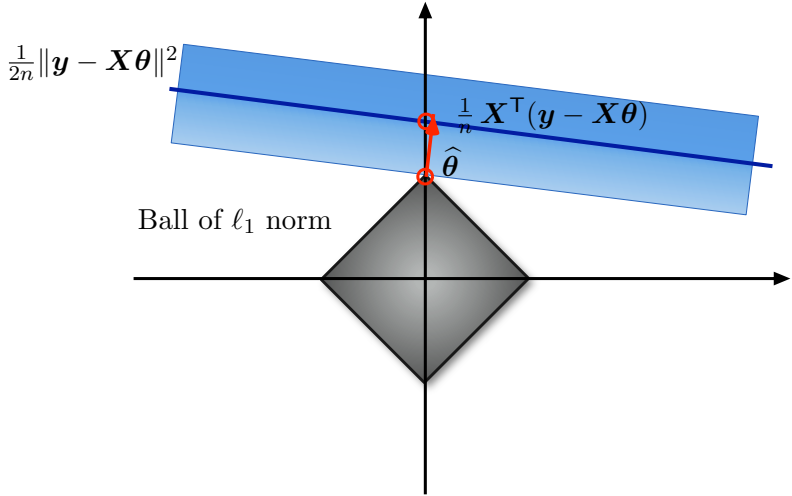

\subsection{Asymptotic analysis}\label{subsec:std-analysis}

For given dimension $p$, an \emph{instance} of the standard Gaussian design model
is defined by the tuple
$ (\mtx{\theta}_0, n,\sigma)$, where $\mtx{\theta}_0 \in \reals^p$, $n \in \naturals$,
$\sigma\in\reals_+$. We consider sequences of instances indexed by the problem dimension 
$\{ (\mtx{\theta}_0(p), n(p),\sigma(p))\}_{p \in \naturals}$. 
\begin{definition}\label{def:converging}
The sequence of instances $\{ (\mtx{\theta}_0(p), n(p),\sigma(p))\}_{p \in \naturals}$ indexed by $p$
is said to be a \emph{converging sequence} if $ n(p) /p \to \delta \in
(0,\infty)$, $\sigma(p)^2/n\to \sigma_0^2$, and the empirical 
distribution of the entries $\mtx{\theta}_0(p)$ converges weakly to a probability measure $p_{\Theta_0}$
on $\reals$ with bounded second moment.
Further $p^{-1} \sum_{i\in [p]} \theta_{0,i}(p)^2 \to \E_{p_{\Theta_0}}\{\Theta_0^2\}$.
\end{definition}
Note that this definition assumes the coefficients $\theta_{0,i}$
are of order one, while the noise is scaled as $\sigma(p)^2 =
\Theta(n)$. Equivalently, we could have assumed $\theta_{0,i} =
\Theta(1/\sqrt{n})$ and $\sigma^2(p)=\Theta(1)$:  the two
settings only differ by a scaling of $\by$. We favor the first scaling
as it simplifies somewhat the notation in the following.

As before, we will measure the quality of the proposed test in terms of
its {significance level} ({size}) $\alpha$ and {power}
$1- \beta$. Recall that $\alpha$ and $\beta$ respectively indicate the type I error (false positive) and type II error (false negative) rates.
The following theorem establishes that the $P_i$'s are indeed valid
p-values, i.e., allow to control type I errors. Throughout
$S_0(p) = \{i \in [p] : \theta_{0,i}(p) \neq 0\}$ is the support of $\mtx{\theta}_0(p)$.
\begin{thm}\label{thm:type_I}
Let $\{(\mtx{\theta}_0(p), n(p),\sigma(p))\}_{p\in\naturals}$ be a converging sequence of
instances of the standard Gaussian design model. Assume $\lim_{p\to \infty} |S_0(p)|/p =\P(\Theta_0 \neq 0)$. 
Then, for $i\in S_0^c(p)$, we have
\begin{eqnarray}
\lim_{p\to \infty} \prob_{\mtx{\theta}_0(p)}(T_{i,\bX}(\by)=1) = \alpha\,.
\end{eqnarray}
\end{thm}

A more general form of Theorem~\ref{thm:type_I} (cf. Theorem~\ref{thm:type_I_nstd})
is proved in Section~\ref{sec:Proofs}. We indeed prove the stronger claim that the following holds 
true almost surely
\begin{eqnarray}\label{eqn:strong-iid-alpha}
\lim_{p\to \infty} \frac{1}{|S^c_0(p)|} \sum_{i \in S^c_0(p)} T_{i,\bX}(\by) =\alpha\,.
\end{eqnarray}
The result of Theorem~\ref{thm:type_I} follows then by taking the expectation of both sides of Eq.~\eqref{eqn:strong-iid-alpha}
and using bounded
convergence theorem and exchangeability of the columns of $\bX$.

Our next theorem proves a lower bound for the power of the proposed test.
In order to obtain a non-trivial result, we need to make suitable assumption on
the parameter vectors $\mtx{\theta}_0 = \mtx{\theta}_0(p)$. In particular, we need
to assume that the non-zero entries of $\mtx{\theta}_0$ are lower bounded in
magnitude. If this were not the case, it would
be impossible to distinguish arbitrarily small parameters
$\theta_{0,i}$ from $\theta_{0,i}= 0$. 
(In Appendix~\ref{app:lambda}, we also provide an explicit formula for 
the regularization parameter $\lambda=\lambda(p_{\Theta_0}, \sigma, \ve, \delta)$ that
achieves this power.)
\begin{thm}\label{thm:power}
There exists a (deterministic) choice of
$\lambda=\lambda(\sigma,\eps)$ such that the following happens.

Let $\{(\mtx{\theta}_0(p), n(p),\sigma(p))\}_{p\in \naturals}$ be a converging sequence of instances
under the standard Gaussian design model. Assume that
$|S_0(p)|\le \ve p$, and for all $i \in S_0(p)$, $|\theta_{0,i}(p)| \ge \mu$ with $\mu= \mu_0 \sigma(p)/\sqrt{n(p)}$. 
for $i\in S_0(p)$, we have
\begin{eqnarray}
\lim_{p\to \infty} \prob_{\mtx{\theta}_0(p)}(T_{i,\bX}(\by) = 1) \ge G\Big(\alpha,\frac{\mu_0}{\tau_*}\Big)\,,
\end{eqnarray}
where $\tau_* = \tau_*(\sigma_0,\eps,\delta)$ is defined as follows
\begin{align}
\tau_*^2 &= \begin{cases}
 \dfrac{1}{1- M(\ve)/\delta} , & \text{if } \delta > M(\ve),\\
 \infty, & \text{if } \delta \le M(\ve).  
 \end{cases}\label{eqn:tau_*}
\end{align}
%
%Finally, $M(\ve)$ is the minimax risk of the soft thresholding
%denoiser, with following parametric expression in terms of the parameter $\xi \in (0,\infty)$:
 Here, $M(\ve)$ is given by the following parametric expression in terms of the parameter $\kappa \in (0,\infty)$:
 \begin{align}\label{eqn:M_ve}
 \ve = \frac{2(\phi(\kappa) - \kappa \Phi(-\kappa))}{\kappa + 2(\phi(\kappa) - \kappa \Phi(-\kappa))},\;\;\;\;\;\;\;\;\;
 M(\ve) = \frac{2\phi(\kappa)}{\kappa + 2(\phi(\kappa) - \kappa \Phi(-\kappa))}.  
 \end{align}
\end{thm}

Theorem \ref{thm:power} is proved in Section \ref{sec:Proofs}. We indeed prove the stronger claim that the following holds true almost surely:
\begin{eqnarray}\label{eqn:strong-iid-beta}
\lim_{p\to \infty} \frac{1}{|S_0(p)|} \sum_{i\in S_0(p)} T_{i,\bX}(\by) \ge G\Big(\alpha,\frac{\mu_0}{\tau_*}\Big)\,.
\end{eqnarray}
The result of Theorem~\ref{thm:power} follows then by taking the expectation of both sides of Eq.~\eqref{eqn:strong-iid-beta}
and using exchangeability of the columns of $\bX$.

Again, it is convenient to rephrase Theorem~\ref{thm:power} in terms
of the minimum value of $\mu$ for which we can achieve statistical
power $1-\beta\in (\alpha,1)$ at significance level $\alpha$. 
It is known that $M(\eps) = 2\eps\log(1/\eps)\,(1+O(\eps))$ \cite{DMM-NSPT-11}.
Hence, for $n\ge 2\, s_0\log(p/s_0)\, (1+O(s_0/p))$, we have $\tau_*^2= O(1)$. Since
$\lim_{u\to\infty}G(\alpha,u) = 1$, any pre-assigned statistical power
can be achieved by taking $\mu\ge C(\eps,\delta)\sigma/\sqrt{n}$ which
matches the fundamental limit established  in the previous section. 

Let us finally comment on the choice of the regularization parameter $\lambda$. 
Theorem \ref{thm:type_I} holds irrespective of $\lambda$, as long as it
is kept fixed in the asymptotic limit. In other words, control of type
I errors is fairly insensitive to the regularization parameters. On
the other hand, to achieve optimal minimax power, it is necessary to
tune $\lambda$ to the correct value.  
The tuned value of $\lambda = \lambda(p_{\Theta_0}, \sigma, \ve, \delta)$ for the standard Gaussian sequence model
is provided in Appendix~\ref{app:tau}.
%As explained in {\bf [AM:
%  WHERE????]}, 
%the theory of \cite{BayatiMontanariLASSO,DMM-NSPT-11} implies the
%choice $\lambda(\eps,\sigma) = \kappa(\eps)\sigma$ with $\kappa(\eps) =
%\sqrt{2\log(1/\eps)}(1+O(\eps))$ the minimax threshold value for
%estimation using soft thresholding in the Gaussian sequence model.
%The function $\kappa(\eps)$ is 
%explicitly given in {\bf [AM:
%WHERE????]}.
Further, the factor $\sigma$ (and hence the need to
estimate the noise level) can be omitted if --instead of the Lasso--
we use the scaled Lasso \cite{SZ-scaledLasso}.
In subsection \ref{sec:experiment-std}, we discuss another way of
choosing $\lambda$ that also avoid estimating the noise level.

\subsection{Gaussian limit}\label{subsec:std-Gaussian}

Theorems \ref{thm:type_I} and \ref{thm:power} are based on an
asymptotic distributional characterization of the Lasso estimator
developed in~\cite{BayatiMontanariLASSO}.  We restate it here for the
reader's convenience.
\begin{thm}[\cite{BayatiMontanariLASSO}]\label{pro:lasso-limit} 
Let $\{(\mtx{\theta}_0(p), n(p),\sigma(p))\}_{p \in \naturals}$ be a converging sequence of instances of
the standard Gaussian design model.
Denote by $\mtx{\htheta} = \mtx{\htheta}(\by,\bX,\lambda)$ the Lasso estimator given as per
Eq.~(\ref{eqn:Lasso_cost}) and define $\mtx{\htheta}^u\in\reals^p$,
$\mtx{r}\in\reals^n$ by letting
\begin{align}
\mtx{\htheta}^u\equiv \mtx{\htheta} + \frac{\scale}{n} \, \bX^\sT (\mtx{y}-\bX \mtx{\htheta})\,,\;\;\;\;\;\;
\mtx{r} \equiv \frac{\scale}{\sqrt{n}} (\mtx{y}-\bX \mtx{\htheta})\,, 
\end{align}
with $\scale = (1- \|\mtx{\htheta}\|_0/n)^{-1}$. 

Then, with probability one, the empirical distribution of $\{(\theta_{0,i},\htheta^u_i)\}_{i=1}^p$
converges weakly to the probability distribution of
$(\Theta_0, \Theta_0 + \tau_0 Z)$, for some $\tau_0 \in \reals$, where $Z \sim
\normal(0,1)$, and $\Theta_0 \sim p_{\Theta_0}$ is
independent of $Z$. Furthermore,
with probability one, the empirical distribution of $\{r_i\}_{i=1}^p$ converges weakly to $\normal(0,\tau^2_0)$.

Finally $\tau_0\in \reals$ is defined by the unique solution of
Eqs.~(\ref{eq:Tau1}) and (\ref{eq:Tau2}) in Appendix \ref{app:tau}.
\end{thm}
In particular, this result implies that the empirical distribution of
$\{\htheta^u_i - \theta_{0,i}\}_{i=1}^p$ is asymptotically normal with
variance $\tau_0^2$. This naturally motivates  the use of $|\htheta^u_i|/\tau_0$
as a test statistics for hypothesis $H_{0,i}:\, \theta_{0,i}=0$.

The definitions of $\scale$ and $\tau$ in step 2 are also motivated by 
Theorem~\ref{pro:lasso-limit}.    In particular, $\scale (\mtx{y}-\bX \mtx{\htheta})/\sqrt{n}$
is asymptotically normal with variance $\tau_0^2$. This is used in
step 2, where $\tau$ is just the robust median absolute deviation
(MAD) estimator
(we choose this estimator since it is more resilient to outliers than the sample variance~\cite{HuberBook}).

%
 %%%%%%%%%%%%%%%%%%%
 %
 \subsection{Numerical experiments}\label{sec:experiment-std}

%{\bf [A: Why not move to the same table the comparison with LPDE?]} 

As an illustration, we generated synthetic data from the linear model~\eqref{eqn:regression} with $\mtx{w} \sim \normal(0, \id_{p\times p})$ and the following configurations.

\emph{Design matrix:} For pairs of values $(n,p) =
\{(300,1000), (600,1000), (600, 2000)\}$, the design matrix is
generated from a realization of $n$ i.i.d. rows 
$\mtx{x}_i\sim \normal(0,\id_{p \times p})$. 
%Here $\mtx{\Sigma} \in \reals^{p \times p}$  is a
%circulant matrix with $\mtx{\Sigma}_{ii} =1$, $\mtx{\Sigma}_{jk} = 0.1$ for $j \neq k,\,|j-k| \le 5$ and zero everywhere else. (The difference between indices is understood modulo $p$.)

\emph{Regression parameters:}  We consider active sets $S_0$ with $|S_0| = s_0 \in \{10, 20, 25, 50, 100\}$, chosen uniformly at random from the index set $\{1, \cdots, p\}$. We also consider two different strengths of active parameters $\theta_{0,i} = \mu$, for $i \in S_0$, with $\mu \in \{0.1, 0.15\}$.

We examine the performance of \sdl\, (cf. Table~\ref{tbl:procedure1}) 
at significance levels $\alpha = 0.025, 0.05$. The experiments are done using {\sf glmnet}-package in
R that fits the entire Lasso path for linear regression
models. Let $\ve = s_0/p$ and $\delta = n/p$. We do not
assume $\ve$ is known, but rather estimate it as  $\bar{\ve} = 0.25\,
\delta/\log(2/\delta)$. The value of $\bar{\ve}$ 
is half the maximum sparsity level $\ve$ for the given $\delta$ such
that the Lasso estimator can correctly recover the parameter vector if
the measurements were
noiseless~\cite{DMM09,BayatiMontanariLASSO}. Provided it makes sense
to use Lasso at all, $\bar{\ve}$ is thus a reasonable ballpark estimate.

\begin{figure}[!h]
\centering
\includegraphics*[width =
3.6in]{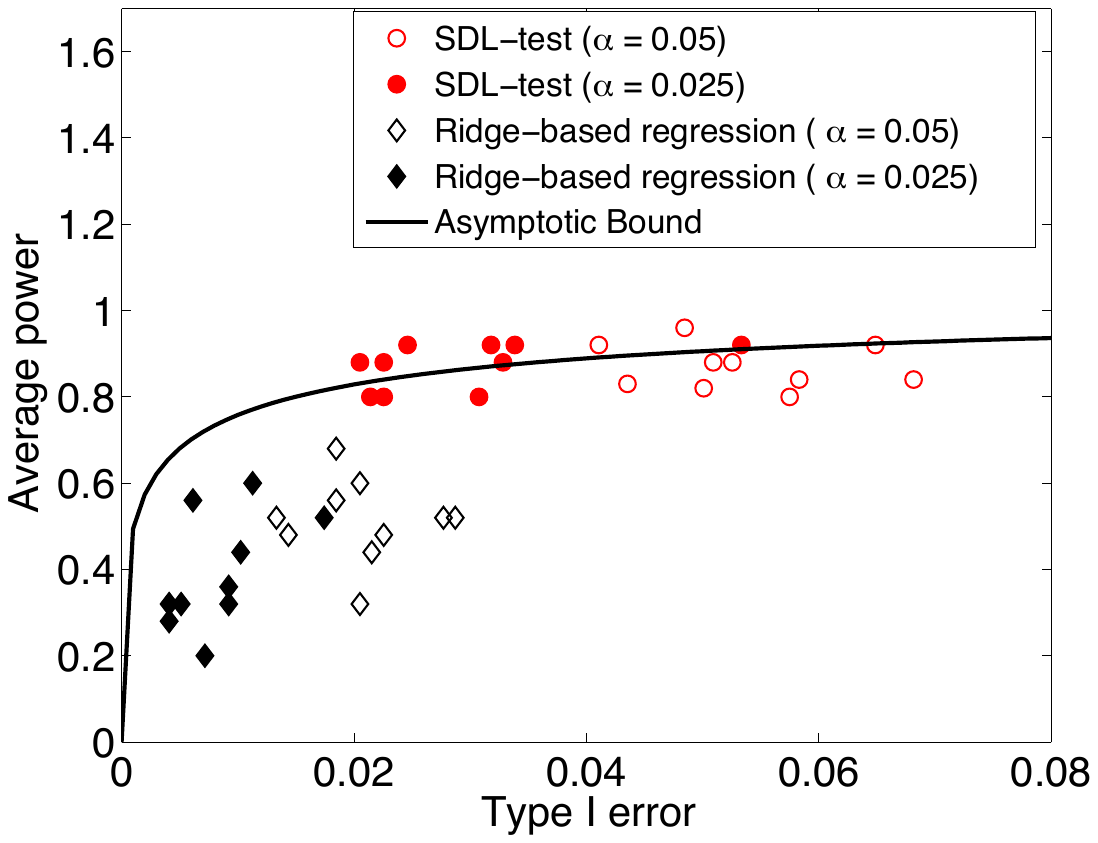}
\caption{Comparison between $\sdl$ (Table~\ref{tbl:procedure1}), ridge-based regression~\cite{BuhlmannSignificance} and the asymptotic bound for \sdl\, (established in Theorem~\ref{thm:power}). 
Here, $p = 1000, n = 600, s_0 = 25, \mu = 0.15$. }
\label{fig:plot_iid}
\end{figure}

\begin{table*}[!h]
\begin{center}
{\small
\begin{tabular}{c|cccc|}
%\cline{2-13}
%& \multicolumn{12}{|c|}{\bf Y} \\  \cline{2-13}

{\bf Method} & Type I err & Type I err & Avg. power & Avg. power \\ 
&(mean) & (std.) & (mean) & (std)
\\
\cline{1-5} \cline{2-5}
\multicolumn{1}{c|}{\bf $\sdl$  $(1000, 600, 100, 0.1)$} &   0.05422 & 0.01069 & 0.44900 & 0.06951
    \\ 
 \multicolumn{1}{c|}{\bf Ridge-based regression $(1000, 600, 100, 0.1)$} & \color{red}0.01089 & \color{red} 0.00358& \color{red} 0.13600& \color{red} 0.02951
 \\
  \multicolumn{1}{c|}{\bf $\ldpe\,$ $(1000, 600, 100, 0.1)$} & \color{darkblue}0.02012 & \color{darkblue} 0.00417& \color{darkblue} 0.29503& \color{darkblue} 0.03248
 \\
 \multicolumn{1}{c|}{\bf Asymptotic Bound $(1000, 600, 100, 0.1)$} & \color{olivegreen}0.05 & \color{olivegreen}NA & \color{olivegreen}0.37692 & \color{olivegreen}NA
 \\
 \hline
\multicolumn{1}{c|}{\bf $\sdl$ $(1000, 600, 50, 0.1)$} &  0.04832 & 0.00681& 0.52000 & 0.06928
 \\ 
\multicolumn{1}{c|}{\bf Ridge-based regression $(1000, 600, 50, 0.1)$} &  \color{red} 0.01989& \color{red} 0.00533& \color{red} 0.17400  & \color{red} 0.06670
\\ 
\multicolumn{1}{c|}{\bf $\ldpe\,$ $(1000, 600, 50, 0.1)$} &  \color{darkblue} 0.02211& \color{darkblue} 0.01031& \color{darkblue} 0.20300  & \color{darkblue} 0.08630
\\ 
 \multicolumn{1}{c|}{\bf Asymptotic Bound $(1000, 600, 50, 0.1)$} & \color{olivegreen}0.05 & \color{olivegreen}NA & \color{olivegreen}0.51177 & \color{olivegreen}NA
\\
\hline 
\multicolumn{1}{c|}{\bf $\sdl$ $(1000, 600, 25, 0.1)$} &  0.05662 & 0.01502 & 0.56400 & 0.11384
 \\ 
\multicolumn{1}{c|}{\bf Ridge-based regression $(1000, 600, 25, 0.1)$} &  \color{red}  0.02431& \color{red}  0.00536& \color{red} 0.25600& \color{red}0.06586
\\
\multicolumn{1}{c|}{\bf $\ldpe\,$ $(1000, 600, 25, 0.1)$} &  \color{darkblue}  0.02305& \color{darkblue}  0.00862& \color{darkblue} 0.27900& \color{darkblue}0.07230
\\
\multicolumn{1}{c|}{\bf Asymptotic Bound $(1000, 600, 25, 0.1)$} & \color{olivegreen}0.05 & \color{olivegreen}NA & \color{olivegreen}0.58822 & \color{olivegreen}NA
\\ 
\end{tabular}
}
\end{center}
\caption{Comparison between $\sdl$ (Table~\ref{tbl:procedure1}), ridge-based regression~\cite{BuhlmannSignificance}, $\ldpe$~\cite{ZhangZhangSignificance} and the asymptotic bound for \sdl\, (established in Theorem~\ref{thm:power}) on the setup described in Section~\ref{sec:experiment-std}. The significance level is $\alpha = 0.05$. The means and the standard deviations are obtained by testing over $10$ realizations of the corresponding configuration. Here a quadruple such as $(1000, 600, 50, 0.1)$ denotes the values of $p = 1000$, $n = 600$, $s_0 = 50$, $\mu = 0.1$. \vspace{-0.2cm}}\label{tbl:iid_small_alpha05}
\end{table*}

The regularization parameter $\lambda$ is chosen as to satisfy 
\begin{align}
\lambda \scale = \kappa_* \tau
\end{align}
where $\tau$ and $\scale$ are determined in step 2 of the
procedure. Here $\kappa_* = \kappa_*(\bar{\ve})$ is the  minimax threshold value for
estimation using soft thresholding in the Gaussian sequence model, see
\cite{DMM-NSPT-11} and Remark~\ref{rem:lambda}.
Note that $\tau$ and $\scale$ in the equation above depend implicitly
upon $\lambda$. Since {\tt glmnet} returns the entire Lasso path, the
value of $\lambda$  solving the above equation can be computed by  the
bisection method. 

As mentioned above,  the control of type I error is fairly robust for a wide range of values of
$\lambda$. However, the above is an educated guess based on the
analysis of~\cite{DMM09,BayatiMontanariLASSO}. We also tried the values of
$\lambda$ proposed  for instance in
\cite{BuhlmannVanDeGeer,BuhlmannSignificance} on the basis of oracle
inequalities. 

Figure \ref{fig:plot_iid} shows the results of $\sdl$ and
the method of \cite{BuhlmannSignificance} for parameter values $ p =
1000, n = 600, s_0 = 25, \mu = 0.15$, and significance levels $\alpha
\in \{0.025, 0.05\}$. Each point in the plot corresponds to one
realization of this configuration (there are a total of $10$
realizations). We also depict the theoretical curve
$(\alpha,G(\alpha,\mu_0/\tau_*))$, predicted by
Theorem~\ref{thm:power}. The empirical results are in good agreement
with the asymptotic prediction.

\begin{figure}[!h]
\centering
\includegraphics*[width =3.3in]{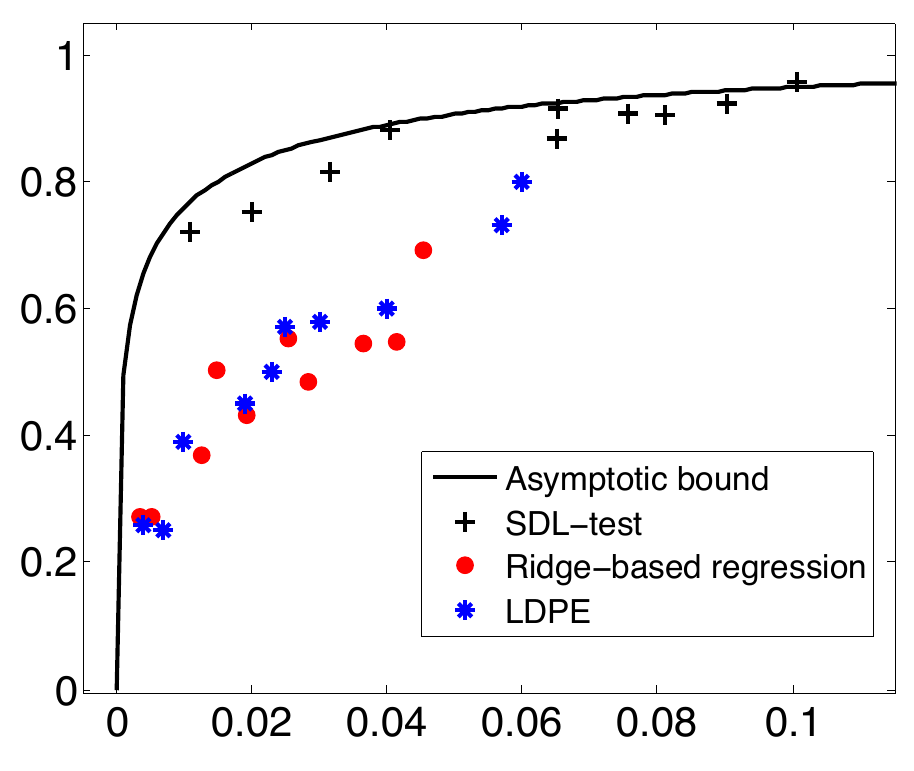}
\caption{Comparison between $\sdl\,$, ridge-based regression~\cite{BuhlmannSignificance}, and $\ldpe$~\cite{ZhangZhangSignificance}. The curve corresponds to the asymptotic bound for $\sdl\,$ as established in Theorem~\ref{thm:power}. For the same values of type I error achieved by methods, $\sdl\,$ results in a higher statistical power. Here, $p = 1000, n = 600, s_0 = 25, \mu = 0.15$. }
\label{fig:different_alpha}
\end{figure}
We compare $\sdl\,$ with the ridge-based regression method~\cite{BuhlmannSignificance}
and the low dimensional projection estimator ($\ldpe\,$) 
~\cite{ZhangZhangSignificance}. Table~\ref{tbl:iid_small_alpha05}
summarizes the results for a few
configurations $(p, n, s_0, \mu)$, and $\alpha = 0.05$. Simulation
results for a larger number of configurations and $\alpha = 0.05,
0.025$ are reported in Tables~\ref{tbl:iid_alpha05}
and~\ref{tbl:iid_alpha025} in Appendix~\ref{app:Simulation
  Results}. 

As demonstrated by these results, $\ldpe\,$\cite{ZhangZhangSignificance} and the ridge-based regression
\cite{BuhlmannSignificance} are both overly conservative. Namely, they achieve
smaller type I error than the prescribed level $\alpha$ and this comes
at the cost of a smaller statistical power than our testing
procedure. This is to be expected since the approach of
\cite{BuhlmannSignificance} and \cite{ZhangZhangSignificance} cover a broader class of  design matrices
$\bX$, and are not tailored to random designs.

Note that being overly conservative is a drawback, when this comes at
the expense of statistical power. The data analysts should be able to
decide the level of statistical significance $\alpha$, and obtain
optimal statistical power at that level.

The reader might wonder whether the loss in statistical power of methods in
\cite{BuhlmannSignificance} and \cite{ZhangZhangSignificance} is entirely due to the fact that 
these methods achieve a smaller number of false positives than requested.
In Fig.~\ref{fig:different_alpha}, we run \sdl\,, ridge-based regression~\cite{BuhlmannSignificance}, and $\ldpe\,$ for $\alpha \in \{0.01,0.02, \cdots, 0.1\}$ and for $10$ realizations of the
problem per each value of $\alpha$. We plot the average type I error and the average power of each method
versus $\alpha$. As we see 
\emph{even for the same empirical fraction of type I errors}, $\sdl\,$ results in a higher statistical power.

%
%%%%%%%%%%%%%%%%%%%%%%%%%%%%%%%%%%%
%
\section{Hypothesis testing for nonstandard Gaussian
  designs}\label{sec:nstd}

In this section, we generalize our testing procedure to nonstandard Gaussian
design models where the rows of the design matrix $\bX$ are drawn independently
from distribution $\normal(0, \mtx{\Sigma})$. 

We first describe the generalized \sdl\, procedure in subsection \ref{sec:gen-proc}
under the assumption that $\mtx{\Sigma}$ is known. 
In subsection \ref{sec:gen-analysis}, we show that this generalization
can be justified from a certain generalization of the Gaussian limit theorem
\ref{pro:lasso-limit} to nonstandard Gaussian designs. 

Establishing such a generalization of Theorem \ref{pro:lasso-limit}
appears extremely challenging. We nevertheless show that such a limit theorem
 follows from the replica method of statistical physics
in section \ref{sec:gen-gauss}. We also show that a version of this
limit theorem is relatively straightforward in the regime $s_0=o(n/(\log
p)^2)$. 

Finally, in Section \ref{sec:gen-covariance} we discuss a procedure for
estimating  the covariance $\mtx{\Sigma}$ (cf. \Subroutine\, in
Table~\ref{tbl:procedure2}).
Appendix \ref{app:CovarianceFree} proposes an alternative implementation that
does not estimate $\mtx{\Sigma}$ but instead bounds the effect of unknown $\mtx{\Sigma}$.

\subsection{Hypothesis testing procedure}\label{sec:gen-proc}

     \begin{table}[!t]
     \caption{\sdl\, for testing hypothesis $H_{0,i}$ under nonstandard Gaussian design model}\label{tbl:SDL}
     \noindent\rule{\textwidth}{1pt}\\
     \sdl: Testing hypothesis $H_{0,i}$ under nonstandard Gaussian design model.\\
     \noindent\rule{\textwidth}{1pt}\\
     \textbf{Input:} regularization parameter $\lambda$, significance level $\alpha$, covariance matrix $\mtx{\Sigma}$ \\
     \textbf{Output:} p-values $P_i$, test statistics $T_{i,\bX}(\by)$\\
     %1:\quad Estimate the covariance matrix, $\hSigma = \Subroutine(\bX)$.\\
     1:\quad Let
     \vspace{-0.5cm}
     \begin{align*}
     \mtx{\htheta}(\lambda) = \text{argmin}_{\mtx{\theta}\in\reals^p} \,\,\Big\{ \frac{1}{2n} \|\by-\bX\mtx{\theta}\|^2 + \lambda \|\mtx{\theta}\|_1\Big\}.
     \end{align*}
      2:\quad Let 
      \vspace{-0.3cm}
      \begin{align} \label{eq:Step2}
       \scale = \left(1-\frac{1}{n} \|\mtx{\htheta}(\lambda)\|_0\right)^{-1},\;\;\;\;\;
      \tau = \frac{1}{\Phi^{-1}(0.75)} \frac{\scale}{\sqrt{n}} |(\mtx{y} - \bX\mtx{\htheta}(\lambda))|_{(n/2)},
      \end{align}
      where for $\mtx{v} \in \reals^K$, $|\mtx{v}|_{\ell}$ is the $\ell$-th largest entry in the vector $(|v_1|,\cdots,|v_n|)$. \\
       3:\quad Let
       \vspace{-0cm}
       \begin{align*}
        \mtx{\htheta}^u = \mtx{\htheta}(\lambda) + \frac{\scale}{n} \mtx{\Sigma}^{-1} \bX^\sT(\mtx{y} - \bX \mtx{\htheta}(\lambda)).
        \end{align*}
       4:\quad Assign the p-values $P_i$ for the test $H_{0,i}$ as follows.
       \vspace{-0.2cm}
       \begin{align*}
        P_i = 2\bigg(1- \Phi \Big(\Big|\frac{\htheta^u_i}{\tau [(\mtx{\Sigma}^{-1})_{ii}]^{1/2}}\Big|\Big)\bigg).
        \end{align*}
        5:\quad The decision rule is then based on the p-values:
	\begin{align*}
	T_{i,\bX}(\by)=
	\begin{cases}
	1&\text{if  $P_i \le \alpha$} \quad \quad \; (\text{reject the null hypothesis }H_{0,i}),\\
	0& \text{otherwise}  \quad \quad  \text{(accept the null hypothesis)}.
	\end{cases}
	\end{align*}
      \noindent\rule{\textwidth}{1pt}
      \end{table}

The hypothesis testing procedure \sdl for general Gaussian designs is
defined in Table~\ref{tbl:SDL}. 

The basic intuition of this generalization is that  $(\htheta^u_i-\htheta_{0,i})/(\tau
[(\mtx{\Sigma}^{-1})_{ii}]^{1/2})$ is expected to be asymptotically $\normal(0,1)$, whence 
the definition of (two-sided) p-values $P_i$ follows as in step 4.
Parameters $\scale$ and $\tau$ in step 2 are defined in the same manner to 
the standard Gaussian designs.

\subsection{Asymptotic analysis}\label{sec:gen-analysis}

For given dimension $p$, an \emph{instance} of the nonstandard Gaussian design model
is defined by the tuple $(\mtx{\Sigma}, \mtx{\theta}_0, n,\sigma)$, where $\mtx{\Sigma} \in \reals^{p \times p}$,
$\mtx{\Sigma} \succ 0$, $\mtx{\theta}_0 \in \reals^p$, $n \in \naturals$,
$\sigma\in\reals_+$. We are interested in the asymptotic
properties of sequences of instances indexed by the problem dimension 
$\{(\mtx{\Sigma}(p), \mtx{\theta}_0(p), n(p),\sigma(p))\}_{p \in \naturals}$. Motivated by Proposition~\ref{pro:lasso-limit},
we define a property of a sequence of instances that we refer to as \emph{standard distributional limit}.
\begin{definition}\label{def:SDL}
A sequence of instances $\{(\mtx{\Sigma}(p), \mtx{\theta}_0(p), n(p),\sigma(p))\}_{p \in \naturals}$
 indexed by $p$ is said to have an (almost sure) \emph{standard distributional limit}
if there exist $\tau,\scale\in\reals$ (with $\scale$ potentially
 random, and both $\tau$, $\scale$ potentially depending on $p$), such that the following holds. 
Denote by $\mtx{\htheta} = \mtx{\htheta}(y,\bX,\lambda)$ the Lasso estimator given as per
Eq.~(\ref{eqn:Lasso_cost}) and define $\mtx{\htheta}^u\in\reals^p$,
$\mtx{r}\in\reals^n$ by letting
\begin{align}
\mtx{\htheta}^u\equiv \mtx{\htheta} + \frac{\scale}{n} \, \mtx{\Sigma}^{-1} \bX^\sT (\mtx{y}-\bX \mtx{\htheta})\,,\;\;\;\;\;
\mtx{r} \equiv \frac{\scale}{\sqrt{n}} (\mtx{y}-\bX \mtx{\htheta}). \label{eq:ThetauDef}
\end{align}
Let $v_i = (\theta_{0,i}, (\htheta^u_i-\theta_{0,i})/\tau, (\mtx{\Sigma}^{-1})_{ii})$, for $1\le i \le p$, and $\nu^{(p)}$ be the empirical distribution of $\{v_i\}_{i=1}^{p}$ defined as
\begin{eqnarray}
\nu^{(p)} = \frac{1}{p} \sum_{i=1}^p \delta_{v_i}\,,  \label{eq:EmpDef}
\end{eqnarray}
where $\delta_{v_i}$ denotes the Dirac delta function centered at $v_i$. Then, with probability one, the empirical distribution $\nu^{(p)}$
converges weakly to a probability measure $\nu$ on $\reals^3$ as $p
\to \infty$. Here, $\nu$ is the probability distribution of
$(\Theta_0,  \Upsilon^{1/2} Z, \Upsilon)$, where $Z \sim
\normal(0,1)$, and $\Theta_0$ and $\Upsilon$ are random variables
independent of $Z$. Furthermore,
with probability one, the empirical distribution of $\{r_i/\tau\}_{i=1}^n$ converges weakly to $\normal(0,1)$.
\end{definition}
\begin{remark}
This definition is non-empty
by Theorem \ref{pro:lasso-limit}. Indeed, if 
$\{ (\mtx{\theta}_0(p), n(p),\sigma(p))\}_{p \in \naturals}$ is
converging as per Definition \ref{def:converging}, and $a>0$ is a
constant, then Theorem \ref{pro:lasso-limit} states that 
$\{(\mtx{\Sigma}(p)= a\,\id_{p\times p}, \mtx{\theta}_0(p),
n(p),\sigma(p))\}_{p \in \naturals}$ has a standard distributional limit.
\end{remark}

Proving the standard distributional limit for
general sequences
$\{(\mtx{\Sigma}(p),\mtx{\theta}_0(p),$ $n(p),\sigma(p))\}_{p\in\naturals}$ is an
outstanding mathematical challenge. In sections \ref{sec:gen-gauss}
and \ref{sec:discussion} we discuss
both rigorous and non-rigorous evidence towards its validity. The
numerical simulations in Sections \ref{sec:experiment-nstd} and
\ref{sec:discussion} further support the usefulness of this notion.

We will next show that the \sdl\, procedure is appropriate for any 
random design model for which the standard distributional limit holds.
Our first theorem is a generalization of Theorem~\ref{thm:type_I} to
this setting.
\begin{thm}\label{thm:type_I_nstd}
Let $\{(\mtx{\Sigma}(p),\mtx{\theta}_0(p), n(p),\sigma(p))\}_{p\in\naturals}$ be a sequence of
instances for which a standard distributional limit holds. 
Further assume $\lim_{p\to \infty} |S_0(p)|/p =\P(\Theta_0 \neq 0)$. 
Then, 
\begin{eqnarray}
\lim_{p\to \infty} \frac{1}{|S^c_0(p)|} \sum_{i \in S^c_0(p)} \prob_{\mtx{\theta}_0(p)}(T_{i,\bX}(\by)=1) = \alpha\,.
\end{eqnarray}
\end{thm}

 The proof of Theorem~\ref{thm:type_I_nstd} is deferred to Section~\ref{sec:Proofs}. In the proof, we show 
 the stronger result that the following holds true almost surely  
 \begin{eqnarray}\label{eqn:strong-nstd-alpha}
 \lim_{p\to \infty} \frac{1}{|S^c_0(p)|} \sum_{i \in S^c_0(p)} T_{i,\bX}(\by) = \alpha\,.
 \end{eqnarray}
 The result of Theorem~\ref{thm:type_I_nstd} follows then by taking the expectation of both sides
 of Eq.~\eqref{eqn:strong-nstd-alpha} and using bounded
 convergence theorem.
 
The following theorem characterizes the power of \sdl\, for
 general $\mtx{\Sigma}$, and under the
assumption that a standard distributional limit holds .
\begin{thm}\label{thm:power2}
 Let $\{(\mtx{\Sigma}(p), \mtx{\theta}_0(p), n(p),\sigma(p))\}_{p\in
  \naturals}$ be a sequence of instances with standard distributional
limit. Assume (without loss of generality) $\sigma(p) = \sqrt{n(p)}$,
and further
$|\theta_{0,i}(p)|/[(\mtx{\Sigma}^{-1})_{ii}]^{1/2} \ge \mu_0$ for all
$i \in S_0(p)$,
and $\lim_{o\to \infty}|S_0(p)|/p = \prob(\Theta_0\neq 0)\in (0,1)$.
Then,
\begin{eqnarray}
\lim_{p\to \infty} \frac{1}{|S_0(p)|} \sum_{i\in S_0(p)}
\prob_{\mtx{\theta}_0(p)}(T_{i,\bX}(\by) =1) \ge G\Big(\alpha,\frac{\mu_0}{\tau}\Big)\,.
\end{eqnarray}
\end{thm}
Theorem \ref{thm:power2} is proved in Section \ref{sec:Proofs}. We indeed prove the stronger result that
the following holds true almost surely
\begin{eqnarray}
\lim_{p\to \infty} \frac{1}{|S_0(p)|} \sum_{i \in S_0(p)}T_{i,\bX}(\by) \ge G\Big(\alpha,\frac{\mu_0}{\tau}\Big)\,.
\end{eqnarray}

We also notice that in contrast to Theorem \ref{thm:power}, where $\tau_*$ has an explicit formula
that leads to an analytical lower bound for the power (for a suitable choice of $\lambda$),
in Theorem~\ref{thm:power2}, $\tau$ depends upon $\lambda$ implicitly and can be estimated from the data as in step 3
of \sdl\, procedure. The result of Theorem~\ref{thm:power2} holds for any value of $\lambda$.

\subsection{Gaussian limit for $n\gg s_0(\log p)^2$}

In the following theorem we show that if sample size $n$
asymptotically dominates $s_0 (\log p)^2$, then 
 the standard distributional limit can be established rigorously.  
\begin{thm}\label{thm:Replica-Rigor}
Assume the sequence of instances $\{\mtx{\Sigma}(p), \mtx{\theta}_0(p), n(p), \sigma(p)\}_{p\in \naturals}$
such that, as $p\to \infty$ (letting $s_0 = \|\mtx{\theta}_0(p)\|_0$): 
\begin{itemize}
\item[$(i)$] $n(p)\le p$, and $s_0 (\log p)^2/n(p) \to 0$;
\item[$(ii)$] $\sigma(p)^2/n(p) \to \sigma_0^2 >0$; 
\item[$(iii)$] There exist constants $c_{\rm min},c_{\rm max}>0$ such that 
the eigenvalues of $\mtx{\Sigma}$ lie in the interval $[c_{\rm
  min},c_{\rm max}]$: $c_{\rm min}\le\lambda_{\rm
  min}(\mtx{\Sigma})\le \lambda_{\rm max}(\mtx{\Sigma})\le c_{\rm
  max}$;
\item[$(iv)$] The empirical distribution of
  $\{(\mtx{\Sigma}^{-1})_{ii})\}_{1\le i\le p}$
  converges weakly to the probability distribution of the random
  variable $\Upsilon$;
\item[$(v)$] The regularization parameter is $\lambda = C_*\sigma
  \sqrt{(\log p)/n}$ for $C_*= C_*(c_{\rm min},c_{\rm max})$ a sufficiently large constant.
\end{itemize}
Then the sequence has a standard distributional limit with $\scale =
(1-\|\mtx{\htheta}(\lambda)\|_0/n)^{-1}$
and $\tau=\sigma_0$. Alternatively, $\tau$ can be taken to be a
solution of Eq.~(\ref{eq:ClaimedTau})
below.
\end{thm}

Theorem~\ref{thm:Replica-Rigor} is proved in
Section~\ref{proof:Replica-Rigor}. The proof uses techniques from our
conference paper \cite{JavanMon-OptSample}. 

Notice that this result does allow to control type I errors using
Theorem \ref{thm:type_I_nstd}, but does not allow to lower bound the
power, using Theorem \ref{thm:power2}, since $|S_0(p)|/p\to 0$. A lower bound on the power
under the same assumptions presented in this section can be found in \cite{JavanMon-OptSample}. 
In the present paper we focus instead on the case $|S_0(p)|/p$ bounded
away from $0$.

\subsection{Gaussian limit via the replica heuristics for smaller
  sample size $n$}\label{sec:gen-gauss}

As mentioned above, the standard distributional limit follows from Theorem
\ref{pro:lasso-limit} for $\Sigma = \id_{p\times p}$. Even in this
simple case, the proof is rather challenging
\cite{BayatiMontanariLASSO}. Partial generalization to non-gaussian
designs and other convex problems appeared recently in \cite{BM-Universality} and \cite{oymak2013squared},
each requiring over 50 pages of proofs.

On the other hand, these and similar asymptotic  results can be derived
heuristically using the `replica method' from statistical physics. 
In Appendix~\ref{app:replica}, we use this approach to derive the
following claim\footnote{In Appendix \ref{app:replica} we derive
  indeed a more general result, where the $\ell_1$ regularization is
  replaced by an arbitrary separable penalty.}.
\begin{rclaim}\label{claim:Replica}
Assume the sequence of instances
$ \{(\mtx{\Sigma}(p),\mtx{\theta}_0(p),n(p),\sigma(p))\}_{p\in\naturals}$ to be
such that, as $p\to\infty$: $(i)$ $n(p)/p\to\delta>0$; $(ii)$
$\sigma(p)^2/n(p)\to \sigma_0^2>0$; $(iii)$ The sequence of functions
\vspace{-0.2cm}
\begin{eqnarray}
\Energy^{(p)}(a,b)\equiv \frac{1}{p}\,\E\min_{\mtx{\theta}\in\reals^p}
\Big\{\frac{b}{2}\|\mtx{\theta}-\mtx{\theta}_0-\sqrt{a}\mtx{\Sigma}^{-1/2}\mtx{z}\|_{\mtx{\Sigma}}^2 +
\lambda\|\mtx{\theta}\|_1\Big\}\, ,\label{eq:ReplicaAssumption}
\end{eqnarray}  
 with $\|\mtx{v}\|_{\mtx{\Sigma}}^2\equiv\<\mtx{v},\mtx{\Sigma} \mtx{v}\>$ and $\mtx{z}\sim \normal(0,\id_{p\times p})$
admits a differentiable limit $\Energy(a,b)$ on $\reals_+\times\reals_+$, with
$\nabla \Energy^{(p)}(a,b)\to \nabla \Energy(a,b) $. Then
the sequence has a standard distributional limit. 
Further let
\begin{align}
\eta_b(\by) &\equiv \underset{\mtx{\theta}\in\reals^p}{\arg\min}\Big\{\frac{b}{2}\|\mtx{\theta}-\by\|_{\mtx{\Sigma}}^2 +
\lambda\|\mtx{\theta}\|_1\Big\}\, ,\label{eq:ProximalOperator}\\
\Ena(a,b) 
&\equiv\lim_{p\to\infty}\frac{1}{p}\E\big\{\big\|\eta_{b}(\mtx{\theta}_0+\sqrt{a}\mtx{\Sigma}^{-1/2}\mtx{z})-\mtx{\theta}_0\big\|_{\mtx{\Sigma}}^2\big\}\,
,
\end{align}
where the the limit exists by the above assumptions on the convergence
of $\Energy^{(p)}(a,b)$. Then, the parameters 
$\tau$ and $\scale$ of the standard distributional limit are obtained
by setting $\scale = (1-\mtx{\htheta}/n)^{-1}$ and solving the
following with respect to $\tau^2$:
\begin{align}
\tau^2 = \sigma_0^2 +\frac{1}{\delta}\, \Ena(\tau^2,1/\scale)\,. \label{eq:ClaimedTau}
\end{align}
\end{rclaim}
In other words, the replica method indicates that the standard
distributional limit holds for a large class of non-diagonal
covariance
structures $\mtx{\Sigma}$. It is worth stressing that convergence assumption for the sequence
$\Energy^{(p)}(a,b)$ is quite mild, and is satisfied by a large family
of covariance matrices.  For instance,  it can be proved that it holds
for block-diagonal matrices $\mtx{\Sigma}$ as long as the blocks have
bounded length and the blocks empirical
distribution converges.

The replica method is  a non-rigorous but highly sophisticated 
calculation procedure that has proved successful in a number of very difficult problems
in probability theory and probabilistic combinatorics.
Attempts to make the replica method rigorous have been pursued over the last 30 years
by some world-leading mathematicians 
\cite{TalagrandVolI,panchenko2013sherrington,guerra2003broken,aizenman2003extended}.
This effort achieved spectacular successes, but so far does not provide
tools to prove the above replica claim. In particular, the rigorous work mainly
focuses on `i.i.d. randomness',
corresponding to the case covered by Theorem  \ref{pro:lasso-limit}.

Over the last ten years, the replica method has been used to derive a number
of fascinating results in information theory and communications theory, see e.g. 
\cite{TanakaCDMA,GuoVerdu,tanaka2005approximate,campo2011large,wu2012optimal}. 
More recently, several groups used it successfully in the analysis of
high-dimensional sparse regression under
standard Gaussian designs
\cite{RanganFletcherGoyal,KabashimaTanaka,BaronGuoShamai,wu2012optimal,TakedaKabashima,tulino2011support,KabashimaOrthogonal}.
The rigorous analysis of ours and other groups
\cite{MoT06,BayatiMontanariLASSO,BM-Universality,oymak2013squared}
 subsequently confirmed these heuristic calculations in several cases.

\vspace{0.3cm}

There is a fundamental reason that makes establishing the standard
distributional limit a challenging task. This requires in fact to
characterize the distribution of the estimator (\ref{eqn:Lasso_cost}) in a regime where
the standard deviation of $\htheta_i$ is of the same
order as its mean. Further, $\htheta_i$ does not converge to the true
value $\theta_{0,i}$, hence making perturbative arguments ineffective.

The analysis becomes easier for a larger number of samples. In Theorem
\ref{thm:Replica-Rigor} below we will show that (a suitable version
of) the standard distributional holds for $n$ asymptotically larger
than $s_0(\log p)^2$. This uses methods from our companion paper \cite{confidenceJM}.

\subsection{Covariance estimation}\label{sec:gen-covariance}

So far we assumed that the design covariance  $\mtx{\Sigma}$ is
known. 
This  setting is relevant for 
semi-supervised learning applications, where the data analyst has access 
to a large number $N\gg p$ of `unlabeled examples'. These
are i.i.d. feature vectors  $\mtx{u}_1$, $\mtx{u}_2$,\dots $\mtx{u}_{N}$
with $\mtx{u}_1\sim\normal(0,\mtx{\Sigma})$ distributed as
$\mtx{x}_1$, for which the response variable $y_i$ is not
available. In this case $\mtx{\Sigma}$   can be estimated accurately
by $N^{-1}\sum_{i=1}^n \mtx{u}_i\mtx{u}_i^{\sT}$.
We refer to \cite{chapelle2006semi} for further background on such applications.

In other applications, $\mtx{\Sigma}$ is unknown and no additional
data is  available. In this case we proceed as follows:
\begin{enumerate}
\item We estimate $\Sigma$ from the design matrix $\mtx{X}$ (equivalently, from the
feature vectors $\mtx{x}_1$, $\mtx{x}_2$, \dots $\mtx{x}_n$). We let
$\mtx{\hSigma}$ denote the resulting estimate.
\item We use $\mtx{\hSigma}$ instead of $\mtx{\Sigma}$ in step 3 of
  our hypothesis testing procedure.
\end{enumerate} 
The problem of estimating covariance matrices in high-dimensional setting 
has attracted considerable attention in the past.
Several estimation methods provide a consistent estimate $\mtx{\hSigma}$, under suitable
structural assumptions on $\mtx{\Sigma}$. 
For instance if $\mtx{\Sigma}^{-1}$ is sparse, one can apply the graphical model method of \cite{MeinshausenBuhlmann},
the regression approach of~\cite{Zhao-Cov}, or CLIME
estimator~\cite{CLIME}, to name a few. 

Since the covariance estimation problem is not the focus of our paper,
we will test the above approach using a very simple
covariance estimation method.
Namely, we assume that $\mtx{\Sigma}$ is sparse and estimate it
by thresholding the empirical covariance. A detailed description of
this estimator is given in Table~\ref{tbl:procedure2}. We refer to
\cite{bickel2008regularized}
for a theoretical analysis of this type of methods. 
Note that the Lasso is unlikely to perform well if the columns of
$\bX$ are highly correlated and hence the assumption of sparse
$\mtx{\Sigma}$ is very natural. On the other hand, we would like to
emphasize that this covariance thresholding estimation is only one
among many possible approaches.

As an additional contribution, in Appendix
\ref{app:CovarianceFree} we describe an alternative covariance-free
procedure that only uses bounds on $\mtx{\Sigma}$ where the bounds are estimated from
the data.

In our numerical experiments,
we use the estimated covariance returned by \Subroutine.
As shown in the next section, computed p-values appear to be fairly robust with respect to 
errors in the estimation of $\mtx{\Sigma}$.
It would be interesting to develop a rigorous analysis of
\sdl\, that
accounts for the covariance estimation error. 

\begin{center}
     \begin{table}[!t]
     \caption{\Subroutine\, for estimating covariance $\mtx{\Sigma}$}\label{tbl:procedure2}
     \noindent\rule{\textwidth}{1pt}\\
     \Subroutine: Estimating covariance matrix $\mtx{\Sigma}$\\
     \noindent\rule{\textwidth}{1pt}\\
     \textbf{Input:} Design matrix $\bX$\\
     \textbf{Output:} Estimate $\mtx{\hSigma}$\\
     1:\quad Let $\mtx{C} = (1/n) \bX^\sT \bX \in \reals^{p\times p}$.\\
     2:\quad Let $\sigma_1$ be the empirical variance of the entries in $S$ and let $\mathcal{A} = \{C_{ij}: |C_{ij}| \le 3 \sigma_1\}$.\\
     3:\quad Let $\sigma_2$ be the variance of entries in $\mathcal{A}$.\\
     4:\quad Construct $\mtx{\widehat{C}}$ as follows:
      \begin{eqnarray}
      \widehat{C}_{ij} = C_{ij}\, \ind (|C_{ij}| \ge 3\sigma_2).
      \end{eqnarray}      
      5:\quad Denote by $\zeta_1$ and $\zeta_2$ the smallest and the smallest 
      positive eigenvalues of $\mtx{\widehat{C}}$ respectively.\\
      6:\quad Set
      \begin{eqnarray}
      \mtx{\hSigma} = \mtx{\widehat{C}} + (\zeta_2-\zeta_1) \id\,.
      \end{eqnarray}
    \noindent\rule{\textwidth}{1pt}
      \end{table}
\end{center}  
%
%\vspace{-0.7cm}
%
%%%%%%%%%%%%%%%%%%%%
%%%%%%%%%%%%%%%%%%%%
\subsection{Numerical experiments}\label{sec:experiment-nstd}

%{\bf [A: Why not move to the same table the comparison with LPDE?]} 

In carrying out our numerical experiments for correlated Gaussian
designs,
we  consider the same setup as the one in
Section~\ref{sec:experiment-std}. The only difference is that 
the rows of the design matrix are independently $\mtx{x}_i\sim \normal(0,\mtx{\Sigma})$.
We choose $\mtx{\Sigma} \in \reals^{p \times p}$  to be a the
symmetric matrix with entries $\Sigma_{jk}$ are defined as follows for
$j\le k$
%{\bf [Are we using bold for the entries of a matrix? I think no.]}
%
\begin{equation}\label{eq:Sigma}
\Sigma_{jk} = \begin{cases}
1 & \text{if } k=j\,,\\
0.1 & \text{if } k \in \{j+1,\cdots,j+5\} \\
& \text{or } k\in\{j+p-5,\dots, j+p-1\},\\
0 & \text{for all other $j\le k$.}
\end{cases}
\end{equation}
Elements below the diagonal are given by the symmetry condition
$\Sigma_{kj} = \Sigma_{jk}$.
(Notice that this is a circulant matrix.)

In Fig.~\ref{fig:plot}, we compare \sdl\,with the ridge-based regression method proposed
in~\cite{BuhlmannSignificance}. While the type I errors 
of \sdl\, are in good match with the chosen significance level
$\alpha$, the method of \cite{BuhlmannSignificance}  is
conservative. 
As in the case of standard Gaussian designs, this results in
significantly smaller type I errors than $\alpha$ and smaller average
power in return.  
%{\bf [A: CAN WE PLOT THE UPPER BOUND FROM THEOREM
%  2.3 (NEGLECTING $F$)?  ALSO, THE FIGURE IS NOT CORRECTLY
%  NUMBERED. ALSO, CAN WE DO AN ANALOGOUS OF FIGURE 4 IN THIS CASE?]}
Also, in Fig.~\ref{fig:different_alpha_cov}, we run \sdl\,, ridge-based regression~\cite{BuhlmannSignificance}, and $\ldpe\,$~\cite{ZhangZhangSignificance} for $\alpha \in \{0.01,0.02, \cdots, 0.1\}$ and for $10$ realizations of the
problem per each value of $\alpha$. We plot the average type I error and the average power of each method
versus $\alpha$. As we see, similar to the case of standard  Gaussian designs, 
\emph{even for the same empirical fraction of type I errors}, $\sdl\,$ results in a higher statistical power.

Table~\ref{tbl:small_alpha05} summarizes the
performances of the these methods for a few configurations $(p, n, s_0,
\mu)$, and $\alpha = 0.05$. Simulation results for a larger number of
configurations and $\alpha = 0.05, 0.025$ are reported in 
Tables~\ref{tbl:illus_alpha05} and~\ref{tbl:illus_alpha025} in Appendix~\ref{app:Simulation Results}.

Let $\mtx{z} = (z_i)_{i=1}^p$ denote the vector with entries $z_i\equiv
\htheta^u_i/(\tau[(\mtx{\Sigma}^{-1})_{ii}]^{1/2})$. In Fig.~\ref{fig:gaussian2} we
plot the normalized histograms of $\mtx{z}_{S_0}$ (in red) and $\mtx{z}_{S_0^c}$
(in white), where $\mtx{z}_{S_0}$ and $\mtx{z}_{S_0^c}$ respectively denote the
restrictions of $\mtx{z}$ to the active set $S_0$ and the inactive set $S_0^c$.  The plot
clearly exhibits the fact that $\mtx{z}_{S^c_0}$ has (asymptotically)
standard normal distribution and the histogram of $\mtx{z}_{S_0}$ appears
as a distinguishable bump. This is the core intuition in defining \sdl. 
%For further experiments regarding the design matrices with i.i.d. entries we refer
%to Appendix~\ref{app:iid_simulation}. 
% 
\begin{figure}[!t]
\centering
\subfigure[\scriptsize{Comparison between \sdl\, and ridge-based regression \cite{BuhlmannSignificance}}.] {
\includegraphics*[width =
3.5in]{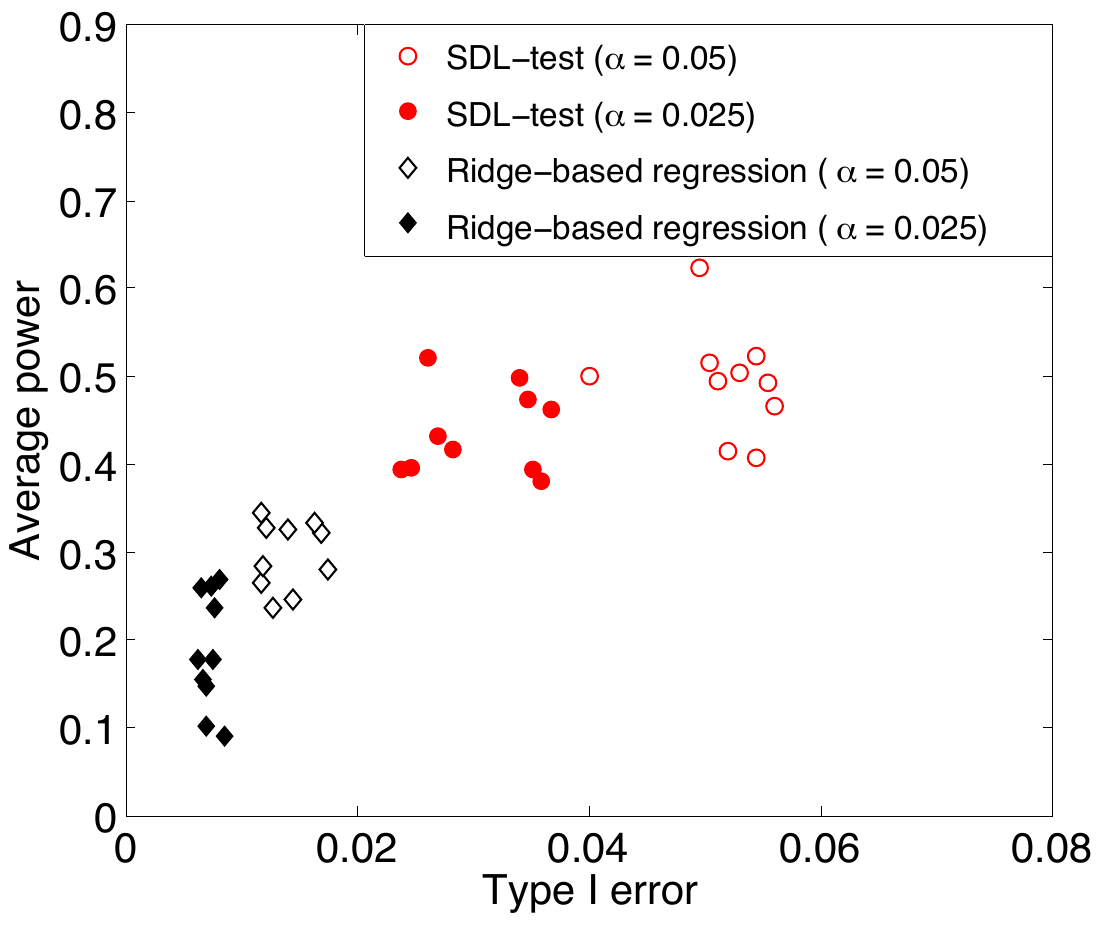}
\label{fig:plot}
%\includegraphics*[viewport = -10 190 600 600, width =
%2.6in]{figs/gaussian1.pdf}
%\label{fig:gaussian1}
}\hspace{0.1in}
\subfigure[ \scriptsize{Normalized histograms of $\mtx{z}_{S_0}$ (in red) and $\mtx{z}_{S_0^c}$ (in white) for one realization.}] {
\includegraphics*[width =
3.5in]{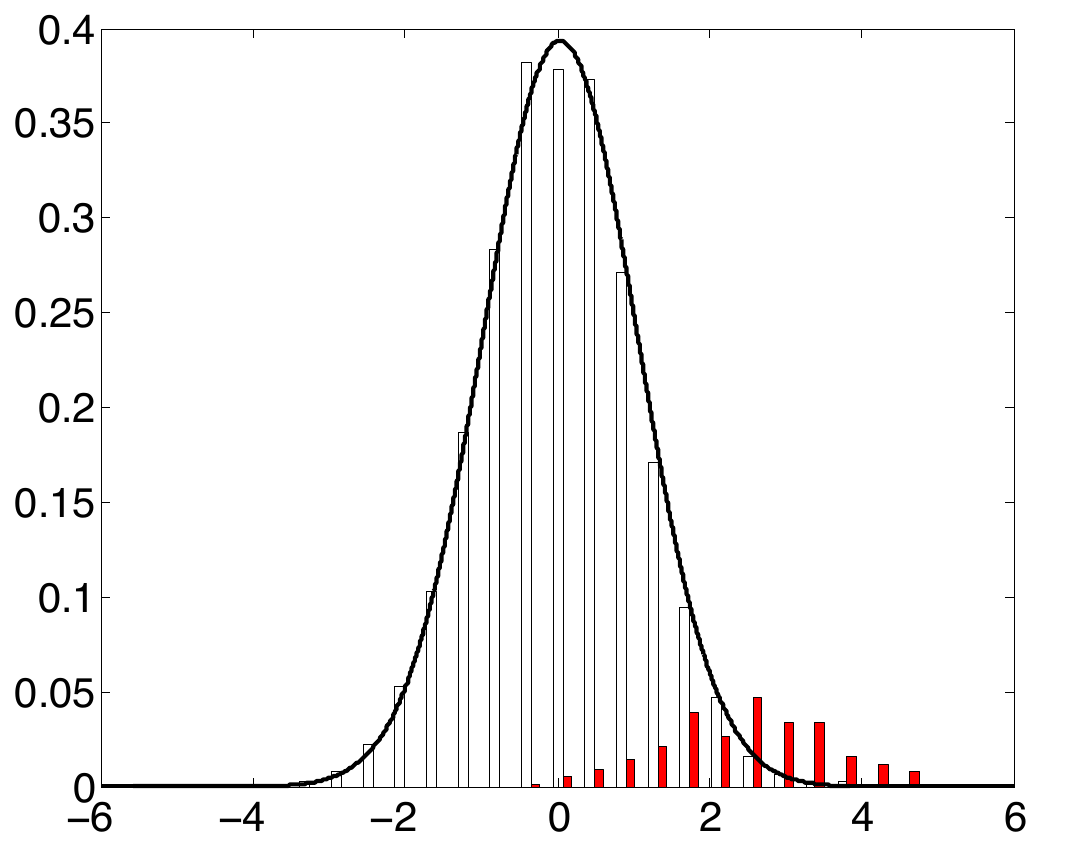}
\label{fig:gaussian2}
}
\caption{Numerical results for setting of Section~\ref{sec:experiment-nstd} and $p = 2000$, $n = 600$, $s_0 = 50$, $\mu = 0.1$.}
\end{figure}

\begin{figure}[!h]
\centering
\includegraphics*[width =3.3in]{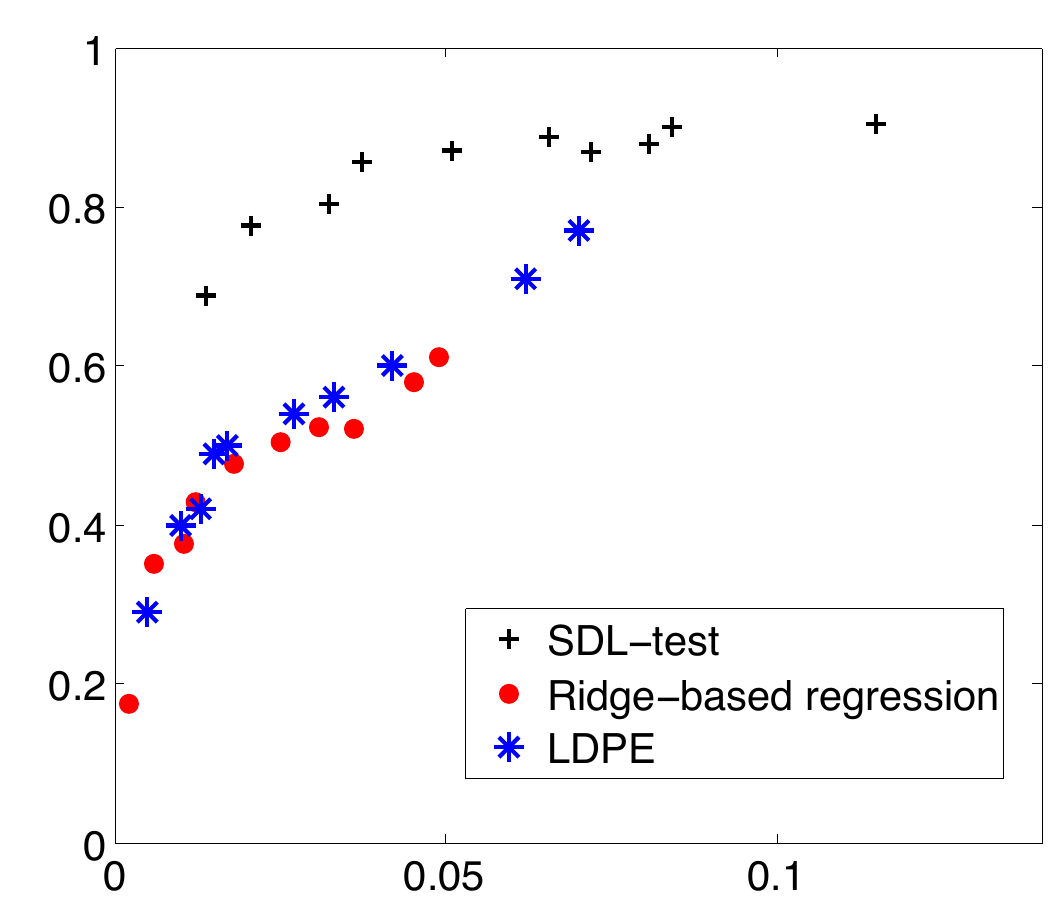}
\caption{Comparison between $\sdl\,$, ridge-based regression~\cite{BuhlmannSignificance}, and $\ldpe$~\cite{ZhangZhangSignificance} in the setting of nonstandard Gaussian designs. For the same values of type I error achieved by methods, $\sdl\,$ results in a higher statistical power. Here, $p = 1000, n = 600, s_0 = 25, \mu = 0.15$. }
\label{fig:different_alpha_cov}
\end{figure}
%

%
%\begin{figure}[!t]
%\centering
%\includegraphics*[viewport = -10 140 600 640, width =
%2.6in]{figs/plot.pdf}
%\caption{Comparison between \Proced 1 and B\"uhlmann's method~\cite{BuhlmannSignificance} for the illustration setup (Section~\ref{sec:illustration}).
%Here, $p = 2000, \delta = 0.3, \ve = 0.025, b = 0.1$.}
%\label{fig:plot}
%\end{figure}
\begin{table*}[!h]
\begin{center}
{\small
\begin{tabular}{c|cccc|}
%\cline{2-13}
%& \multicolumn{12}{|c|}{\bf Y} \\  \cline{2-13}

{\bf Method} & Type I err & Type I err & Avg. power & Avg. power \\ 
&(mean) & (std.) & (mean) & (std)
\\
\cline{1-5} \cline{2-5}
\multicolumn{1}{c|}{\bf \sdl\, $(1000, 600, 100, 0.1)$} &  0.06733 & 0.01720 & 0.48300 & 0.03433 
    \\ 
 \multicolumn{1}{c|}{\bf Ridge-based regression $(1000, 600, 100, 0.1)$} & \color{red}0.00856 & \color{red}0.00416 & \color{red}0.17000 & \color{red}0.03828
 \\
\multicolumn{1}{c|}{\bf $\ldpe\,$ $(1000, 600, 100, 0.1)$} & \color{darkblue}0.01011 & \color{darkblue} 0.00219& \color{darkblue} 0.29503& \color{darkblue} 0.03248
 \\
 \multicolumn{1}{c|}{\bf Lower bound $(1000, 600, 100, 0.1)$} &\color{olivegreen}0.05 & \color{olivegreen} NA& \color{olivegreen}0.45685& \color{olivegreen}0.04540
     \\ \hline
\multicolumn{1}{c|}{\bf \sdl\, $(1000, 600, 50, 0.1)$} &  0.04968 & 0.00997 & 0.50800 & 0.05827 
    \\ 
\multicolumn{1}{c|}{\bf Ridge-based regression $(1000, 600, 50, 0.1)$} &  \color{red}0.01642 & \color{red}0.00439 & \color{red}0.21000 & \color{red} 0.04738
\\
  \multicolumn{1}{c|}{\bf $\ldpe\,$ $(1000, 600, 50, 0.1)$} & \color{darkblue}0.02037 & \color{darkblue} 0.00751& \color{darkblue} 0.32117& \color{darkblue} 0.06481
 \\
 \multicolumn{1}{c|}{\bf Lower bound $(1000, 600, 50, 0.1)$} &\color{olivegreen}0.05 & \color{olivegreen}NA& \color{olivegreen}0.50793& \color{olivegreen}0.03545
\\ 
\hline 
\multicolumn{1}{c|}{\bf \sdl\, $(1000, 600, 25, 0.1)$} &  0.05979 & 0.01435 & 0.55200 & 0.08390 
    \\ 
\multicolumn{1}{c|}{\bf Ridge-based regression $(1000, 600, 25, 0.1)$} &  \color{red}0.02421 & \color{red}0.00804 & \color{red}0.22400& \color{red}0.10013
\\
 \multicolumn{1}{c|}{\bf $\ldpe\,$ $(1000, 600, 25, 0.1)$} & \color{darkblue}0.02604 & \color{darkblue} 0.00540& \color{darkblue} 0.31008& \color{darkblue} 0.06903
 \\
 \multicolumn{1}{c|}{\bf Lower bound $(1000, 600, 25, 0.1)$} &\color{olivegreen}0.05 & \color{olivegreen}NA& \color{olivegreen}0.54936& \color{olivegreen}0.06176
\\ 
\end{tabular}
}
\end{center}
\caption{ Comparison between \sdl, ridge-based regression~\cite{BuhlmannSignificance}, $\ldpe$~\cite{ZhangZhangSignificance} and the lower bound for \sdl\, power (cf. Theorem~\ref{thm:power2}) on the setup described in Section~\ref{sec:experiment-nstd}. The significance level is $\alpha = 0.05$. The means and the standard deviations are obtained by testing over $10$ realizations of the corresponding configuration. Here a quadruple such as $(1000, 600, 50, 0.1)$ denotes the values of $p = 1000$, $n = 600$, $s_0 = 50$, $\mu = 0.1$. \vspace{-0.2cm}}\label{tbl:small_alpha05}
\end{table*}
%

%**********************************************************************
\section{Discussion}\label{sec:discussion}

In this section we compare our contribution with related work in order 
to put it in proper perspective. We first compare it with other recent 
debiasing methods
\cite{ZhangZhangSignificance,GBR-hypothesis,confidenceJM} in subsection
\ref{sec:Comparison}. 
In subsection \ref{sec:RoleScale}
we then discuss the role of of the factor $\scale$ in our definition
of $\mtx{\htheta}^u$: this is an important difference with respect to
the methods of
\cite{ZhangZhangSignificance,GBR-hypothesis,confidenceJM}.
We finally contrast the Gaussian limit in
Theorem \ref{pro:lasso-limit}  and Le Cam's local asymptotic normality
theory,  that plays a pivotal role in classical statistics.

\subsection{Comparison with other debiasing methods}
\label{sec:Comparison}

As explained several times in the previous sections, the key step in
our procedure is to correct the Lasso estimator through a debiasing
procedure. For the reader's convenience, we copy here the definition
of the latter:
\begin{align}\label{eq:ThetaUDisc}
\mtx{\htheta}^u = \mtx{\htheta}(\lambda) + \frac{\scale}{n}
\mtx{\Sigma}^{-1} \bX^\sT(\mtx{y} - \bX \mtx{\htheta}(\lambda)).
\end{align}
The approach of~\cite{ZhangZhangSignificance} is similar in that it is based on debiased estimator of the form
\begin{align}\label{eq:de-biasing}
\mtx{\htheta^*} = \mtx{\htheta} + \frac{1}{n} \mtx{M} \bX^\sT(\mtx{y}-\bX\mtx{\htheta})\,,
\end{align}
where $\mtx{M}$ is computed from the design matrix $\bX$.  The authors
of \cite{ZhangZhangSignificance} propose to compute $\mtx{M}$ by
doing sparse regression of each column of $\bX$ onto the
others.
%\footnote{In other words, $\mtx{M}$ is viewed as a `sparse inverse' of $\bX\bX^{\sT}/n$.}. 

After a first version of the present paper became available as an
online preprint, de Geer, B\"uhlmann and Ritov~\cite{GBR-hypothesis}
studied an approach similar to \cite{ZhangZhangSignificance} (and to
ours) in a random design setting. They provide guarantees under the
assumptions that $\mtx{\Sigma}^{-1}$ is sparse and that 
the sample size $n$ asymptotically dominates $(s_0 \log p)^2$.
The authors also establish asymptotic optimality of their method in terms of semiparametric efficiency.
The semiparametric  setting is also at the center
of~\cite{ZhangZhangSignificance,Report-Semi}.

A further development over the approaches of
\cite{ZhangZhangSignificance,GBR-hypothesis}
was proposed by the present authors in \cite{confidenceJM}.
This paper constructs the  matrix $\mtx{M}$ by solving an optimization problem that controls the bias of $\mtx{\htheta^*}$
and minimize its variance meanwhile. This method does not require any sparsity assumption on $\mtx{\Sigma}$ or
$\mtx{\Sigma}^{-1}$, but still requires sample size $n$ to
asymptotically dominate $(s_0 \log p)^2$.

It is interesting to compare and contrast the results of
\cite{ZhangZhangSignificance,GBR-hypothesis,confidenceJM},
with the contribution of the present paper.
(Let us emphasize that \cite{GBR-hypothesis} appeared
after submission of the present work.)
\begin{description}
\item[Assumptions on the design matrix.] The approach of
  \cite{ZhangZhangSignificance,GBR-hypothesis,confidenceJM}
guarantees control of type I error, and optimality for  non-Gaussian
designs. (Both of  \cite{ZhangZhangSignificance,GBR-hypothesis}
require however sparsity of $\mtx{\Sigma}^{-1}$.)

In contrast, our results are fully rigorous only in the special case $\mtx{\Sigma}=\id$.
\item [Covariance estimation.] Neither of the papers
  \cite{ZhangZhangSignificance,GBR-hypothesis,confidenceJM}
 requires knowledge of covariance $\mtx{\Sigma}$. 
The method in~\cite{GBR-hypothesis} estimates $\mtx{\Sigma}^{-1}$
assuming that it is sparse, however the method~\cite{confidenceJM} does not require such estimation.

In contrast, our generalization to arbitrary Gaussian designs
postulates knowledge of $\mtx{\Sigma}$. (Further this generalization
relies
on the standard distributional limit assumption.)
\item [Sample size assumption.] The work of
  \cite{GBR-hypothesis,confidenceJM} focuses on random designs, but
  requires $n$ much larger than $(s_0\log p)^2$. This is roughly the
  square of the number of samples needed for consistent estimation.

 In contrast, we achieve similar power, and confidence intervals with
optimal sample size
$n = O(s_0\log(p/s_0))$.
 \end{description}
In summary, the present  work is complementary to the one in
\cite{ZhangZhangSignificance,GBR-hypothesis,confidenceJM} in that it
provides a sharper characterization, within a more restrictive
setting.
Together, these papers provide support for the use of debiasing
methods of the form (\ref{eq:de-biasing}).

\subsection{Role of the factor $\scale$} 
\label{sec:RoleScale}

It is worth stressing one subtle, yet interesting, difference between
the methods of
of~\cite{ZhangZhangSignificance,GBR-hypothesis} and the one of the
present paper. In both cases, a debiased estimator is constructed
using Eq.~(\ref{eq:de-biasing}).
However:
\begin{itemize}
\item The approach of \cite{ZhangZhangSignificance,GBR-hypothesis}
  sets $\mtx{M}$ to be an estimate of $(\mtx{\Sigma}^{-1})$.  In
  the idealized situation where $\mtx{\Sigma}$ is known, this
  construction reduces to setting $\mtx{M} = \mtx{\Sigma}^{-1}$.
\item In contrast, our prescription (\ref{eq:ThetaUDisc}) amounts to
  setting $\mtx{M} = \scale\,\mtx{\Sigma}^{-1}$, with $\scale =
  (1-\|\mtx{\htheta}\|_0/n)^{-1}$.
In other words, we choose $\mtx{M}$ as a \emph{scaled version of the
  inverse covariance}.
\end{itemize}
The mathematical reason for the specific scaling factor is
elucidated by the proof of Theorem \ref{pro:lasso-limit}  in
\cite{BayatiMontanariLASSO}.
Here we limit ourselves to illustrating through numerical
simulations that this factor is indeed crucial to ensure the
normality of $(\htheta_i^u-\theta_{0,i})$
in the regime $n= \Theta(s_0\log(p/s_0))$.

We consider the same setup as in Section~\ref{sec:experiment-nstd} where the rows of the design matrix are generated 
independently from $\normal(0,\mtx{\Sigma})$ with
$\mtx{\Sigma}_{jk}$ given by~\eqref{eq:Sigma} for $j\le k$.
We fix undersampling ratio $\delta = n/p$ and sparsity level $\eps = s_0/p$ and consider values $p \in \{250,500,750,\cdots, 3500\}$.
We also take active sets $S_0$ with $|S_0|= s_0$ chosen uniformly at random from the index set $\{1,\cdots,p\}$
and set $\theta_{0,i} = 0.15$ for $i\in S$.

The goal is to illustrate the effect of the  scaling factor $\scale$
on the empirical distribution of $(\htheta^u_i-\theta_{0,i})$, for
large $n, p, s_0$. As we will see, the effect becomes more pronounced
as the ratio $n/s_0 = \delta/\ve$ (i.e. the number of samples per
non-zero coefficient) becomes smaller. As above, we use 
$\mtx{\htheta}^u$ for the unbiased estimator developed in this paper
(which amounts to Eq.~(\ref{eq:de-biasing}) with $\mtx{M} =
\scale\mtx{\Sigma}^{-1}$). We will use $\mtx{\htheta}^{\scale=1}$ for
the `ideal' unbiased estimator corresponding to the proposal of \cite{ZhangZhangSignificance,GBR-hypothesis}
(which amounts to Eq.~(\ref{eq:de-biasing}) with $\mtx{M} =
\mtx{\Sigma}^{-1}$).
\begin{description} 
\item[$\bullet$ {$\bf n=3\, s_0$}] ($\eps = 0.2, \delta =0.6$).
Let $\mtx{v} = (v_i)_{i=1}^p$ with $v_i \equiv
(\htheta_i-\theta_{0,i}) /(\tau[(\mtx{\Sigma}^{-1})_{ii}]^{1/2})$. In
Fig~\ref{fig:Kurtosis-eps02}, the empirical kurtosis\footnote{Recall
  that the empirical of sample kurtosis is defined as $\kappa \equiv
  (m_4/m_2^2)-3$ with $m_{\ell}\equiv
  p^{-1}\sum_{i=1}^p(v_i-\overline{v})^{\ell}$ and $\overline{v}
  \equiv p^{-1}\sum_{i=1}^pv_i$.} of
$\{v_i\}_{i=1}^p$ is plotted for the two cases $\htheta_i =
\htheta^u_i$, and $\htheta_i = \htheta^{\scale=1}_i$. 
When using $\htheta^u$, the kurtosis is very small and data are
consistent with the kurtosis vanishing as $p \to \infty$. This
is suggestive of the fact that $(\htheta^u_i-\theta_{0,i})
/(\tau[(\mtx{\Sigma}^{-1})_{ii}]^{1/2})$
is asymptotically Gaussian, and hence satisfies a standard
distributional limit.
 However, if we use $\mtx{\htheta}^{\scale=1}$, the empirical kurtosis of
 $\mtx{v}$ does not converge to zero.

In Fig.~\ref{fig:hist-eps02}, we plot the histogram of $\mtx{v}$ for
$p = 3000$ and using both $\mtx{\htheta}^{u}$  and $\mtx{\htheta}^{\scale=1}$. Again, the plots clearly demonstrate
importance of $\scale$ in obtaining a Gaussian behavior.   

\begin{figure}[]
\centering
\subfigure[$\eps = 0.2$] {
\includegraphics*[width=3in]{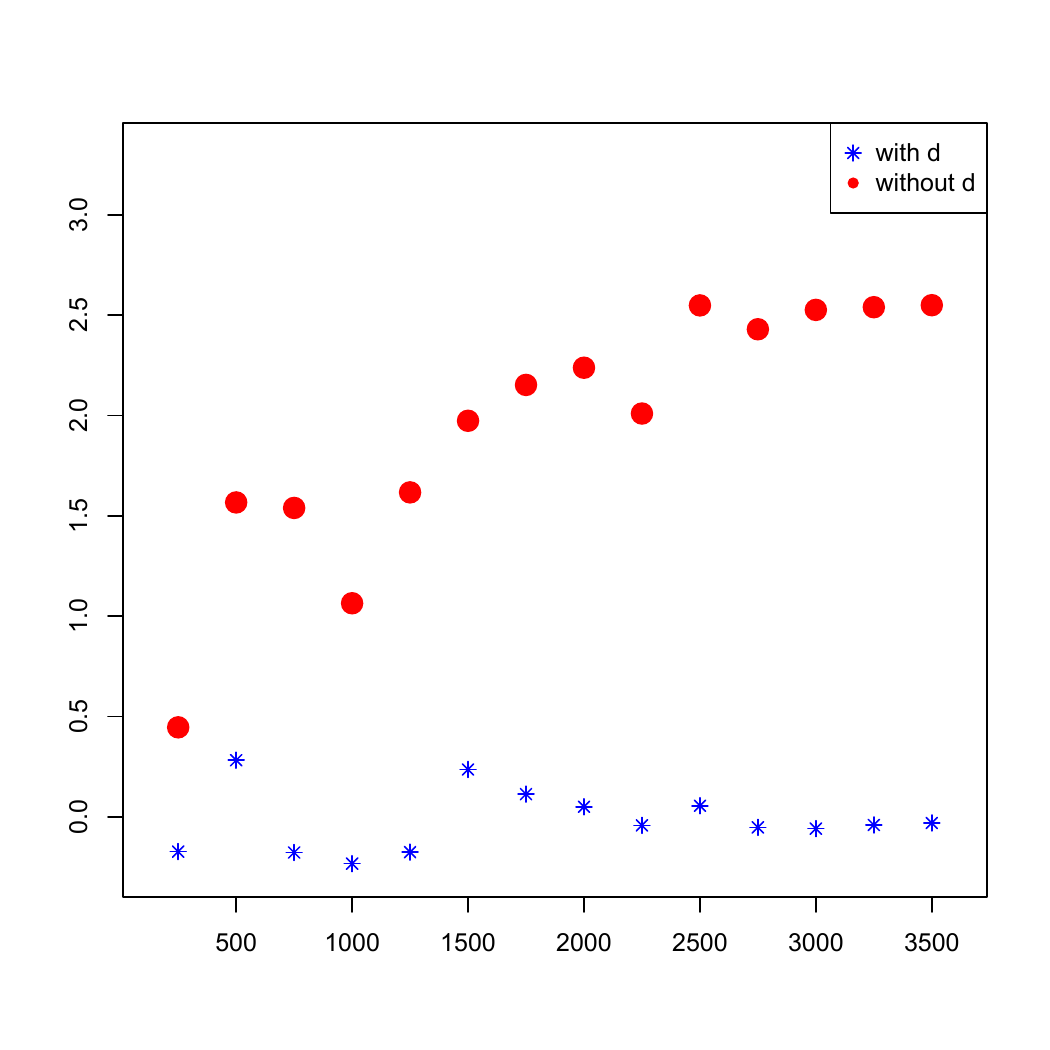}
\label{fig:Kurtosis-eps02}
\put(-105,5){$p$}
\put(-230,60){\rotatebox{90}{Empirical kurtosis of $\mtx{v}$}}
}\hspace{0.8cm}
\subfigure[$\eps = 0.02$] {
\includegraphics*[width =3in]{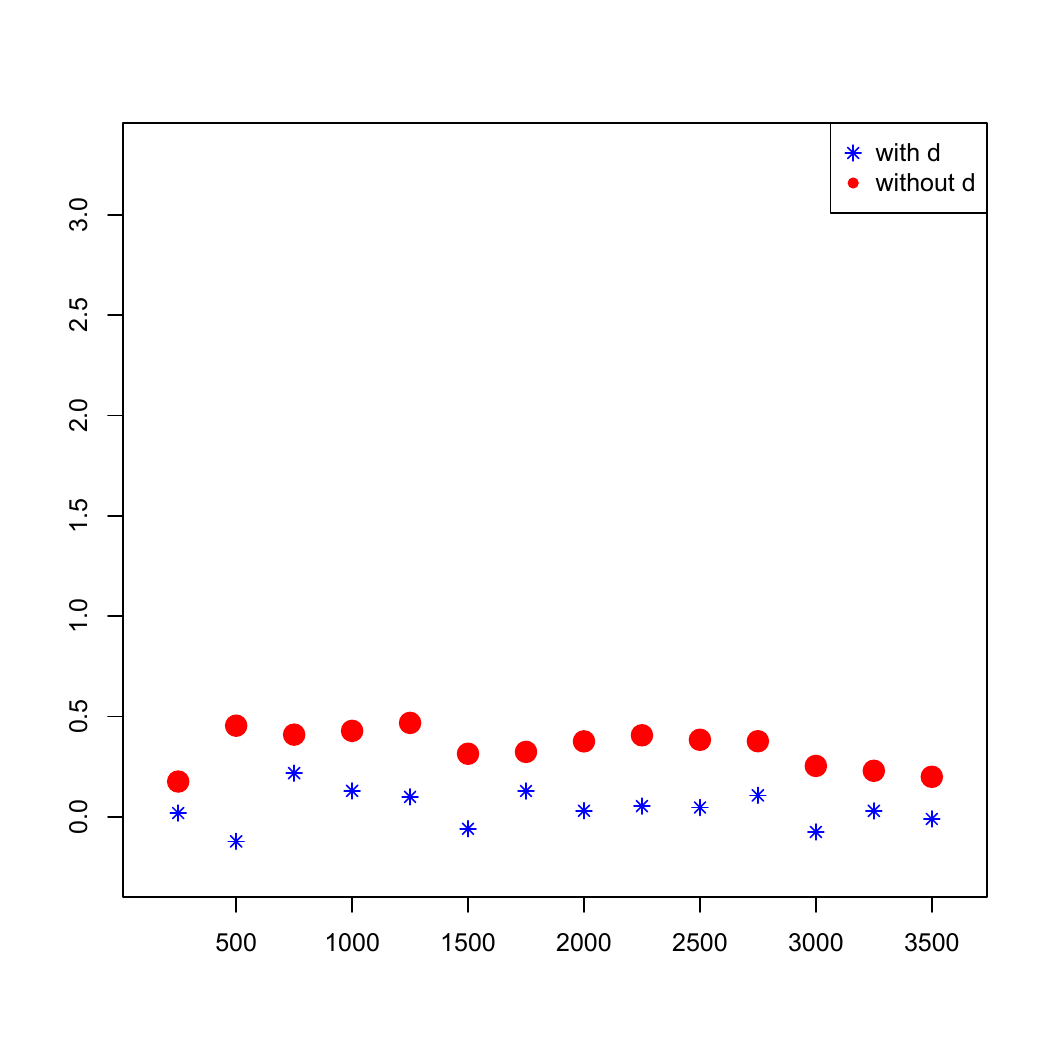}
\label{fig:Kurtosis-eps002}
\put(-105,5){$p$}
\put(-230,60){\rotatebox{90}{Empirical kurtosis of $\mtx{v}$}}
}
\caption{Empirical kurtosis of vector $\mtx{v}$ with and without
  normalization factor $\scale$. In left panel $n=3\, s_0$
(with $\eps=0.2$, $\delta=0.6$) and in the right panel $n = 30\, s_0$ (with $\eps=0.02$, $\delta=0.6$).}
\end{figure}
\begin{figure}[]
\centering
\subfigure[with factor $\scale$] {
\includegraphics*[width=3in]{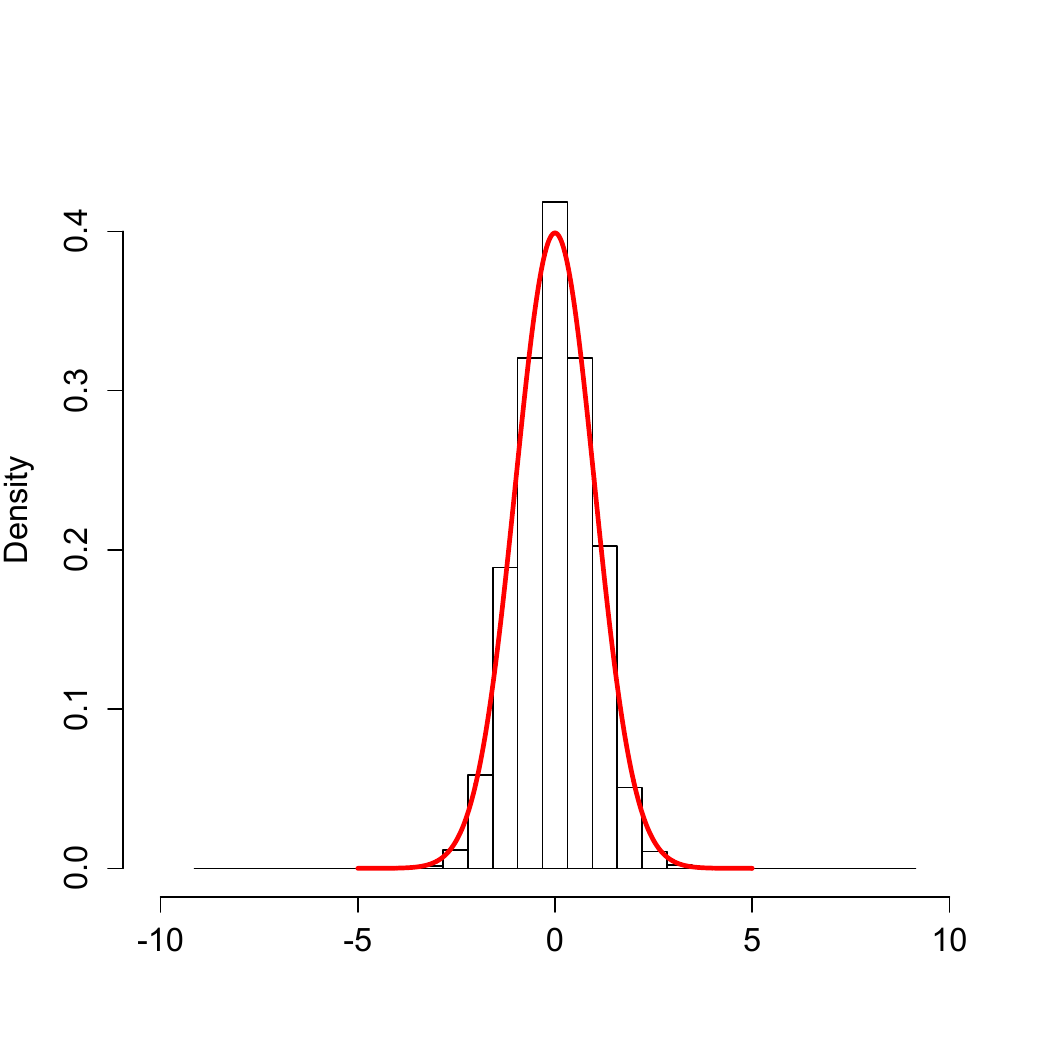}
\label{fig:hist1-eps02}
\put(-105,5){$p$}
%\put(-230,60){\rotatebox{90}{Histogram of $\mtx{v}$}}
\put(-140,190){Histogram of $\mtx{v}$}
}\hspace{0.8cm}
\subfigure[without factor $\scale$] {
\includegraphics*[width =3in]{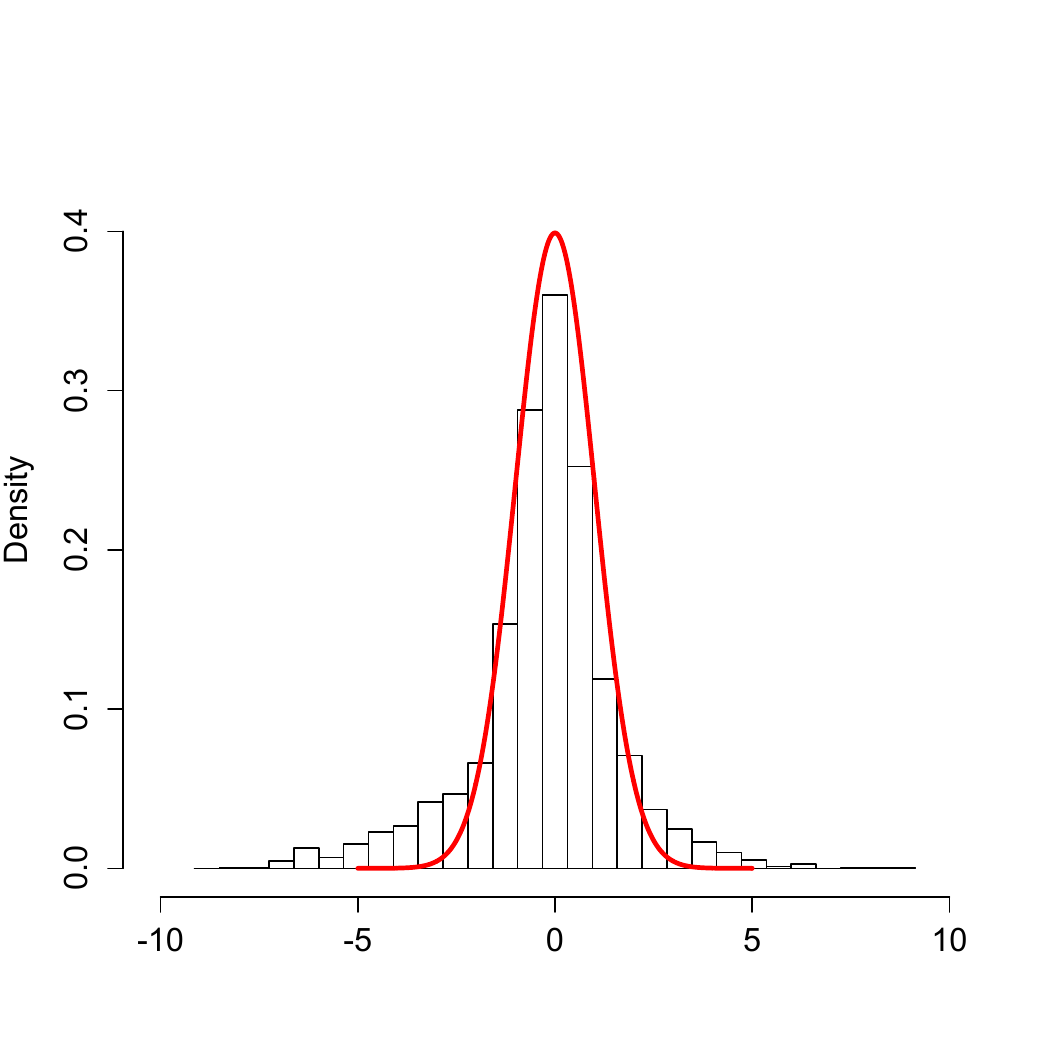}
\label{fig:hist2-eps02}
\put(-105,5){$p$}
%\put(-230,60){\rotatebox{90}{Histogram of $\mtx{v}$}}
\put(-140,190){Histogram of $\mtx{v}$}
}
\caption{Histogram of $\mtx{v}$ for $n=3\, s_0$ ($\eps = 0.2$, $\delta
  = 0.6$) and $p = 3000$. In left panel, factor $\scale$ is computed by Eq.~\eqref{eq:Step2} and in the right panel, $\scale =1$.}\label{fig:hist-eps02}
\end{figure}
\begin{figure}[]
\centering
\subfigure[with factor $\scale$] {
\includegraphics*[width=3in]{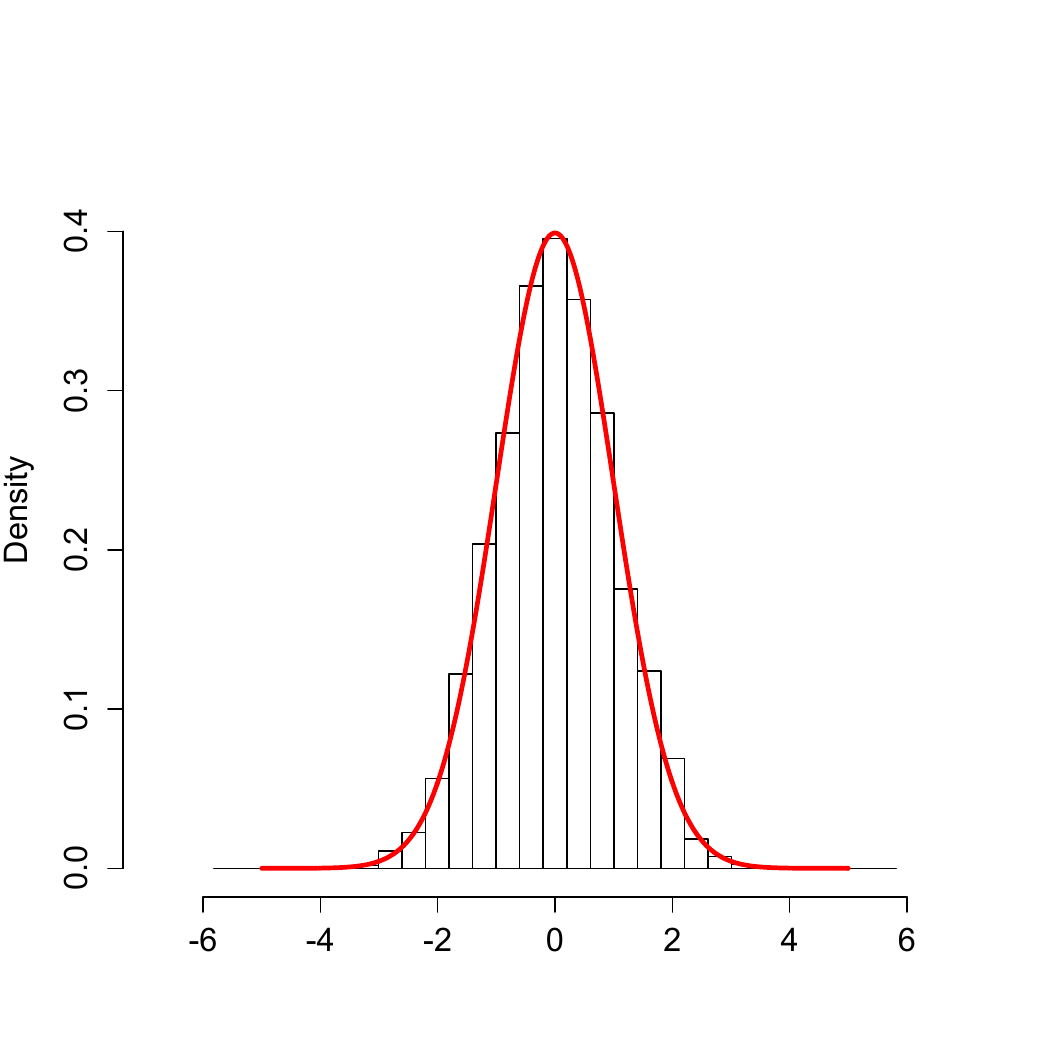}
\label{fig:hist1-eps002}
\put(-105,5){$p$}
%\put(-230,60){\rotatebox{90}{Histogram of $\mtx{v}$}}
\put(-140,190){Histogram of $\mtx{v}$}
}\hspace{0.8cm}
\subfigure[without factor $\scale$] {
\includegraphics*[width =3in]{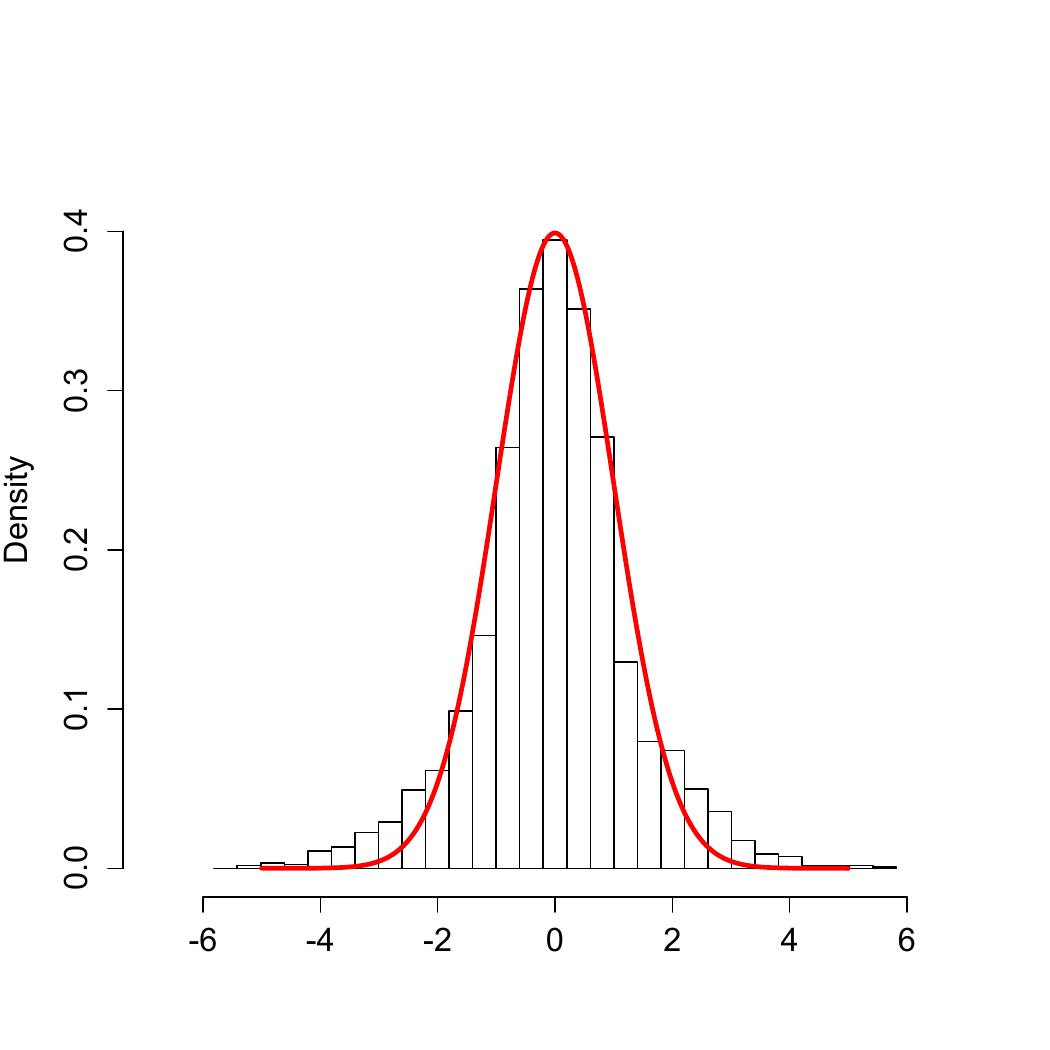}
\label{fig:hist2-eps002}
\put(-105,5){$p$}
%\put(-230,60){\rotatebox{90}{Histogram of $\mtx{v}$}}
\put(-140,190){Histogram of $\mtx{v}$}
}
\caption{Histogram of $\mtx{v}$ for $n=30\, s_0$ ($\eps = 0.02$, $\delta
  = 0.6$)  and $p = 3000$. In left panel, factor $\scale$ is computed by Eq.~\eqref{eq:Step2} and in the right panel, $\scale =1$.}\label{fig:hist-eps002}
\end{figure}

\item[$\bullet$ {$\bf n=30\, s_0$}] ($\eps = 0.02, \delta =0.6$). Figures \ref{fig:Kurtosis-eps002} and \ref{fig:hist-eps002} show similar plots for this case. As we see, the effect of $\scale$ becomes less noticeable
here. The reason is that we expect $\|\mtx{\htheta}\|_0/n = O(s_0/n)$,
and $\scale = (1-\|\mtx{\htheta}\|_0/n)^{-1} = 1+O(s_0/n)\approx 1$
for $s_0$ much smaller than $n$.
\end{description}

\subsection{Comparison with Local Asymptotic Normality}

Our approach is based on an asymptotic distributional characterization
of the Lasso estimator, cf. Theorem \ref{pro:lasso-limit}. Simplifying, the
Lasso estimator is in correspondence with a debiased estimator
$\mtx{\htheta}^u$
that is asymptotically normal  in the sense of finite-dimensional distributions.
This is analogous to what happens in classical statistics, where 
\emph{local asymptotic normality} (LAN) can be used to characterize an
estimator distribution, and hence derive test statistics
\cite{le1956asymptotic,van2000asymptotic}.

This analogy is only superficial, and the mathematical phenomenon
underlying Theorem \ref{pro:lasso-limit} is altogether different from
the one in local asymptotic normality. We refer to
\cite{BayatiMontanariLASSO} for a more complete understanding, and
only mention a few points:
\begin{enumerate}
\item LAN theory holds in the low-dimensional limit, where the
  number of parameters $p$ is much smaller than the number of samples
  $n$. Even more, the focus is on $p$ fixed, and $n\to\infty$. 

In contrast, the Gaussian limit in  Theorem \ref{pro:lasso-limit}
holds with $p$ proportional to $n$.
\item The starting point of LAN theory is low-dimensional consistency,
  namely $\mtx{\htheta}\to \mtx{\theta}_0$ as $n\to\infty$. As a
  consequence, the distribution of $\mtx{\htheta}$ can be
  characterized by a local approximation around $\mtx{\theta}_0$.

In contrast, in the high-dimensional asymptotic regime of
Theorem \ref{pro:lasso-limit}, the mean square error \emph{per coordinate}
$\|\mtx{\htheta}-\mtx{\theta}_0\|_2^2/p$ remains bounded away from zero
\cite{BayatiMontanariLASSO}. As a consequence, normality does not
follow from local approximation.
\item Indeed, in the present case, the Lasso estimator (which is of
  course a special case
  of M-estimator) $\mtx{\htheta}$ is \emph{not normal}. Only the
  debiased estimator $\mtx{\htheta}^u$ is asymptotically normal. 
Further, while LAN theory holds quite generally in the classical
asymptotics, the present theory is more sensitive to the properties of
the design matrix $\bX$. 
\end{enumerate}
%*********************************************************************
%
\section{Real data application} \label{sec:crime}
\begin{figure}[!t]
\centering
\includegraphics*[viewport = -10 40 550 510, width =
3.2in]{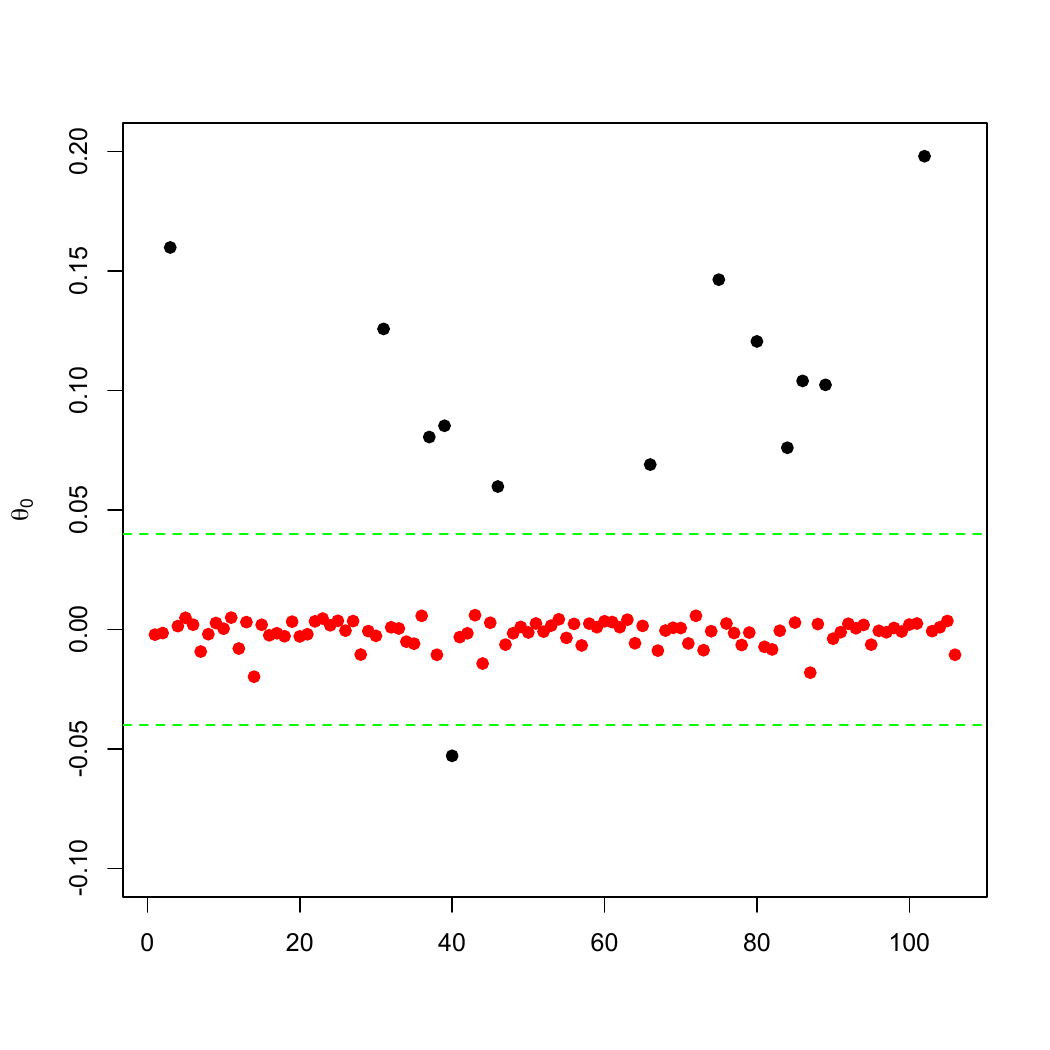}
\caption{Parameter vector $\mtx{\theta}_0$ for the communities data set.}
\label{fig:true_theta}
\end{figure}
We tested our method on the UCI communities and crimes dataset~\cite{FrankAsuncion2010}. This
concerns the prediction of the rate of violent crime in different
communities within US, based on other demographic attributes of the 
communities. The dataset consists of a response variable along with
122 predictive attributes for 1994 communities.  Covariates are
quantitative, including e.g., the fraction of urban population or the median family income. We consider a linear model as in~\eqref{eq:NoisyModel} and hypotheses $H_{0,i}$. Rejection of $H_{0,i}$ indicates that the $i$-th attribute is significant in predicting the response variable. 

We perform the following preprocessing steps: $(i)$ Each missing value
is replaced by the mean of the non missing values of that attribute
for other communities. $(ii)$ We eliminate $16$ attributes to make the
ensemble of the attribute vectors linearly independent. 
Thus we obtain a design matrix $\bX_{\rm tot}\in\reals^{n_{\rm
    tot}\times p}$ with $n_{\rm tot}=1994$ and $p = 106$;
$(iii)$  We normalize each column of the resulting design matrix to
have mean zero and $\ell_2$ norm  equal to $\sqrt{n_{\rm tot}}$.
%Therefore, for each community $i$, we have $y_i \in \reals$ and $x_i \in \reals^p$.

In order to evaluate various hypothesis testing procedures, we need to
know the true significant variables. To this end,
we let $\mtx{\theta}_0 =  (\bX_{\rm tot}^{\sT}\bX_{\rm tot})^{-1}\bX_{\rm
  tot}^{\sT} \by$ be the least-square estimator, using the whole data
set. Figure~\ref{fig:true_theta} shows the
the entries of $\mtx{\theta}_0$.  Clearly, only a few entries have non negligible
values which correspond to the significant attributes. 
In computing type I errors and powers, we take the elements in $\mtx{\theta}_0$ with magnitude larger than $0.04$ as active and the others as inactive.

In order to validate our approach in the high-dimensional $p>n$
regime, we take random subsamples of the communities (hence  subsamples of
the rows of $\bX_{\rm tot}$) of size $n = 84$. We
compare \sdl\, with the method of \cite{BuhlmannSignificance}, over $20$ realizations and
significance levels $\alpha = 0.01, 0.025, 0.05$. The fraction of type I errors and
statistical power is computed by comparing to $\mtx{\theta}_0$. Table~\ref{tbl:crime}
summarizes the results. As the reader can see, Buhlmann's method is
very conservative yielding to no type-I errors and but much  smaller power than
\sdl.

In table~\ref{tbl:featuresname}, we report the relevant features obtained from the whole dataset
as described above, corresponding to the nonzero entries in
$\mtx{\theta}_0$.
We also report  the  features identified as relevant by \sdl\, and
those identified as relevant by Ridge-based regression method,
from one random subsample of communities of size $n=84$.
Features description is available in~\cite{FrankAsuncion2010}.

Finally, in Fig.~\ref{fig:gaussian_crime_partial}
we plot the normalized histograms of $\mtx{v}_{S_0}$ (in red) and
$\mtx{V}_{S_0^c}$ (in white). Recall that $\mtx{v} = (v_i)_{i=1}^p$ denotes the vector
 with $v_i\equiv \htheta^u_i/(\tau[(\mtx{\Sigma}^{-1})_{ii}]^{1/2})$. Further,
 $\mtx{v}_{S_0}$ and $\mtx{v}_{S_0^c}$ respectively denote the
restrictions of $\mtx{v}$ to the active set $S_0$ and the inactive set $S_0^c$.
 This plot demonstrates that $\mtx{v}_{S^c_0}$ has
roughly standard normal distribution as predicted by the theory.

\begin{figure}[!t]
\begin{minipage}{0.45\textwidth}
\includegraphics*[viewport = 70 200 540 590, width =
2.9in]{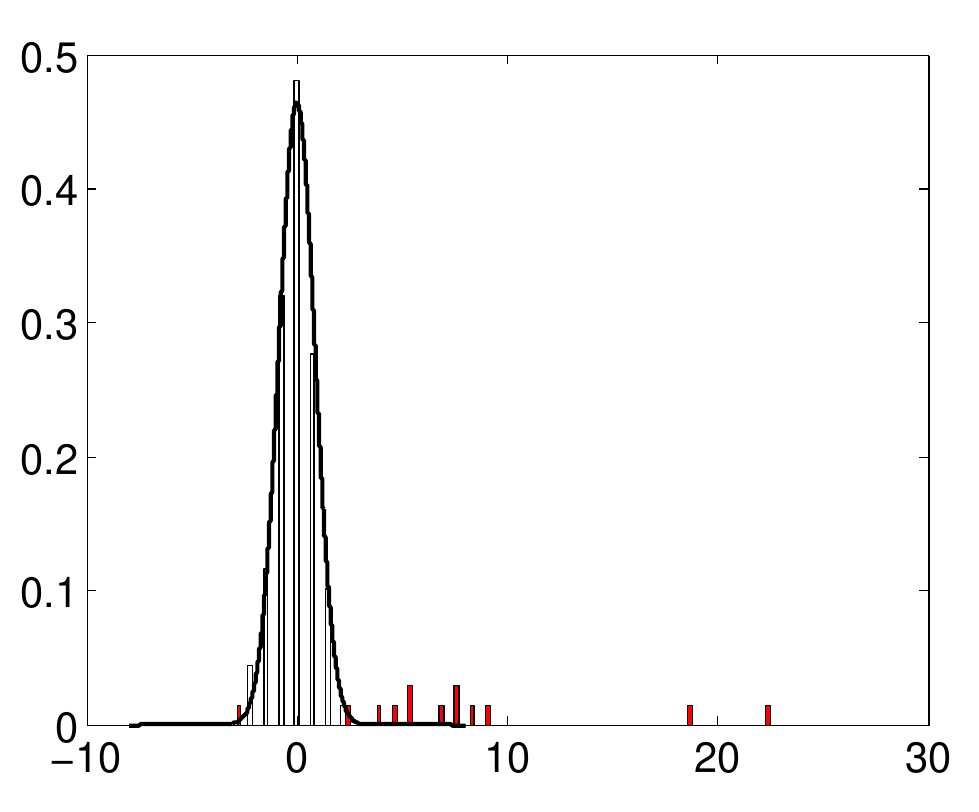}
\captionof{figure}{Normalized histogram of $\mtx{v}_{S_0}$ (in red) and $\mtx{v}_{S_0^c}$ (in white) for the communities data set.}
\label{fig:gaussian_crime_partial}
\end{minipage}
\hspace{0.4cm}
\begin{minipage}{0.5\textwidth}
{\small
\vspace{0.4cm}
\begin{tabular}{c|cc|}
%\cline{2-13}
%& \multicolumn{12}{|c|}{\bf Y} \\  \cline{2-13}

{\bf Method} & Type I err & Avg. power \\ 
&(mean) & (mean)
\\
\cline{1-3} \cline{2-3}
\multicolumn{1}{c|}{\bf \sdl\, $(\alpha = 0.05)$} & 0.0172043 &  0.4807692
    \\ 
 \multicolumn{1}{c|}{\bf Ridge-based regression} & \color{red} 0 & \color{red} 0.1423077
     \\ \hline
\multicolumn{1}{c|}{\bf \sdl\, $(\alpha = 0.025)$} &  0.01129032&   0.4230769
    \\ 
\multicolumn{1}{c|}{\bf Ridge-based regression } &  \color{red}0 & \color{red}0.1269231
     \\ \hline
\multicolumn{1}{c|}{\bf \sdl\, $(\alpha = 0.01)$} &  0.008602151&   0.3576923
    \\ 
\multicolumn{1}{c|}{\bf Ridge-based regression } &  \color{red}0 & \color{red}0.1076923
\end{tabular}}
\vspace{1.8cm}
\captionof{table}{Simulation results for the communities data set.}\label{tbl:crime}
\end{minipage}
\end{figure}
\vspace{-1cm}
\begin{table}[h]
\begin{center}
{\small
\begin{tabular}{|m{4cm}|m{4cm}|m{10cm}|}
\hline
\multicolumn{2}{ |c| }{Relevant features} & 
racePctHisp, PctTeen2Par, PctImmigRecent, PctImmigRec8, PctImmigRec10,
PctNotSpeakEnglWell,  OwnOccHiQuart, NumStreet, PctSameState85,
LemasSwFTFieldPerPop, LemasTotReqPerPop, RacialMatchCommPol, PolicOperBudg
 \\
\hline\hline
\multicolumn{1}{ |c }{\multirow{2}{*}{$\alpha = 0.01$} } &
\multicolumn{1}{ |>{\centering}p{4cm}| }{Relevant features (\sdl\,) } & 
racePctHisp, PctTeen2Par, PctImmigRecent, PctImmigRec8, PctImmigRec10,
PctNotSpeakEnglWell, OwnOccHiQuart, NumStreet, PctSameState85,
LemasSwFTFieldPerPop, LemasTotReqPerPop, RacialMatchCommPol, PolicOperBudg
\\ \cline{2-3}
\multicolumn{1}{ |c  }{} &
\multicolumn{1}{ |>{\centering}p{4cm}| }{ Relevant features (ridge-based regression)} & 
racePctHisp, PctSameState85
\\ 
\hline\hline
\multicolumn{1}{ |c  }{\multirow{2}{*}{$\alpha = 0.025$} } &
\multicolumn{1}{ |>{\centering}p{4cm}| }{Relevant features (\sdl\,)} & 
racePctHisp, PctTeen2Par, PctImmigRecent, PctImmigRec8, PctImmigRec10,
PctNotSpeakEnglWell, {\color{red}{PctHousOccup}}, OwnOccHiQuart, NumStreet,
PctSameState85, LemasSwFTFieldPerPop, LemasTotReqPerPop,
RacialMatchCommPol, PolicOperBudg
\\ \cline{2-3}
\multicolumn{1}{ |c  }{} &
\multicolumn{1}{ |>{\centering}p{4cm}| }{Relevant features (ridge-based regression)} & 
racePctHisp, PctSameState85
\\ 
\hline\hline
\multicolumn{1}{ |c  }{\multirow{2}{*}{$\alpha = 0.05$} } &
\multicolumn{1}{ |>{\centering}p{4cm}| }{Relevant features (\sdl\,)} & 
racePctHisp, {\color{red}{PctUnemployed}}, PctTeen2Par, PctImmigRecent, PctImmigRec8,
PctImmigRec10, PctNotSpeakEnglWell, {\color{red}{PctHousOccup}}, OwnOccHiQuart,
NumStreet, PctSameState85, {\color{red}{LemasSwornFT}}, LemasSwFTFieldPerPop,
LemasTotReqPerPop, RacialMatchCommPol, {\color{red}{PctPolicWhite}}
\\ \cline{2-3}
\multicolumn{1}{ |c  }{} &
\multicolumn{1}{ |>{\centering}p{4cm}| }{Relevant features (ridge-based regression)} & 
racePctHisp, PctSameState85
 \\ 
\hline
\end{tabular}
}
\end{center}
\caption{The relevant features (using the whole dataset) and the relevant features
predicted by \sdl\, and  the method of \cite{BuhlmannSignificance} for a random subsample of size
$n=84$ from the communities. The false positive predictions are in red.}\label{tbl:featuresname}
\end{table}

%====================================================
%
%*********************************************************************************
%

\pagebreak
\section{Proofs}\label{sec:Proofs}

\subsection{Proof of Lemma \ref{lemma:Binary}}
\label{sec:ProofBinary}

Fix $\alpha\in [0,1]$, $\mu>0$, and assume that the minimum error rate for type
II errors in testing hypothesis $H_{0,i}$ at significance level
$\alpha$ is $\beta = \bopt_i(\alpha;\mu)$. Further fix $\xi>0$
arbitrarily small.
By definition there exists a statistical test $T_{i,\bX}:\reals^m\to\{0,1\}$ such that
$\prob_{\mtx{\theta}}(T_{i,\bX}(\by) =1) \le \alpha$ for any
$\mtx{\theta}\in\cR_0$ and $\prob_{\mtx{\theta}}(T_{i,\bX}(\by) =0) \le
\beta+\xi$ for any $\mtx{\theta}\in \cR_1$ (with $\cR_0,
\cR_1\in\reals^p$ defined as in Definition \ref{definition:RandomizedTheta}).
Equivalently:
\begin{align}
\begin{split}
\E\big\{\prob_{\mtx{\theta}}(T_{i,\bX}(\by) =1|\bX) \big\}&\le \alpha,\quad\quad \;\;
\mbox{for any  $\mtx{\theta}\in\cR_0$,}\\
\E\big\{\prob_{\mtx{\theta}}(T_{i,\bX}(\by) =0|\bX) \big\}&\le \beta+\xi,\;\;\;
\mbox{for any  $\mtx{\theta}\in\cR_1$.}
\end{split}
\end{align}
We now take expectation of these inequalities with respect to
$\mtx{\theta}\sim Q_0$ (in the first case) and $\mtx{\theta}\sim Q_1$ (in the
second case) and we get, with the notation introduced in the
Definition \ref{definition:RandomizedTheta},
\begin{align*}
\begin{split}
\E\big\{\prob_{Q,0,\bX}(T_{i,\bX}(\by) =1) \big\} &\le \alpha\, ,\\
\E\big\{\prob_{Q,1,\bX}(T_{i,\bX}(\by) =0) \big\} &\le \beta+\xi\, .
\end{split}
\end{align*}
Call $\alpha_{\bX} \equiv \prob_{Q,0,\bX}(T_{i,\bX}(\by) =1)$.
By assumption, for any test $T$, we have 
$\prob_{Q,1,\bX}(T_{i,\bX}(\by) =0) \ge \bbin_{i,\bX}(\alpha_{\bX};Q)$
and therefore the last inequalities imply
\begin{align}
\begin{split}
\E\big\{\alpha_{\bX} \big\}&\le \alpha\, ,\\
\E\big\{\bbin_{i,\bX}(\alpha_{\bX};Q)\big\}&\le \beta+\xi\, .
\end{split}
\end{align}
The thesis follows since $\xi>0$ is arbitrary.

\subsection{Proof of Lemma \ref{lemma:PerXUpperBound}}
\label{sec:ProofPerXUpperBound}

Fix $\bX$, $\alpha$, $i$, $S$ as in the statement and assume, without
loss of generality, $\projp_{S}\mtx{\tx}_i\neq 0$, and $\rank(\bX_S)=|S|<n$.
We take $Q_0 = \normal(0,\mtx{J})$ where 
$\mtx{J}\in\reals^{p\times p}$ is the diagonal matrix with
$\mtx{J}_{jj}= a$ if $j\in S$ and $\mtx{J}_{jj}=0$ otherwise. Here $a\in\reals_+$ and will be chosen later.
For the same covariance matrix $\mtx{J}$, we let $Q_1 = \normal(\mu\, \mtx{e}_i, \mtx{J})$
where $\mtx{e}_i$ is the $i$-th element of the standard basis. 
Recalling that $i \notin S$, and $|S| < s_0$,
the support of $Q_0$ is in $\cR_0$ and the support of $Q_1$
is in $\cR_1$.

Under $\prob_{Q,0,\bX}$ we have $\by\sim
\normal(\mtx{0},a\,\bX_S\bX_S^{\sT}+\sigma^2\id)$, and under $\prob_{Q,1,\bX}$ we
have  $\by\sim\normal(\mu\mtx{\tx}_i,a\,\bX_S\bX_S^{\sT}+\sigma^2\id)$.
Hence the binary hypothesis testing problem under study reduces to the
problem of testing a null  hypothesis on the mean of a Gaussian
random vector with known covariance against a simple alternative. It is well known that the most
powerful test \cite[Chapter 8]{lehmann2005testing} is obtained by comparing the ratio 
$\prob_{Q,0,\bX}(\by)/\prob_{Q,1,\bX}(\by)$ with a threshold.
 Equivalently, the most powerful test is of the form
\begin{align}
T_{i,\bX}(\by) =
\ind\Big\{\<\mu\mtx{\tx}_i,(a\bX_S\bX_S^{\sT}+\sigma^2\id)^{-1}\mtx{y}\>\ge c\Big\}\, ,
\end{align}
for some $c\in\reals$ that is to be chosen to achieve the desired
significance level $\alpha$.
Letting 
\begin{align}
\alpha \equiv 2\Phi\Big(-\frac{c}{\mu\|(a\bX_S\bX_S^{\sT}+\sigma^2\id)^{-1/2}\mtx{\tx}_i\|}\Big)\,,
\end{align}
 it is a straightforward calculation to drive the power of this test as
  %
%  \[
%  G\Big(\alpha, \frac{\mu \mtx{\tx}_i^\sT (a\bX_S\bX_S^{\sT}+\sigma^2\id)^{-1} \mtx{\tx}_i}{\|(a\bX_S\bX_S^{\sT}+\sigma^2\id)^{-1/2}\mtx{\tx}_i\|}\Big)\,.
%  \]
\[
G\Big(\alpha, \mu\|(a\bX_S\bX_S^{\sT}+\sigma^2\id)^{-1/2}\mtx{\tx}_i\|\Big)\,,
\]
where the function $G(\alpha,u)$ is defined as per Eq.~\eqref{eqn:G}.
 Next we show that the power of this test
converges to $1 - \bora_{i,\bX}(\alpha;S,\mu)$ as $a\to\infty$. Hence
the claim is proved by taking $a\ge a(\xi)$ for some $a(\xi)$ large enough.

%Let $\bX_S = \mtx{U} \mtx{\Delta} \mtx{V}^\sT$ be a singular value decomposition of $\bX_S$. 
%Therefore, columns of $\mtx{U}$ form a basis for the linear subspace spanned by
%$\{\mtx{\tx}_i\}_{i\in S}$. Let $\mtx{\tilde{U}}$ be such that its columns form a basis for 
%the orthogonal subspace $\{\mtx{\tx}_i\}_{i\notin S}$. Then,
%%
%\[
%\frac{\mu \mtx{\tx}_i^\sT (a\bX_S\bX_S^{\sT}+\sigma^2\id)^{-1} \mtx{\tx}_i}{\|(a\bX_S\bX_S^{\sT}+\sigma^2\id)^{-1/2}\mtx{\tx}_i\|}
%= \frac{\mu\mtx{\tx}_i^\sT\{\mtx{U}(a \mtx{\Delta}^2+\sigma^2 \id)^{-1} \mtx{U}^\sT + \sigma^{-2} \mtx{\tilde{U}} \mtx{\tilde{U}}^\sT\}\mtx{\tx}_i}
%{\|\{\mtx{U}(a\mtx{\Delta}^2+\sigma^2\id)^{-1/2}\mtx{U}^\sT + \sigma^{-1} \mtx{\tilde{U}}\mtx{\tilde{U}}^\sT\} \mtx{\tx}_i\|}\,.
%\]
%%
Write
\begin{align}
(a \bX_S \bX_S^\sT + \sigma^2 \id)^{-1/2} &= \frac{1}{\sigma} \Big(\id + \frac{a}{\sigma^2} \bX_S \bX_S^\sT\Big)^{-1/2}\nonumber\\
&= \frac{1}{\sigma} \Big\{\id - \bX_S \Big(\frac{\sigma^2}{a} \id + \bX_S^\sT \bX_S\Big)^{-1} \bX_S^\sT\Big\}^{1/2}\,,
\end{align}
where the second step follows from matrix inversion lemma. Clearly, as $a\to \infty$, the right hand side of the above equation converges to  $(1/\sigma)\, \projp_S$. Therefore, the power converges to $1 - \bora_{i,\bX}(\alpha;S,\mu) = G(\alpha, \mu \sigma^{-1}\|\projp_S \mtx{\tx}_i\|)$.
%%%%%%%%%%%%%%%%%%%%%%%%%%%%%%%%%
\subsection{Proof of Theorem \ref{thm:GeneralUpperBound}}
\label{sec:ProofGeneralUpperBound}

%Define, for $\alpha\in [0,1]$ and $u\in\reals_+$,
%%
%\begin{align}\label{eq:G}
%%
%G(\alpha,u)\equiv 2-\Phi\Big(\Phi^{-1}(1-\frac{\alpha}{2})+\frac{\mu\, u}{\sigma}\Big) -
%\Phi\Big(\Phi^{-1}(1-\frac{\alpha}{2})-\frac{\mu\, u}{\sigma}\Big)\, .
%%
%\end{align}
%%
Let $u_{\bX}\equiv \mu\|\projp_S\mtx{\tx}_i\|_2/\sigma$.  
By Lemma \ref{lemma:Binary} and \ref{lemma:PerXUpperBound}, we
have, 
\begin{align}
1-\bopt_{i}(\alpha;\mu)  \le \sup\Big\{\E G(\alpha_\bX,u_{\bX})\, :\;\;
\E(\alpha_{\bX})\le\alpha\Big\}\, ,
\end{align}
with the $\sup$ taken over measurable functions $\bX\mapsto
\alpha_{\bX}$, and $G(\alpha,u)$ defined as per Eq.~\eqref{eqn:G}. 

It is easy to check that $\alpha\mapsto G(\alpha,u)$ is concave for
any $u\in\reals_+$ and $u\mapsto G(\alpha,u)$ is non-decreasing for
any $\alpha\in [0,1]$ (see Fig.~\ref{fig:AlphaBeta1d}). Further $G$ takes values in $[0,1]$. Hence
\begin{align}
\begin{split}
\E G(\alpha_\bX,u_{\bX})&\le \E\big\{G(\alpha_{\bX},u_{\bX}) \ind(u\le
u_0)\big\} +\prob(u_{\bX}>u_0)\\
&\le  \E\big\{G(\alpha_{\bX},u_0)\big\} +\prob(u_{\bX}>u_0)\\
&\le  G(\E (\alpha_{\bX}),u_0) +\prob(u_{\bX}>u_0)\\
&\le G(\alpha,u_0) +\prob(u_{\bX}>u_0)
\end{split}
\end{align}
Since $\mtx{\tx}_i$ and $\bX_S$ are jointly Gaussian, we have
\begin{align}
\mtx{\tx}_i = \mtx{\Sigma}_{i,S}\mtx{\Sigma}_{S,S}^{-1}\bX_S + {\Sigma}_{i|S}^{1/2} \mtx{z}_i\, ,
\end{align}
with $\mtx{z}_i\sim \normal(0,\id_{n\times n})$ independent of $\bX_S$. It
follows that
\begin{align}
u_{\bX} = (\mu/\sigma)\,{\Sigma}_{i|S}^{1/2} \, \|\projp_S\mtx{z}_i\|_2 \ed
(\mu/\sigma)\, \sqrt{{\Sigma}_{i|S} Z_{n-s_0+1}}\, ,
\end{align}
with $Z_{n-s_0+1}$ a chi-squared random variable with $n-s_0+1$
degrees of freedom. The desired claim follows by taking $u_0= (\mu/\sigma)\sqrt{{\Sigma}_{i|S}(n-s_0+\ell)}$.

%%%%%%%%%%%%%%%%%%%%%%%%

%%%%%%%%%%%%%%%%%%%%%%%%
\subsection{Proof of Theorem~\ref{thm:power}}\label{sec:ProofPower}

Since $\{(\mtx{\Sigma}(p) = \id_{p \times p}, \mtx{\theta}_0(p), n(p), \sigma(p))\}_{p \in \naturals}$ has a standard distributional limit, the empirical distribution of $\{(\theta_{0,i}, \htheta^u_i)\}_{i=1}^p$ converges weakly to $(\Theta_0, \Theta_0 + \tau Z)$ (with probability one). By the portmanteau theorem, and the fact that $\underset{p\to \infty}{\lim \inf}\, \sigma(p) / \sqrt{n(p)} = \sigma_0$, we have
\begin{eqnarray}\label{eqn:nomass}
\P(0 < |\Theta_0| < \mu_0 \sigma_0) \le \lim_{p \to \infty} \frac{1}{p} \sum_{i=1}^p \ind\bigg(0 < \theta_{0,i} < \mu_0 \frac{\sigma(p)}{\sqrt{n(p)}}\bigg) = 0.
\end{eqnarray}
In addition, since $\mu_0 \sigma_0/2$ is a continuity point of the distribution of $\Theta_0$, we have
\begin{eqnarray}\label{eqn:contpoint}
\lim_{p\to \infty} \frac{1}{p} \sum_{i=1}^p \ind(|\theta_{0,i}| \ge \frac{\mu_0\sigma_0}{2}) = \P(|\Theta_0| \ge \frac{\mu_0 \sigma_0}{2}).
\end{eqnarray}
Now, by Eq.~\eqref{eqn:nomass}, $\P(|\Theta_0| \ge \mu_0 \sigma_0/2) = \P(\Theta_0 \neq 0)$. Further, $\ind(|\theta_{0,i}| \ge \mu_0\sigma_0/2) = \ind(\theta_{0,i} \neq 0)$ for $1\le i \le p$, as $p \to \infty$. Therefore, Eq.~\eqref{eqn:contpoint} yields
\begin{eqnarray}\label{eqn:mass_nz}
\begin{split}
\lim_{p \to \infty} \frac{1}{p} |S_0(p)| &= \lim_{p \to \infty} \frac{1}{p}\sum_{i=1}^p \ind(\theta_{0,i} \neq 0)= \P(\Theta_0 \neq 0).
\end{split}
\end{eqnarray}
Hence,
\begin{eqnarray}
\begin{split}
\lim_{p\to \infty} \frac{1}{|S_0(p)|} \sum_{i \in S_0(p)} T_{i,\bX}(\by) &= 
\lim_{p\to \infty} \frac{1}{|S_0(p)|} \sum_{i \in S_0(p)} \ind(P_i \le \alpha)\\
& = \frac{1}{\P(\Theta_0 \neq 0)} \lim_{p\to \infty} \frac{1}{p}\sum_{i=1}^p \ind(P_i \le \alpha, i\in S_0(p))\\
&=\frac{1}{\P(\Theta_0 \neq 0)} \lim_{p\to \infty} \frac{1}{p}  \sum_{i =1}^p \ind\bigg(\Phi^{-1}(1-\alpha/2) \le \frac{|\htheta^u_i|}{\tau},\, |\theta_{0,i}| \ge \mu_0\frac{\sigma(p)}{\sqrt{n(p)}}\bigg)\\
&\ge \frac{1}{\P(\Theta_0 \neq 0)} \P \bigg(\Phi^{-1}(1-\alpha/2) \le \big|\frac{\Theta_0}{\tau} + Z\big|, \,\, |\Theta_0| \ge \mu_0 \sigma_0\bigg).
\end{split}
\end{eqnarray}
Note that $\tau$ depends on the distribution $p_{\Theta_0}$. Since $|S_0(p)| \le \ve p$, using Eq.~\eqref{eqn:mass_nz}, we have $\P(\Theta_0 \neq 0) \le \ve$, i.e, $p_{\Theta_0}$ is $\ve$-sparse. Let $\tilde{\tau}$ denote the maximum $\tau$ corresponding to densities in the family of $\ve$-sparse densities. As shown in~\cite{DMM09}, $\tilde{\tau} = \tau_* \sigma_0$, where $\tau_*$ is defined by Eqs.~\eqref{eqn:tau_*} and~\eqref{eqn:M_ve}. Consequently,
\begin{eqnarray}
\begin{split}
\lim_{p\to \infty} \frac{1}{|S_0(p)|} \sum_{i \in S_0(p)} T_{i,\bX}(\by)
 &\ge \P \bigg(\Phi^{-1}(1-\alpha/2) \le \big|\frac{\mu_0}{\tau_*} + Z\big|\bigg)\\
&= 1 - \P \bigg(-\Phi^{-1}(1-\alpha/2)- \frac{\mu_0}{\tau_*} \le Z\le \Phi^{-1}(1-\alpha/2)- \frac{\mu_0}{\tau_*}\bigg)\\
& = 1- \{\Phi(\Phi^{-1}(1-\alpha/2)- \mu_0/\tau_*) - \Phi(-\Phi^{-1}(1-\alpha/2)- \mu_0/\tau_*)\}\\
& = G(\alpha,\mu_0/\tau_*)\,.
\end{split}\label{eq:S0-T}
\end{eqnarray}
%
% where $\beta(\mu_0)$ is given by Eq.~\eqref{eqn:beta_min}.
 
 Now, we take the expectation of both sides of Eq.~\eqref{eq:S0-T} with respect to the law of random design $\bX$
 and random noise $\mtx{w}$. Changing the order of limit and expectation by applying dominated convergence theorem and using
 linearity of expectation, we obtain
\begin{eqnarray}
\lim_{p\to \infty} \frac{1}{|S_0(p)|} \sum_{i \in S_0(p)} \E_{\bX,\mtx{w}}\{T_{i,\bX}(\by)\} \ge G\Big(\alpha,\frac{\mu_0}{\tau_*}\Big)\,.
\end{eqnarray}
Since $T_{i,\bX}(\by)$ takes values in $\{0,1\}$, we have $\E_{\bX,\mtx{w}}\{T_{i,\bX}(\by)\} = \prob_{\mtx{\theta}_0(p)}(T_{i,\bX}(\by)=1)$. The result follows by
 noting that the columns of $\bX$ are exchangeable and therefore $\prob_{\mtx{\theta}_0(p)}(T_{i,\bX}(\by)=1)$ does not depend on $i$.
 
 %%%%%%%%%%%%%%%%%%%%%%%%%
 \subsection{Proof of Theorem~\ref{thm:type_I_nstd}}
Since the sequence $\{\mtx{\Sigma}(p), \mtx{\theta}_0(p), n(p), \sigma(p)\}_{p \in \naturals}$ has a standard distributional limit, with probability one the empirical distribution of $\{(\theta_{0,i}, \htheta^u_i, (\mtx{\Sigma}^{-1})_{ii})\}_{i=1}^p$ converges weakly to the distribution of $(\Theta_0, \Theta_0 +\tau \Upsilon^{1/2} Z, \Upsilon)$. Therefore, with probability one, the empirical distribution of 
\begin{eqnarray*}
\bigg\{\frac{\htheta^u_i - \theta_{0,i}}{\tau[(\mtx{\Sigma}^{-1})_{ii}]^{1/2}} \bigg\}_{i=1}^p 
\end{eqnarray*}
converges weakly to $\normal(0,1)$. Hence,
\begin{eqnarray}
\begin{split}
\lim_{p\to \infty} \frac{1}{|S_0^c(p)|} \sum_{i \in S_0^c(p)} T_{i,\bX}(\by) &=
\lim_{p\to \infty} \frac{1}{|S^c_0(p)|} \sum_{i \in S^c_0(p)} \ind(P_i \le \alpha)\\
&= \frac{1}{\P(\Theta_0 = 0)} \, \lim_{p\to \infty} \frac{1}{p} \sum_{i =1}^p \ind(P_i \le \alpha,\, i\in S^c_0(p))\\
&= \frac{1}{\P(\Theta_0 = 0)} \, \lim_{p\to \infty} \frac{1}{p}  \sum_{i =1}^p \ind\bigg(\Phi^{-1}(1-\alpha/2) \le \frac{|\htheta^u_i|}{\tau[(\mtx{\Sigma}^{-1})_{ii}]^{1/2}},\, \theta_{0,i} = 0\bigg)\\
&= \frac{1}{\P(\Theta_0 = 0)} \,\P(\Phi^{-1}(1-\alpha/2) \le |Z|, \Theta_0 = 0)\\
& = \P(\Phi^{-1}(1-\alpha/2) \le |Z|) = \alpha.
\end{split}
\end{eqnarray}
Applying the same argument as in the proof of Theorem~\ref{thm:power}, we obtain the following by taking the expectation 
of both sides of the above equation
\begin{eqnarray}
\lim_{p\to \infty} \frac{1}{|S_0(p)|} \sum_{i \in S_0(p)} \prob_{\mtx{\theta}_0(p)}(T_{i,\bX}(\by)=1) = \alpha\,.
\end{eqnarray}
In particular, for the standard Gaussian design (cf. Theorem~\ref{thm:type_I}), since the columns of $\bX$ are exchangeable we
get $\lim_{p\to \infty} \prob_{\mtx{\theta}_0(p)}(T_{i,\bX}(\by)=1) = \alpha $ for all $i\in S_0(p)$.

 %%%%%%%%%%%%%%%%%%%%%%%%%%
 
 \subsection{Proof of Theorem~\ref{thm:power2}}
 The proof of Theorem~\ref{thm:power2} proceeds along the same lines as the proof of Theorem~\ref{thm:power}. Since $\{(\mtx{\Sigma}(p), \mtx{\theta}_0(p), n(p),$  $\sigma(p))\}_{p \in \naturals}$ has a standard distributional limit, with probability one the empirical distribution of $\{(\theta_{0,i}, \htheta^u_i, (\mtx{\Sigma}^{-1})_{ii})\}_{i=1}^p$ converges weakly to the distribution of $(\Theta_0, \Theta_0 +\tau \Upsilon^{1/2} Z, \Upsilon)$. Similar to Eq.~\eqref{eqn:mass_nz}, we have
 \begin{eqnarray}
 \lim_{p \to \infty} \frac{1}{p} |S_0(p)| = \P(\Theta_0 \neq 0).
 \end{eqnarray}
 Also
\begin{align}
\lim_{p\to \infty} \frac{1}{|S_0(p)|} \sum_{i \in S_0(p)} T_{i,\bX}(\by) &=
\lim_{p\to \infty} \frac{1}{|S_0(p)|} \sum_{i \in S_0(p)} \ind(P_i \le \alpha)\nonumber\\
& = \frac{1}{\P(\Theta_0 \neq 0)} \lim_{p\to \infty} \frac{1}{p}\sum_{i=1}^p \ind(P_i \le \alpha, i\in S_0(p))\nonumber\\
&=\frac{1}{\P(\Theta_0 \neq 0)} \lim_{p\to \infty} \frac{1}{p}  \sum_{i =1}^p \ind\bigg(\Phi^{-1}(1-\alpha/2) \le \frac{|\htheta^u_i|}{\tau[(\mtx{\Sigma}^{-1})_{ii}]^{1/2}},\, \frac{|\theta_{0,i}|}{[(\mtx{\Sigma}^{-1})_{ii}]^{1/2}} \ge \mu_0\bigg)\nonumber\\
&= \frac{1}{\P(\Theta_0 \neq 0)} \P \bigg(\Phi^{-1}(1-\alpha/2) \le \big|\frac{\Theta_0}{\tau\Upsilon^{1/2}} + Z\big|, \,\, \frac{|\Theta_0|}{\Upsilon^{1/2}} \ge \mu_0\bigg)\nonumber\\
&\ge \frac{1}{\P(\Theta_0 \neq 0)} \P \bigg(\Phi^{-1}(1-\alpha/2) \le \big|\frac{\mu_0}{\tau} + Z\big|\bigg)\nonumber\\
& = 1- \{\Phi(\Phi^{-1}(1-\alpha/2)- \mu_0/\tau) - \Phi(-\Phi^{-1}(1-\alpha/2)- \mu_0/\tau)\}\nonumber\\
& = G(\alpha,\mu_0/\tau)\,.
\end{align}
%
% where $\tilde{\beta}(\mu_0)$ is given by Eq.~\eqref{eqn:tildebeta_min}. 

Similar to the proof of Theorem~\ref{thm:power}, by taking the expectation of both sides of the above inequality we get
\begin{eqnarray}
\lim_{p\to \infty} \frac{1}{|S_0(p)|} \sum_{i \in S_0(p)} \prob_{\mtx{\theta}_0}(T_{i,\bX}(\by)=1)
\ge  G\Big(\alpha,\frac{\mu_0}{\tau}\Big)\,.
\end{eqnarray}
% 
%%%%%%%%%%%%%%%%%%%%%%%%%%%%%%%%%%%%%%%%%%%%%
\subsection{Proof of Theorem~\ref{thm:Replica-Rigor}}
\label{proof:Replica-Rigor}

In order to prove the claim, we
will establish the following (corresponding to the the case
$\Theta_0=0$ of Definition \ref{def:SDL}):
\begin{description}
\item[Claim 1.] If $\tau$ solves
  Eq.~(\ref{eq:ClaimedTau}), then $\tau^2\to
  \sigma^2_0$ as $p\to\infty$.
\item[Claim 2.] The empirical distribution of
  $\{(\theta_{0,i},\htheta^u_i,(\mtx{\Sigma}^{-1})_{ii})\}_{1\le i\le
    p}$ converges weakly to the random vector
  $(0,\sigma_0\Upsilon^{1/2}Z,\Upsilon)$, with $Z\sim\normal(0,1)$
  independent of $\Upsilon$. Namely fixing
  $\psi:\reals^3\to \reals$, bounded Lipschitz, we need to prove
\begin{align}
\lim_{p\to\infty}\frac{1}{p}\sum_{i=1}^p
\psi(\theta_{0,i},\htheta^{u}_i,(\mtx{\Sigma}^{-1})_{ii}) =
\E\big\{\psi(0,\sigma_0\Upsilon^{1/2} Z, \Upsilon)\big\}\, .
\end{align}
\item[Claim 3.] Recalling $\mtx{r} \equiv \scale (\mtx{y} - \mtx{X} \mtx{\htheta})/\sqrt{n}$, the empirical distribution of $\{r_i\}_{1\le i\le n}$ converges
  weakly to $\normal(0,\sigma_0^2)$.
\end{description}
We will prove these three claims after some preliminary remarks. 
First notice that, by \cite[Theorem 6]{rudelson2013reconstruction}
(and using assumptions $(i)$ and $(iii)$)
$\bX$ satisfies the restricted eigenvalue property RE$(s_0,3s_0,3)$ of
\cite{BickelEtAl} with a $p$-independent constant $\kappa = \kappa(c_{\rm min},c_{\rm
  max})>0$, almost surely for all $p$ large enough. (Indeed
Theorem 6 of \cite{rudelson2013reconstruction} ensures that this holds
with probability at least $1-e^{-\Omega(n(p))}$, and hence almost surely
for all $p$ large enough by Borel-Cantelli lemma.)

We can therefore apply \cite[Theorem 7.2]{BickelEtAl} to conclude that
there exists a constant $C_0$ such that, almost surely for all $p$
large enough, we have
\begin{align}
\|\bX(\mtx{\htheta}-\mtx{\theta}_0)\|_2^2 & \le \frac{1}{2}\,
C_0s_0\sigma^2\log p\le C_0\sigma_0^2ns_0\log p\, ,\label{eq:PRED}\\
\|\mtx{\htheta}-\mtx{\theta}_0\|_1& \le \frac{C_0s_0\sigma}{2} \sqrt{\frac{\log
    p}{n}}\le C_0\sigma_0s_0\sqrt{\log p}\, ,\label{eq:ELL1}\\
\|\mtx{\htheta}-\mtx{\theta}_0\|_2^2& \le \frac{C_0\sigma^2}{2}\, \frac{s_0\log
    p}{n}\le C_0\sigma_0^2s_0\log p\, ,\label{eq:ELL2}\\
\|\mtx{\htheta}\|_0& \le C_0 s_0\, .\label{eq:ELL0}
\end{align}
(Here we used $\sigma^2\le 2n\sigma_0^2$ for all $p$ large enough.)
In particular, from Eq.~(\ref{eq:ELL0}) and assumption $(i)$, it follows that
$\lim_{p\to\infty}\|\htheta\|_0/n=0$ and hence, almost surely,
\begin{align}
\lim_{p\to\infty}\scale = 1\, .\label{eq:LimitScale}
\end{align}

\subsubsection{Claim 1}

By Eq.~(\ref{eq:LimitScale}), we can assume $\scale\in (1/2,2)$ for
all $p$ large enough. By Eq.~(\ref{eq:ClaimedTau}) it is sufficient
to show that $\Ena(\tau^2,b)\to 0$
uniformly for $b\in [1/2,2]$, $\tau\in [0,M\sigma_0]$, for some $M\ge 2$. Since $\|\mtx{\theta}\|_0/p\to 0$, and
by dominated convergence, we have
\begin{align}
\Ena(\tau^2,b) 
&\equiv\lim_{p\to\infty}\frac{1}{p}\E\big\{\big\|\eta_{b}(\tau\mtx{\Sigma}^{-1/2}\mtx{z})\big\|_{\mtx{\Sigma}}^2\big\}\,
.
\end{align}
It is easy to see that $\|\eta_{b}(\mtx{y})\big\|_{\mtx{\Sigma}}\le
C\|\mtx{y}\|_2$ for some constant $C$ depending on $c_{\rm
  min}$, $c_{\rm max}$.
Hence, letting $Y_p \equiv
\|\eta_{b}(\tau\mtx{\Sigma}^{-1/2}\mtx{z})\big\|_{\mtx{\Sigma}}^2/p$,
we conclude that $\E\{Y_p^2\}$ is bounded uniformly in $p$. By Cauchy-Schwarz
\begin{align}
\Ena(\tau^2,b) 
&\equiv\lim_{p\to\infty}\E\{Y_p\} \le
\lim_{p\to\infty}\E\{Y_p^2\}^{1/2}\prob(Y_p\neq 0)^{1/2} \, . \label{eq:EnaCS}
\end{align}
It is therefore sufficient to prove that $\prob(Y_p\neq 0)\to 0$.
By definition of $\eta_{b}(\, \cdot\, )$, cf. Eq.~(\ref{eq:ProximalOperator}), we have 
$\eta_{b}(\mtx{y}) = 0$  if and only if  
\begin{align}
\|\bSigma \by\|_{\infty} \le \frac{\lambda}{b}\, .
\end{align}
Therefore, substituting $\by = \tau\mtx{\Sigma}^{-1/2}\mtx{z}$, we
have
\begin{align}
\prob(Y_p\neq 0) = \prob\Big(\|\bSigma^{1/2}\bz\|_{\infty}>
\frac{\lambda}{b\tau}\Big) \le 
\prob\Big(\max_{i\in [p]}|(\bSigma^{1/2}\bz)_i|>
\frac{\lambda}{2M\sigma_0}\Big) 
\, .
\end{align}
The random variables $(\bSigma^{1/2}\bz)_i$ are
$\normal(0,\Sigma_{ii})$. Therefore by union bound, since
$\Sigma_{ii}\le c_{\rm max}$, for $Z\sim\normal(0,1)$, we have
\begin{align}
\prob(Y_p\neq 0) \le p \, \prob\Big(|Z|\ge
\frac{\lambda}{2M\sigma_0 \sqrt{c_{\rm max}}}\Big) \le 2p \exp\Big(-\frac{\lambda^2}{8M^2\sigma_0^2c_{\max}}\Big)\,
.
\end{align}
Therefore $\prob(Y_p\neq 0) \to 0$ since $\lambda = C_*\sigma_0\sqrt{\log
  p}$, provided  $C_*\ge M\sqrt{8 c_{\max}}$, by Eq.~(\ref{eq:EnaCS}).

\subsubsection{Claim 2}

Let $\mtx{z} = \mtx{\Sigma}^{-1}\bX^\sT \mtx{w} /n$. Conditional on
$\bX$, we have 
\begin{eqnarray}
\mtx{z}|\bX \sim \normal(0, \mtx{C}), \quad \mtx{C} = \frac{\sigma^2}{n}\mtx{\Sigma}^{-1}\Big(\frac{\bX^\sT \bX}{n}\Big) \mtx{\Sigma}^{-1}\,.\label{eq:GaussianZ}
\end{eqnarray} 
Using the assumption $\sigma^2/n \to \sigma_0^2$ and
employing~\cite[Lemma 7.2]{JavanMon-OptSample}, we have,
almost surely,
\begin{eqnarray}
\lim_{p\to \infty}\max_{i\in [p]} \Big| C_{ii} - \sigma^2_0 (\mtx{\Sigma}^{-1})_{ii}\Big| = 0\,.\label{eq:Cii}
\end{eqnarray}
Consequently, we have, for almost every sequence of matrices $\bX$,
letting $Z\sim\normal(0,1)$ independent of $\bX$
\begin{align}
\lim_{p\to\infty}
\frac{1}{p}\sum_{i=1}^p\E\big\{\psi(0,z_i,(\mtx{\Sigma}^{-1})_{ii})\big|\bX\big\}
&= \lim_{p\to\infty}
\frac{1}{p}\sum_{i=1}^p\E\big\{\psi(0,\sqrt{C_{ii}}
Z,(\mtx{\Sigma}^{-1})_{ii})\big|\bX\big\}\\
&= \lim_{p\to\infty}
\frac{1}{p}\sum_{i=1}^p\E\big\{\psi(0,\sigma_0 \sqrt{(\mtx{\Sigma}^{-1})_{ii}}
Z,(\mtx{\Sigma}^{-1})_{ii})\big\}\\
& = \E\big\{\psi(0,\sigma_0 \Upsilon^{1/2}
Z,\Upsilon)\big\}\, .
\end{align}
(Here, the first identity follows from Eq.~(\ref{eq:GaussianZ}), the
second from Eq.~(\ref{eq:Cii}) and the Lipschitz continuity of $\psi$,
and the last from assumption $(iv)$, together with the fact that
$\psi$ is bounded Lipschitz.)

Next, applying Gaussian isoperimetry \cite{Ledoux} to the conditional measure
of $\mtx{z}$ given $\bX$ (noting that $\|\mtx{C}\|_2\le C_1$ almost
surely for all $n$ large enough and some constant $C_1<\infty$), and to the Lipschitz function 
$\mtx{z}\mapsto \Psi(\bz)\equiv p^{-1}\sum_{i=1}^p\psi(0,z_i,(\mtx{\Sigma}^{-1})_{ii})$, we have
\begin{align}
\prob\Big\{|\Psi(\bz)-\E\big(\Psi(\bz)|\bX)|\ge \eps\Big|\bX\Big\}\le
2 e^{-n\eps^2/C_1}\, ,
\end{align}
almost surely for all $n$ large enough.
Using Borel-Cantelli lemma, we conclude that, almost surely
\begin{align}
\lim_{p\to\infty}
\frac{1}{p}\sum_{i=1}^p\psi(0,z_i,(\mtx{\Sigma}^{-1})_{ii})
= \E\big\{\psi(0,\sigma_0 \Upsilon^{1/2}
Z,\Upsilon)\big\}\, .\label{eq:WeakLimitZ}
\end{align}
Substituting $\by = \bX \theta_0+ \mtx{w}$ in definition of $\mtx{\htheta}^u$, we get 
\begin{align}
\mtx{\htheta}^u -\mtx{\theta}_0 &=\Big(\frac{\scale}{n} \mtx{\Sigma}^{-1} \bX^\sT\bX - \id \Big)(\mtx{\theta}_0 - \mtx{\htheta})
+ \frac{\scale}{n} \mtx{\Sigma}^{-1}\bX^\sT \mtx{w}\nonumber\\
&= \scale \Big(\mtx{\Sigma}^{-1} \mtx{\hSigma} - \id \Big)(\mtx{\theta}_0 - \mtx{\htheta})
+(\scale-1)(\mtx{\theta}_0-\mtx{\htheta})
+ \frac{\scale}{n}\mtx{\Sigma}^{-1}\bX^\sT \mtx{w}\\
& = \Delta_1 +\Delta_2 + \scale\, \bz\, ,
\end{align}
where we recall that $\mtx{\hSigma}\equiv (\bX^{\sT}\bX)/n$ and we defined
\begin{align}
\Delta_1 =  \scale \Big(\mtx{\Sigma}^{-1} \mtx{\hSigma} - \id \Big)(\mtx{\theta}_0 - \mtx{\htheta})\,,\quad \quad
\Delta_2 = (\scale-1)(\mtx{\theta}_0-\mtx{\htheta})\,.
\end{align}
The proof is therefore concluded if we can show that, almost surely,
\begin{align}
\lim_{p\to\infty}\frac{1}{p}\sum_{i=1}^p\big|\psi(\theta_{0,i},\theta_{0,i}+\Delta_{1,i}+\Delta_{2,i}+\scale
z_i,(\mtx{\Sigma}^{-1})_{ii})
-\psi(0,z_i,(\mtx{\Sigma}^{-1})_{ii})\big| = 0\, . \label{eq:Claim1Key}
\end{align}
In order to simplify the notation, and since the last argument plays
no role, we let $\psi_i(x,y) \equiv\psi(x,y,(\mtx{\Sigma}^{-1})_{ii})$.
Without loss of generality we will assume that
$\|\psi_i\|_{\infty}\le 1$, and that the Lipschitz modulus of $\psi_i$ is at
most one. 

In order to prove the claim (\ref{eq:Claim1Key}), note that, by
triangular inequality,
\begin{align}
\frac{1}{p}\sum_{i=1}^p\big|\psi_i(\theta_{0,i},\theta_{0,i}+&\Delta_{1,i}+\Delta_{2,i}+\scale
z_i)
-\psi_i(0,z_i)\big| \\
& \le \frac{1}{p}\sum_{i=1}^p g(\theta_{0,i}) +
\frac{1}{p}\sum_{i=1}^p g(\Delta_{1,i}) +
\frac{1}{p}\sum_{i=1}^pg(\Delta_{2,i})+\frac{1}{p}\sum_{i=1}^pg\big(|\scale-1|\,z_i\big)\,
,\label{eq:FinalTriangular}
\end{align}
where $g(x) \equiv \min(|x|,2)$. 

The first term in Eq.~(\ref{eq:FinalTriangular}) vanishes since by
assumption $(i)$, $s_0\le n/(\log p)^2$, and therefore
\begin{align}
\lim_{p\to\infty}\frac{1}{p}\sum_{i=1}^p g(\theta_{0,i}) \le
\lim_{p\to\infty} \frac{2s_0}{p} = 0\, .
\end{align}

Consider next the third term in Eq.~(\ref{eq:FinalTriangular}):
\begin{align}
\frac{1}{p}\sum_{i=1}^pg(\Delta_{2,i})\le\frac{1}{p}\,|\scale-1|\,
\|\mtx{\htheta}-\mtx{\theta}_0\|_1\le \frac{1}{p}\, |\scale-1|\, 
C_0\sigma_0 s_0\sqrt{\log p}\, ,
\end{align}
where the second inequality follows from (\ref{eq:ELL1}), that holds
almost surely for all $p$ large enough. Next, using
Eq.~(\ref{eq:LimitScale}), 
\begin{align}
\lim_{p\to\infty}\frac{1}{p}\sum_{i=1}^pg(\Delta_{2,i})=0\, .
\end{align}

Consider next  the last term in Eq.~(\ref{eq:FinalTriangular}), 
and fix $\delta>0$ arbitrarily small. Since by
Eq.~(\ref{eq:LimitScale}), $|\scale-1|\le \delta$  almost surely for
all $p$ large enough, we have
\begin{align}
\lim\sup_{p\to\infty}\frac{1}{p}\sum_{i=1}^pg\big(|\scale-1|\,z_i\big)
\le \lim\sup_{p\to\infty}\frac{1}{p}\sum_{i=1}^pg\big(\delta\,z_i\big)
= \E\{g(\delta\sigma_0\Upsilon^{1/2}Z)\}\, ,
\end{align}
where the last equality follows from Eq.~(\ref{eq:WeakLimitZ}),
applied to $\psi(a,b,c) = g(\delta b)$. Finally, letting $\delta\to
0$, we get, by dominated convergence, $\lim_{\delta\to
  0}\E\{g(\delta\sigma_0\Upsilon^{1/2}Z)\}=0$,
and hence
\begin{align}
\lim\sup_{p\to\infty}\frac{1}{p}\sum_{i=1}^pg\big(|\scale-1|\,z_i\big)
= 0\, .
\end{align}

Finally, consider the second term. Fix a partition $[p] =\cup_{\ell=1}^LA_{\ell}$, where
$s_0\le |A_{\ell}|\le 9\, s_0$, and $p/(9s_0)\le L\le (p/s_0)$. Then
\begin{align}
\frac{1}{p}\sum_{i=1}^pg(\Delta_{1,i})\le\frac{1}{p}\|\Delta_1\|_1&\le\frac{1}{p}
\sum_{\ell=1}^L\sqrt{|A_{\ell}|}\,\|\Delta_{1,A_{\ell}}\|_2
\le \frac{3}{\sqrt{s_0}}\max_{1\le   \ell\le L}\,\|\Delta_{1,A_{\ell}}\|_2\, . \label{eq:GDelta1}
\end{align}
Let $T\equiv \supp(\mtx{\htheta})\cup\supp(\mtx{\theta}_0)$. By
Eq.~(\ref{eq:ELL0}) we have $|T|\le (C_0+1)s_0$ almost surely for all
$p$ large enough. Hence, using $\scale\le 2$ for all $p$ large enough,
we get
\begin{align}
\|\Delta_{1,A_{\ell}}\|_2\le
2\big\|\big(\mtx{\Sigma}^{-1}\mtx{\hSigma}-\id\big)_{A_{\ell},T}\big\|_2
\|\mtx{\htheta}-\mtx{\theta}_0\|_2\, .\label{eq:Delta1Bound}
\end{align}
The operator norm can be upper bounded using the following lemma,
whose proof can be found in Appendix \ref{app:MatrixBound}.
(See also the conference paper \cite{JavanMon-OptSample} for a similar estimate: we provide
a full proof in appendix for the reader's convenience.)
\begin{lemma}\label{lemma:MatrixBound}
Under the assumption of Theorem \ref{thm:Replica-Rigor},
for any constant $c_0$, there exists $K = K(c_{\rm min},c_{\rm max},c_0)$
\begin{align}
\max\Big\{\big\|\big(\mtx{\Sigma}^{-1}\mtx{\hSigma}-\id\big)_{A,B}\big\|_2
:\; A,B\subseteq [p],\; |A|,|B|\le c_{0}\,s_0\Big\}\le K\,
\sqrt{\frac{s_0\log p}{n}}\, , \label{eq:operator_lem}
\end{align}
with probability at least $(1-p^{-5})$ for all $p$ large enough.
\end{lemma}

Using Borel-Cantelli lemma together with Eq.~(\ref{eq:operator_lem}) and
Eq.~(\ref{eq:ELL2}) in Eq.~(\ref{eq:Delta1Bound}) we get, almost
surely for all $p$ large enough, and some constant $C$
\begin{align}
\|\Delta_{1,A_{\ell}}\|_2 \le C \sigma_0\frac{s_0\log p}{\sqrt{n}}\, 
\end{align}
Hence, using Eq.~(\ref{eq:GDelta1}) and assumption $(i)$
\begin{align}
\lim_{p\to\infty}\frac{1}{p}\sum_{i=1}^pg(\Delta_{1,i}) \le
\lim_{p\to\infty}   C' \sigma_0\frac{\sqrt{s_0}\log p}{\sqrt{n}}=0\, .
\end{align}
This finishes the proof of Claim 2.

\subsubsection{Claim 3}

Note that, by definition
\begin{align}
\mtx{r} = \frac{1}{\sqrt{n}}\mtx{w}
+\frac{\scale}{\sqrt{n}} \, \mtx{X}(\mtx{\theta}_0-\mtx{\htheta})+ 
\frac{\scale-1}{\sqrt{n}}\mtx{w}\, .
\end{align}
Defining $\bu \equiv \mtx{w}/\sqrt{n}$, $\bh_1
\equiv\scale\mtx{X}(\mtx{\theta}_0-\mtx{\htheta})/\sqrt{n}$,
and $\bh_2 \equiv (\scale-1)\bu$, the proof consists in two steps.
First, for any Lipschitz bounded function
$\psi:\reals\to\reals$, we have 
\begin{align}
\lim_{p\to\infty}\frac{1}{n}\sum_{i=1}^n\psi(u_i) =
\E\{\psi(\sigma_0\, Z)\}
\, .
\end{align}
This is immediate by the law of large numbers, since $\bu$ has i.i.d. $\normal(0,\sigma^2/n)$
entries and by assumption $\sigma^2/n\to \sigma_0^2$.

Second, we have
\begin{align}
\frac{1}{n}\sum_{i=1}^n\big|\psi(r_i)-\psi(u_i) \big|\le
\frac{1}{n}\sum_{i=1}^ng(h_{1,i})+ \frac{1}{n}\sum_{i=1}^ng(h_{2,i})\,  ,
\end{align}
and the right hand side converges to $0$ as $p\to\infty$. Here the first term is controlled
using Eq.~(\ref{eq:PRED}), and the second using
Eq.~(\ref{eq:LimitScale}).  These derivations are almost identical to the ones
of Claim 2, and we omit them.

\section*{Acknowledgements}

This work was partially supported by the NSF CAREER award CCF-0743978,
and the grants AFOSR FA9550-10-1-0360 and  AFOSR/DARPA FA9550-12-1-0411.
 
 %%%%%%%%%%%%%%%%%%%%%%%%%%%%%%%%%%%%%%%%%%%%%
 
\appendix
\section{Effective noise variance $\tau_0^2$}\label{app:tau}
As stated in Theorem \ref{pro:lasso-limit}  the unbiased estimator
$\mtx{\htheta}^u$ can be regarded 
--asymptotically-- as a noisy version
of $\mtx{\theta}_0$ with noise variance $\tau_0^2$. An explicit formula for $\tau_0$ is given in~\cite{BayatiMontanariLASSO}.
For the reader's convenience, we explain it here using our notations. 

Denote by $\eta: \reals \times \reals_+ \to \reals$ the soft thresholding function
\begin{align}
\eta(x;a) = \begin{cases}
x-a & \text{if } x> a,\\
0 & \text{if } -a\le x\le a\\
x+a & \text{otherwise.}
\end{cases}
\end{align}
Further define function $\Map: \reals_+ \times \reals_+ \to \reals_+$ as
\begin{align}
\Map(\tau^2,a) = \sigma^2 + \frac{1}{\delta} \E \{[\eta(\Theta_0 + \tau Z;a) - \Theta_0]^2\}\,,
\end{align}
where $\Theta_0$ and $Z$ are defined as in Theorem~\ref{pro:lasso-limit}.
Let $\kappa_{\min} = \kappa_{\min}(\delta)$ be the unique non-negative solution of the equation 
\begin{align}
(1+\kappa^2) \Phi(-\kappa) - \kappa \phi(\kappa) = \frac{\delta}{2}\,.
\end{align}
%
%with $\phi(z) \equiv e^{-z^2/2}/\sqrt{2\pi}$ the standard normal density and $\Phi(z)$ its cumulative distribution function.
The effective noise variance $\tau_0^2$ is obtained by solving the following two equations for $\kappa$ and $\tau$, restricted to the interval $\kappa \in (\kappa_{\min},\infty)$: 
\begin{eqnarray}
\tau^2 &=& \Map(\tau^2,\kappa \tau)\,,\label{eq:Tau1}\\
\lambda &=& \kappa \tau \bigg[1-\frac{1}{\delta} \prob(|\Theta_0 + \tau Z| \ge \kappa\tau) \bigg]\,.\label{eq:Tau2}
\end{eqnarray}
Existence and uniqueness of $\tau_0$ is proved in~\cite[Proposition 1.3]{BayatiMontanariLASSO}.

%%%%%%%%%%%%%%%%%%%%
\section{Tunned regularization parameter $\lambda$ } \label{app:lambda}
In previous appendix, we provided the value of $\tau_0$ for a given regularization parameter $\lambda$.
In this appendix, we discuss the tuned value for $\lambda$ to achieve the power stated in Theorem~\ref{thm:power}. 

%Next we give the explicit formula for the regularization parameter $\lambda$ that achieves the power stated in Theorem~\ref{thm:power}. 

Let $\cF_{\ve} \equiv \{p_{\Theta_0}: \,\,\, p_{\Theta_0}(\{0\}) \ge 1-\eps\}$ be the family of $\ve$-sparse distributions.
Also denote by ${M}(\ve,\kappa)$ the minimax risk of soft thresholding denoiser (at threshold value $\kappa$)
over $\cF_\ve$, i.e.,
\begin{align}
M(\ve,\kappa) = \sup_{p_{\Theta_0}\in \cF_\ve} \, \E\{[\eta(\Theta_0 + Z;\kappa) - \Theta_0]^2\}\,.
\end{align}
The function $M$ can be computed explicitly by evaluating the mean square error on the worst case
$\ve$-sparse distribution. A simple calculation gives 
\begin{align}
{M}(\ve, \kappa) = \ve (1+\kappa^2) + (1-\ve) [2(1+\kappa^2) \Phi(-\kappa) - 2\kappa \phi(\kappa)]\,.
\end{align}
Further, let
\begin{align}
%M(\ve) \equiv \min_{\kappa \in \reals_+} \widetilde{M}(\ve,\kappa)\,,\quad \quad 
\kappa_* (\ve)\equiv \arg \min_{\kappa \in \reals_+} {M}(\ve, \kappa)\,.
\end{align}
In words, $\kappa_*(\ve)$ is the minimax optimal value of threshold $\kappa$ over $\cF_\ve$.
The value of $\lambda$ for Theorem~\ref{thm:power} is then obtained by solving Eq.~\eqref{eq:Tau1} for
$\tau$ with $\kappa = \kappa_*(\ve)$, and then substituting $\kappa_*$ and $\tau$ in Eq.~\eqref{eq:Tau2} to get
 $\lambda = \lambda(p_{\Theta_0}, \sigma, \ve, \delta)$.
 
\begin{remark}\label{rem:lambda}
The theory of \cite{BayatiMontanariLASSO,DMM-NSPT-11} implies that in the standard Gaussian setting and
for a converging sequence of instances $\{(\mtx{\theta}_0(p), n(p), \sigma(p))\}_{p \in \naturals}$, Eq.~\eqref{eq:Tau2} is equivalent to 
the following:
\begin{align}
\lambda \scale = \kappa \tau\,,
\end{align}
where the normalization factor $\scale$ is given by Eq.~\eqref{eq:Step2-std}.  
\end{remark}
%Therefore, $\kappa_*(\ve)$ is the minimax optimal threshold over $F_\ve$.
%A straightforward  calculations leads to the following parametric representation of $M(\ve)$ and $\ve$ in terms of $\kappa_*$.
% %
% \begin{align}
%%
% \ve = \frac{2(\phi(\kappa_*) - \kappa \Phi(-\kappa_*))}{\kappa_* + 2(\phi(\kappa_*) - \kappa_* \Phi(-\kappa_*))},\;\;\;\;\;\;\;\;\;
% M(\ve) = \frac{2\phi(\kappa_*)}{\kappa_* + 2(\phi(\kappa_*) - \kappa_* \Phi(-\kappa_*))}.  
%%
% \end{align}
%%

\section{Statistical power of earlier approaches}
\label{app:Alternative}

In this appendix, we briefly compare our results with those of Zhang
and Zhang \cite{ZhangZhangSignificance}, and B\"uhlmann
\cite{BuhlmannSignificance}. Both of these papers consider
deterministic designs under restricted eigenvalue conditions. As a
consequence, controlling both type I and type II errors requires a
significantly larger value of $\mu/\sigma$.

In~\cite{ZhangZhangSignificance}, authors propose low dimensional projection estimator ($\ldpe\,$) to assess
confidence intervals for the parameters $\theta_{0,j}$. Following the treatment of \cite{ZhangZhangSignificance}, a necessary condition for rejecting $H_{0,j}$ with
non-negligible probability is 
\begin{align}
|\theta_{0,j}|\ge c\tau_j\sigma(1+\epsilon'_n) ,
\end{align}
which follows immediately from
\cite[Eq. (23)]{ZhangZhangSignificance}. Further $\tau_j$ and
$\eps'_n$ are lower bounded in
\cite{ZhangZhangSignificance} as follows
\begin{align}
\tau_j & \ge \frac{1}{\|\mtx{\tx}_j\|_2}\, ,\\
\eps'_n & \ge C \eta^* s_0\sqrt{\frac{\log p}{n}}\, ,
\end{align}
where for a standard Gaussian design $\eta^* \ge \sqrt{\log p}$. Using
further $\|\mtx{\tx}_j\|_2\le 2\sqrt{n}$ which again holds with high
probability for standard Gaussian designs, we get the necessary condition
\begin{align}
|\theta_{0,j}|\ge c'\max\Big\{\frac{\sigma s_0\log p}{n}, \frac{\sigma}{\sqrt{n}}\Big\}\, ,
\end{align}
for some constant $c'$.

In \cite{BuhlmannSignificance}, p-values are defined, in the notation 
of the present paper, as
\begin{align}
P_j \equiv 2\Big\{1-\Phi\big((a_{n,p;j}(\sigma)|\htheta_{j,{\rm corr}}|-\Delta_j)_+\big)\Big\}\, ,
\end{align}
with $\htheta_{j,{\rm corr}}$ a `corrected' estimate of
$\theta_{0,j}$, cf. \cite [Eq. (2.14)]{BuhlmannSignificance}. The corrected estimate $\htheta_{j,{\rm corr}}$
is defined by the following motivation. The ridge estimator bias, in general, can be decomposed into two terms.
The first term is the estimation bias governed by the regularization, and the second term is the additional projection bias
$\proj_{\bX} \mtx{\theta}_0 - \mtx{\theta}_0$, where $\proj_{\bX}$ denotes the orthogonal projector on the row space of $\bX$.
The corrected estimate $\htheta_{j,{\rm corr}}$ is defined in such a way to remove the second bias term under the null hypothesis $H_{0,j}$.
Therefore, neglecting the first bias term, we have $\htheta_{j,{\rm corr}} = (\proj_{\bX})_{jj} \theta_{0,j}$. 

Assuming the corrected estimate to be consistent (which it is in $\ell_1$ sense
under the assumption of the paper), rejecting $H_{0,j}$ with
non-negligible probability requires
\begin{align}
|\theta_{0,j}|\ge \frac{c}{a_{n,p;j}(\sigma)|(\proj_{\bX})_{jj}|} \max\{\Delta_j, 1\}\, ,
\end{align}

Following \cite [Eq. (2.13)]{BuhlmannSignificance} and keeping the dependence on $s_0$ instead of assuming $s_0 =
o((n/\log p)^{\xi})$, we have
\begin{eqnarray}
\frac{\Delta_j}{a_{n,p;j}(\sigma)|(\proj_{\bX})_{jj}|} =
 C\max_{k\in [p]\backslash j} \frac{|(\proj_{\bX})_{jk}|}{|(\proj_{\bX})_{jj}|}\,  \sigma s_0 \sqrt{\frac{\log p}{n}}\,.
\end{eqnarray}
Further, plugging for $a_{n,p;j}$ we have
\begin{eqnarray}
\frac{1}{a_{n,p;j}(\sigma)|(\proj_{\bX})_{jj}|} = \frac{\sigma \sqrt{\Omega_{jj}}}{\sqrt{n}\,|(\proj_{\bX})_{jj}|}\,.
\end{eqnarray}
For a standard Gaussian design $(p/n) (\proj_{\bX})_{jk}$ is approximately
distributed as $u_1$, where $\mtx{u}=(u_1,u_2,\dots,$ $u_n)\in\reals^n$ is a
uniformly random vector with  $\|\mtx{u}\|=1$. In particular $u_1$ is
approximately $\normal(0,1/n)$. A standard calculation yields 
$\max_{k\in [p]\setminus j}|(\proj_{\bX})_{jk}|\ge\sqrt{n\log p}/p$
with high probability. Furthermore, $|(\proj_{\bX})_{jj}|$ concentrates
around $n/p$. Finally, by definition of $\Omega_{jj}$ (cf.~\cite [Eq. (2.3)]{BuhlmannSignificance})
and using classical large deviation results about the singular values of a Gaussian matrix, we 
have $\Omega_{jj} \ge (n/p)^2$ with high probability.
Hence, a necessary condition for rejecting $H_{0,j}$ with
non-negligible probability is
\begin{align}
|\theta_{0,j}|\ge C \max\Big\{ \frac{\sigma s_0 \log p}{n},\frac{\sigma}{\sqrt{n}}\Big\} ,
\end{align}
as stated in Section \ref{sec:Introduction}.

\section{Replica method calculation}
\label{app:replica}

In this section we outline the replica calculation leading to the
Claim \ref{claim:Replica}.  Indeed we consider an even more general
setting, whereby the $\ell_1$ regularization is
  replaced by an arbitrary separable penalty.
Namely, instead of the Lasso, we consider regularized  least squares estimators of the form
\begin{eqnarray}
\mtx{\htheta}(\by,\bX) = \arg\min_{\mtx{\theta}\in\reals^{p}}  \label{eqn:general_opt}
\Big\{
\frac{1}{2n} \|\by - \bX\mtx{\theta}\|^2 + J(\mtx{\theta})\Big\}\, ,
\end{eqnarray}
with $J(\mtx{\theta})$ being a convex separable penalty function; namely 
for a vector $\mtx{\theta} \in \reals^p$, we have $J(\mtx{\theta}) = J_1(\theta_1) + \cdots +
J_p(\theta_p)$, where $J_\ell : \reals \to \reals$ is a
convex function. Important instances from this ensemble of estimators
are Ridge-regression ($J(\mtx{\theta}) = \lambda\|\mtx{\theta}\|^2/2$), 
and the Lasso  ($J(\mtx{\theta}) =\lambda \|\mtx{\theta}\|_1$).
The Replica Claim \ref{claim:Replica} is generalized to the present
setting replacing $\lambda \|\mtx{\theta}\|_1$ by $J(\mtx{\theta})$.
The only required modification concerns the
definition of the factor $\scale$. 
We let $\scale$ be the unique positive  solution of the following equation
\vspace{-0.15cm}
\begin{eqnarray}
1 = \frac{1}{\scale}+\frac{1}{n}\, \Tr\Big\{(\id+\scale
\mtx{\Sigma}^{-1/2}\nabla^2 J(\mtx{\htheta}) \mtx{\Sigma}^{-1/2})^{-1}\Big\},
\label{eq:GeneralOnsager}
\end{eqnarray}
where $\nabla^2J(\mtx{\htheta})$ denotes the Hessian, which is diagonal
since $J$ is separable. If $J$ is non differentiable, then we formally
set $[\nabla^2 J(\mtx{\htheta})]_{ii} =\infty$ for all the coordinates $i$
such that $J$ is non-differentiable at $\htheta_i$. It can be
checked that this definition is well posed and that yields the
previous choice for $J(\mtx{\theta}) = \lambda\|\mtx{\theta}\|_1$.

\vspace{0.5cm}

We pass next to establishing the claim.
We limit ourselves to the main steps, since 
analogous calculations can be found in several earlier works
\cite{TanakaCDMA,GuoVerdu,TakedaKabashima}.
For a general introduction to the method and its motivation we refer
to \cite{SpinGlass,MezardMontanari}.
Also, for the sake of simplicity, we shall focus on characterizing the
asymptotic distribution of $\mtx{\htheta}^u$,
cf. Eq.~(\ref{eq:ThetauDef}). The distribution of $r$ is derived
by the same approach.

Fix a sequence of instances
$\{(\mtx{\Sigma}(p),\mtx{\theta}_0(p),n(p),\sigma(p))\}_{p\in\naturals}$. For the
sake of simplicity, we assume $\sigma(p)^2= n(p)\sigma_0^2$ and $n(p) = p\delta$
(the slightly more general case $\sigma(p)^2= n(p)[\sigma_0^2+o(1)]$
and $n(p) = p[\delta+o(1)]$ does not require any change to the
derivation given here, but is more cumbersome notationally).
Fix $\tg:\reals\times\reals\times\reals\to\reals$
a continuous function convex in its first argument, and 
let $g(u,y,z) \equiv \max_{x\in\reals} [ux-\tg(x,y,z)]$ be its
Lagrange dual.
The replica calculation aims at estimating the following
moment generating function (\emph{partition function})
\begin{align}
\cZ_p(\beta,s) \equiv \int \!\exp\Big\{-\frac{\beta}{2n}
\|\by-\bX\mtx{\theta}\|_2^2-\beta  J(\mtx{\theta})&-\beta s\sum_{i=1}^p
[g(u_i,\theta_{0,i},(\mtx{\Sigma}^{-1})_{ii})-u_i\htheta^u_i]\nonumber\\
&-\frac{\beta}{2n}(s\tscale)^2
\|\bX\mtx{\Sigma}^{-1}\mtx{u}\|^2_2\Big\}
\de \mtx{\theta}\,\de \mtx{u}\, .
\label{eq:PartitionFunction}
\end{align}
Here $(y_i,\mtx{x}_i)$ are i.i.d. pairs distributed as per model
(\ref{eqn:regression})
and $\mtx{\htheta}^u = \mtx{\theta} + (\tscale/n) \, \mtx{\Sigma}^{-1} \bX^\sT (\by-\bX \mtx{\theta})$ with $\tscale \in \reals$
to be defined below.
Further, $g:\reals\times\reals\times\reals\to\reals$
is a continuous function strictly convex in its first argument. Finally, $s\in\reals_+$ and $\beta>0$ is a
`temperature' parameter not to be confused with the type II error rate
as used in the main text.
%, and $\tscale\in\reals$ is also a  parameter to be fixed below.
 We will eventually show that the
appropriate choice of $\tscale$ is given by Eq.~(\ref{eq:GeneralOnsager}).

Within the replica method, it is
assumed that the limits $p\to\infty$, $\beta\to\infty$ exist  almost
surely for the quantity $(p\beta)^{-1}\log \cZ_p(\beta,s)$, and that the
order of the limits can be exchanged.
We therefore define
\begin{eqnarray}
\Free(s) &\equiv &-\lim_{\beta\to \infty}\lim_{p\to \infty} \frac{1}{p\beta}\,
\log \cZ_p(\beta,s)\\
&\equiv &-\lim_{p\to \infty}\lim_{\beta\to \infty} \frac{1}{p\beta}\,
\log \cZ_p(\beta,s)\, . \label{eq:FirstBeta}
\end{eqnarray}
In other words $\Free(s)$ is the exponential growth rate of
$\cZ_p(\beta,s)$. It is also assumed that $p^{-1}\log \cZ_p(\beta,s)$
concentrates tightly around its expectation so that $\Free(s)$ can in
fact be evaluated by computing
\begin{eqnarray}
\Free(s) = -\lim_{\beta\to \infty}\lim_{p\to \infty}\frac{1}{p\beta}\, \E
\log \cZ_p(\beta,s)\, ,\label{eq:FreeExp}
\end{eqnarray}
where expectation is being taken with respect to the distribution of
$(y_1,\mtx{x}_1), \cdots, (y_n,\mtx{x}_n)$. Notice that, by
Eq.~(\ref{eq:FirstBeta}) and using Laplace method in the integral
(\ref{eq:PartitionFunction}), we have
\begin{align}
&\Free(s) = \nonumber\\
&\lim_{p\to \infty} \frac{1}{p}\min_{\mtx{\theta},\mtx{u}\in\reals^p}\Big\{\frac{1}{2n}
\|\by-\bX\mtx{\theta}\|_2^2+J(\mtx{\theta})+s\sum_{i=1}^p [g(u_i,\theta_{0,i},(\mtx{\Sigma}^{-1})_{ii})-u_i\htheta^u_i]+\frac{1}{2n}(s\tscale)^2
\|\bX\mtx{\Sigma}^{-1}\mtx{u}\|^2_2\Big\}. 
\end{align}
Finally we assume that the derivative of $\Free(s)$ as $s\to 0$ can be
obtained by differentiating inside the limit. This condition holds,
for instance, if the cost function is strongly convex at $s=0$. We 
get 
\begin{eqnarray}
\frac{\de \Free}{\de s}(s=0) =\lim_{p\to \infty}
\frac{1}{p}\sum_{i=1}^p
\min_{u_i\in\reals}[g(u_i,\theta_{0,i},(\mtx{\Sigma}^{-1})_{ii})-u_i\htheta^u_i]
\end{eqnarray}
where $\mtx{\htheta}^u = \mtx{\htheta} + (\tscale/n) \, \mtx{\Sigma}^{-1} \bX^\sT (\by-\bX
\mtx{\htheta})$ and $\mtx{\htheta}$ is the minimizer of the regularized least
squares as per Eq.~(\ref{eqn:Lasso_cost}).
Since, by duality $\tg(x,y,z) \equiv \max_{u\in\reals} [ux-g(u,y,z)]$,
we get
\begin{eqnarray}
\frac{\de \Free}{\de s}(s=0) &
= -\lim_{p\to \infty}
\frac{1}{p}\sum_{i=1}^p
\tg(\htheta^u_i,\theta_{0,i},(\mtx{\Sigma}^{-1})_{ii})\, .  \label{eq:LimEmp}
\end{eqnarray}
Hence, by computing  $\Free(s)$ using Eq.~(\ref{eq:FreeExp}) for a
complete set of functions $\tg$, we get access to the corresponding
limit quantities (\ref{eq:LimEmp}) and hence, via standard weak
convergence arguments, to the joint empirical distribution of the
triple $(\htheta^u_i,\theta_{0,i},(\mtx{\Sigma}^{-1})_{ii})$,
cf. Eq.~(\ref{eq:EmpDef}).

In order to carry out the calculation of $\Free(s)$, we begin by
rewriting the partition function (\ref{eq:PartitionFunction}) in a
more convenient form. Using the definition of $\mtx{\htheta}^u$ and after a
simple manipulation 
\begin{align}
&\cZ_p(\beta,s) = \nonumber\\
&\int \!\exp\Big\{-\frac{\beta}{2n}
\|\by-\bX(\mtx{\theta}+s\tscale\mtx{\Sigma}^{-1}\mtx{u})\|_2^2-\beta  J(\mtx{\theta})+\beta s\<\mtx{u},\mtx{\theta}\>-\beta\, s\sum_{i=1}^p
g(u_i,\theta_{0,i},(\mtx{\Sigma}^{-1})_{ii})\Big\}
\de \mtx{\theta}\,\de \mtx{u}\, .\label{eq:PartFun2}
\end{align}
Define the measure $\nu(\de \mtx{\theta})$ over $\mtx{\theta}\in\reals^p$ as
follows
\begin{align}
\nu(\de\mtx{\theta}) = \int \exp\Big\{-\beta J(\mtx{\theta}-s\tscale\mtx{\Sigma}^{-1}\mtx{u})+\beta s\<\mtx{\theta}-s\tscale\mtx{\Sigma}^{-1}\mtx{u},\mtx{u}\>-\beta s\sum_{i=1}^p
g(u_i,\theta_{0,i},(\mtx{\Sigma}^{-1})_{ii})\Big\}\,
\de \mtx{u}\, .
\end{align}
Using this definition and with the change of variable $\mtx{\theta}' =
\mtx{\theta}+s\tscale\mtx{\Sigma}^{-1}\mtx{u}$, we can rewrite Eq.~(\ref{eq:PartFun2}) as
\begin{align}
\cZ_p(\beta,s) &\equiv \int \!\exp\Big\{-\frac{\beta}{2n}
\|\by-\bX\mtx{\theta}\|_2^2\Big\}\, \nu(\de\mtx{\theta})\nonumber \\
&=\int \!\exp\Big\{i\sqrt{\frac{\beta}{n}}
\<\mtx{z},\by-\bX\mtx{\theta}\>\Big\}\,\nu(\de\mtx{\theta})\, \gamma_n(\de \mtx{z})\nonumber\\
&=\int \!\exp\Big\{i\sqrt{\frac{\beta}{n}}
\<\mtx{w},\mtx{z}\>+i\sqrt{\frac{\beta}{n}}\<\mtx{z},\bX(\mtx{\theta}_0-\mtx{\theta})\>\Big\}\,\nu(\de\mtx{\theta})\,
\gamma(\de \mtx{z})
\label{PartFunRepresentation}
\, ,
\end{align}
where $\gamma_n(\de \mtx{z})$ denotes the standard Gaussian measure on
$\reals^n$: $\gamma_n(\de \mtx{z}) \equiv (2\pi)^{-n/2}\exp(-\|\mtx{z}\|^2_2/2)\, \de \mtx{z}$.

The replica method aims at computing the expected log-partition
function, cf. Eq.~(\ref{eq:FreeExp}) using the identity
\begin{eqnarray}
\E\log\cZ_p(\beta,s) = \left.\frac{\de\phantom{k}}{\de
  k}\right|_{k=0}\log\E\big\{\cZ_p(\beta,s)^k\big\}\, .
\end{eqnarray}
This formula would require computing fractional moments of $\cZ_p$ as
$k\to 0$. The replica method consists in a prescription that allows to
compute a formal expression for the $k$ integer, and then
extrapolate it as $k\to 0$. Crucially, the limit $k\to 0$ is inverted
with the one $p\to\infty$:
\begin{eqnarray}
\lim_{p\to\infty}\frac{1}{p}\E\log\cZ_p(\beta,s) = \left.\frac{\de\phantom{k}}{\de
  k}\right|_{k=0}\lim_{p\to\infty}\frac{1}{p}\log\E\big\{\cZ_p(\beta,s)^k\big\}\, .\label{eq:ReplicaFormula}
\end{eqnarray}

In order to represent $\cZ_p(\beta,s)^k$, we use the identity
\begin{eqnarray}
\Big(\int f(\mtx{x})\,\rho(\de \mtx{x})\Big)^k= \int f(\mtx{x}^1)f(\mtx{x}^2)\cdots f(\mtx{x}^k)\, \rho(\de
\mtx{x}^1)\cdots\rho(\de \mtx{x}^k)\, .
\end{eqnarray}
In order to apply this formula to Eq.~(\ref{PartFunRepresentation}),
we let, with a slight abuse of notation, 
$\nu^k(\de\mtx{\theta})\equiv \nu(\de\mtx{\theta}^1)\times\nu(\de\mtx{\theta}^2)\times\cdots
\times \nu(\de\mtx{\theta}^k)$ be a measure over $(\reals^p)^k$, with
$\mtx{\theta}^1,\dots,\mtx{\theta}^k\in\reals^p$.
Analogously $\gamma_n^k(\de \mtx{z})\equiv \gamma_n(\de \mtx{z}^1)\times\gamma_n(\de \mtx{z}^2)\times\cdots
\times \gamma_n(\de \mtx{z}^k)$, with
$\mtx{z}^1,\dots,\mtx{z}^k\in\reals^n$.
With these notations, we have
\begin{align}
\E\{\cZ_p(\beta,s)^k\}=\int \!\E\exp\Big\{i\sqrt{\frac{\beta}{n}}
\<\mtx{w},\sum_{a=1}^k \mtx{z}^a\>+i\sqrt{\frac{\beta}{n}}\<\bX,\sum_{a=1}^k \mtx{z}^a(\mtx{\theta}_0-\mtx{\theta}^a)^{\sT}\>\Big\}\,\nu^k(\de\mtx{\theta})\,
\gamma_n^k(\de \mtx{z})\, .\label{eq:ExpectedPartFun}
\end{align}
In the above expression $\E$ denotes expectation with respect to the
noise vector $\mtx{w}$, and the design matrix $\bX$. Further, we used
$\<\,\cdot\,,\,\cdot\,\>$ to denote matrix scalar product as well: $\<\mtx{A},\mtx{B}\>
\equiv \Tr(\mtx{A}^{\sT}\mtx{B})$.

At this point we can take the expectation with respect to $\mtx{w}$,
$\bX$. We use the fact that, for any $\mtx{M}\in\reals^{n\times p}$,
$\mtx{u}\in\reals^n$
\begin{eqnarray}
\begin{split}
\E\big\{\exp\big( i\<\mtx{w},\mtx{u}\>\big)\big\} & = &
\exp\Big\{-\frac{1}{2}n\sigma^2_0\, \|\mtx{u}\|^2_2\Big\}\, ,\\
\E\big\{\exp\big( i\<\mtx{M},\bX\>\big)\big\} & =
&\exp\Big\{-\frac{1}{2}\,\<\mtx{M}, \mtx{M} \mtx{\Sigma} \>\Big\}\, ,
\end{split}
\end{eqnarray}
Using these identities in Eq.~(\ref{eq:ExpectedPartFun}), we obtain
\begin{align}
&\E\{\cZ_p^k\}= \nonumber\\
&\int\!\exp\Big\{-\frac{1}{2}\,\beta\sigma^2_0\sum_{a=1}^k
\|\mtx{z}^a\|_2^2-\frac{\beta}{2n}
\sum_{a,b=1}^k \<\mtx{z}^a,\mtx{z}^b\>\,\<(\mtx{\theta}^a-\mtx{\theta}_0),\mtx{\Sigma}(\mtx{\theta}^b-\mtx{\theta}_0)\>\Big\}\,\nu^k(\de\mtx{\theta})\,
\gamma_n^k(\de \mtx{z})\, .\label{eq:ExpZ}
\end{align}
We next use the identity
\begin{align}
e^{-xy} = \frac{1}{2\pi i}\int_{(-i\infty,i\infty)}\int_{(-\infty,\infty)} e^{-\zeta q+\zeta x-qy} \,
\de\zeta\, \de q\, ,
\end{align}
where the integral is over $\zeta\in (-i\infty,i\infty)$ (imaginary
axis) and $q\in (-\infty,\infty)$. We apply this  identity to
Eq.~(\ref{eq:ExpZ}), and introduce integration variables $\mtx{Q} \equiv
(Q_{ab})_{1\le a,b\le k}$ and $\mtx{\Lambda} \equiv(\Lambda_{ab})_{1\le
  a,b\le k}$.
Letting $\de \mtx{Q}\equiv \prod_{a,b}\de Q_{ab}$ and  $\de \mtx{\Lambda} \equiv \prod_{a,b}\de \Lambda_{ab}$ 
\begin{align}
\E\{\cZ_p^k\} &= \Big(\frac{\beta n}{4\pi i}\Big)^{k^2}\int
\exp\Big\{-p\cS_k(\mtx{Q},\mtx{\Lambda})\Big\}\,\de \mtx{Q}\, \de \mtx{\Lambda}\, ,\label{eq:IntegralDq}\\
\cS_k(\mtx{Q},\mtx{\Lambda})&=\frac{\beta\delta}{2}\sum_{a,b=1}^k\Lambda_{ab}Q_{ab}
-\frac{1}{p}\,
\log\xi(\mtx{\Lambda}) - \delta\,\log \hxi(\mtx{Q})\, ,\label{eq:Action}\\
\xi(\mtx{\Lambda}) & \equiv \int\!
\exp\Big\{\frac{\beta}{2}\sum_{a,b=1}^k\Lambda_{ab}\<(\mtx{\theta}^a-\mtx{\theta}_0),\mtx{\Sigma}(\mtx{\theta}^b-\mtx{\theta}_0)\>\Big\}\,
\nu^k(\de\mtx{\theta})\, ,\\
\hxi(\mtx{Q}) &\equiv \int\!
\exp\Big\{-\frac{\beta}{2}\sum_{a,b=1}^k(\sigma_0^2\id+\mtx{Q})_{a,b}\, z_1^a\, z_1^b\Big\}\,\gamma^k_1(\de
{z}_1)\, .
\end{align}
Notice that above we used the fact that, after introducing
$\mtx{Q},\mtx{\Lambda}$, the integral over $(\mtx{z}^1,\dotsc, \mtx{z}^k)\in (\reals^n)^k$ factors into $n$
integrals over $(\reals)^k$ with measure $\gamma^k_1(\de z_1)$.

We next use the saddle point method in Eq.~(\ref{eq:IntegralDq}) to
obtain 
\begin{eqnarray}
-\lim_{p\to\infty}\frac{1}{p}\log\E\{\cZ_p^k\} &= \cS_k(\mtx{Q}^*,\mtx{\Lambda}^*)\, ,
\end{eqnarray}
where $\mtx{Q}^*$, $\mtx{\Lambda}^*$ is the saddle-point location. The replica
method provides a hierarchy of ansatz for this saddle-point. The
first level of this hierarchy is the so-called \emph{replica
  symmetric} ansatz postulating that $\mtx{Q}^*$, $\mtx{\Lambda}^*$ ought to
be invariant under permutations of the row/column indices. This is motivated by
the fact that $\cS_k(\mtx{Q},\mtx{\Lambda})$ is indeed left unchanged by such
change of variables. This is equivalent to postulating that
\begin{eqnarray}
Q^*_{ab} =
\begin{cases}
q_1 & \mbox{ if } a=b,\\
q_0 & \mbox{ otherwise,}
\end{cases}\,,\;\;\;\;\;\;\;\;
\Lambda^*_{ab} =
\begin{cases}
\beta\zeta_1 & \mbox{ if } a=b,\\
\beta\zeta_0 & \mbox{ otherwise,}
\end{cases}
\end{eqnarray}
where the factor $\beta$ is for future convenience.
Given that the partition function,
cf. Eq.~(\ref{eq:PartitionFunction}) is the integral of a log-concave
function, it is expected that the replica-symmetric ansatz yields in
fact the correct result
\cite{SpinGlass,MezardMontanari}.

The next step consists in substituting  the above expressions for
$\mtx{Q}^*$, $\mtx{\Lambda}^*$ in $\cS_k(\,\cdot\,,\,\cdot\,)$ and then taking the
limit $k\to 0$. We will consider separately each term of
$\cS_k(\mtx{Q},\mtx{\Lambda})$, cf.~Eq.~(\ref{eq:Action}). 

Let us begin with the first term
\begin{eqnarray}
\sum_{a,b=1}^k\Lambda^*_{ab}Q^*_{ab} = k\,
\beta\zeta_1q_1+k(k-1)\beta\zeta_0q_0\, .
\end{eqnarray}
Hence
\begin{eqnarray}
\lim_{k\to\infty}
\frac{\beta\delta}{2k}\sum_{a,b=1}^k\Lambda^*_{ab}Q^*_{ab}=
\frac{\beta^2\delta}{2}\,(\zeta_1q_1-\zeta_0q_0)\, .\label{eq:LambdaQ}
\end{eqnarray}

Let us consider  $\hxi(\mtx{Q}^*)$. We have
\begin{align}
\log \hxi(\mtx{Q}^*) &= -\frac{1}{2}\log \Det(\id+\beta\sigma^2 \id+\beta \mtx{Q}^*) \\
&=
-\frac{k-1}{2}\log\big(1+\beta(q_1-q_0)\big)-\frac{1}{2}\log\big(1+\beta(q_1-q_0)+\beta
k(\sigma^2+q_0)\big)\, .
\end{align}
In the limit $k\to 0$ we thus obtain
\begin{align}
\lim_{k\to 0}\frac{1}{k}(-\delta)\log \hxi(\mtx{Q}^*) &=
\frac{\delta}{2}\,\log\big(1+\beta(q_1-q_0)\big)+\frac{\delta}{2}
\,\frac{\beta(\sigma^2+q_0)}{1+\beta(q_1-q_0)}\, . \label{eq:Xihat}
\end{align}

Finally, introducing the notation $\|\mtx{v}\|_{\mtx{\Sigma}}^2 \equiv \<\mtx{v},\mtx{\Sigma}
\mtx{v}\>$, we have
\begin{align}
\xi(\mtx{\Lambda}^*) & \equiv \int\!
\exp\Big\{\frac{\beta^2}{2} (\zeta_1-\zeta_0)\sum_{a=1}^k\|\mtx{\theta}^a-\mtx{\theta}_0\|_{\mtx{\Sigma}}^2
+\frac{\beta^2\zeta_0}{2} \sum_{a,b=1}^k\< (\mtx{\theta}^a-\mtx{\theta}_0),\mtx{\Sigma} (\mtx{\theta}^b-\mtx{\theta}_0)\>\Big\}\,
\nu^k(\de\mtx{\theta})\, ,\nonumber\\
&= \E\int\!
\exp\Big\{\frac{\beta^2}{2} (\zeta_1-\zeta_0)\sum_{a=1}^k\|\mtx{\theta}^a-\mtx{\theta}_0\|_{\mtx{\Sigma}}^2
+\beta\sqrt{\zeta_0} \sum_{a=1}^k\< \mtx{z},\mtx{\Sigma} ^{1/2}(\mtx{\theta}^a-\mtx{\theta}_0)\>\Big\}\,
\nu^k(\de\mtx{\theta})\, ,
\end{align}
where expectation is with respect to $\mtx{z}\sim\normal(0,\id_{p\times
  p})$. Notice that, given $\mtx{z}\in\reals^p$, the integrals over
$\mtx{\theta}^1,\mtx{\theta}^2,\dots,\mtx{\theta}^k$ factorize, whence
\begin{align}
\xi(\mtx{\Lambda}^*) &= \E \left\{\left[\int\!
\exp\Big\{\frac{\beta^2}{2} (\zeta_1-\zeta_0)\|\mtx{\theta}-\mtx{\theta}_0\|_{\mtx{\Sigma}}^2
+\beta\sqrt{\zeta_0} \< \mtx{z},\mtx{\Sigma} ^{1/2}(\mtx{\theta}-\mtx{\theta}_0)\>\Big\}\,
\nu(\de\mtx{\theta})\right]^k\right\}\, .
\end{align}
Therefore
\begin{align}
\lim_{k\to 0}\frac{(-1)}{pk}&\log \xi(\mtx{\Lambda}^*) =\nonumber \\
&-\frac{1}{p}\E \left\{\log\left[\int\!
\exp\Big\{\frac{\beta^2}{2} (\zeta_1-\zeta_0)\|\mtx{\theta}-\mtx{\theta}_0\|_{\mtx{\Sigma}}^2
+\beta\sqrt{\zeta_0} \< \mtx{z},\mtx{\Sigma} ^{1/2}(\mtx{\theta}-\mtx{\theta}_0)\>\Big\}\,
\nu(\de\mtx{\theta})\right]\right\}\, . \label{eq:Xi} 
\end{align}
Putting Eqs.~(\ref{eq:LambdaQ}), (\ref{eq:Xihat}), and (\ref{eq:Xi})
together we obtain
\begin{align}
-\lim_{p\to\infty}\frac{1}{p\beta}\E\log\cZ_p &= \lim_{k\to
  0}\frac{1}{k\beta}\cS_k(\mtx{Q}^*,\mtx{\Lambda}^*)\nonumber\\
&= \frac{\beta\delta}{2}\,(\zeta_1q_1-\zeta_0q_0)
+\frac{\delta}{2\beta}\,\log\big(1+\beta(q_1-q_0)\big)+\frac{\delta}{2}
\,\frac{\sigma^2+q_0}{1+\beta(q_1-q_0)}\nonumber\\
&-\lim_{p\to \infty}\frac{1}{p\beta}\E \left\{\log\left[\int\!
\exp\Big\{\frac{\beta^2}{2} (\zeta_1-\zeta_0)\|\mtx{\theta}-\mtx{\theta}_0\|_{\mtx{\Sigma}}^2\right.\right.\nonumber\\
&\left.\left.\phantom{AAAAAAA\int}+\beta\sqrt{\zeta_0} \< \mtx{z},\mtx{\Sigma} ^{1/2}(\mtx{\theta}-\mtx{\theta}_0)\>\Big\}\,
\nu(\de\mtx{\theta})\right]\right\}\, .
\end{align}

We can next take the limit $\beta\to\infty$. In doing this, one has to
be careful with respect to the behavior of the saddle point parameters
$q_0,q_1$, $\zeta_0,\zeta_1$. A careful analysis (omitted here)
shows that $q_0, q_1$ have the same limit, denoted here by $q_0$, and  
$\zeta_0, \zeta_1$ have the same limit, denoted by $\zeta_0$.
Moreover $q_1-q_0= (q/\beta)+o(\beta^{-1})$ and 
$\zeta_1-\zeta_0= (-\zeta/\beta)+o(\beta^{-1})$. Substituting in
the above expression, and using Eq.~(\ref{eq:FreeExp}), we get
\begin{align}
\Free(s) =  & \frac{\delta}{2}(\zeta_0q-\zeta q_0) +\frac{\delta}{2}\, \frac{q_0+\sigma^2}{1+q} \nonumber\\
&+\lim_{p\to\infty}\frac{1}{p}
\E\min_{\mtx{\theta}\in\reals^p}\Big\{\frac{\zeta}{2}\|\mtx{\theta}-\mtx{\theta}_0\|_{\mtx{\Sigma}}^2-\sqrt{\zeta_0}\<\mtx{z},\mtx{\Sigma}^{1/2}(\mtx{\theta}-\mtx{\theta}_0)\>+
\ttJ(\mtx{\theta};s)\Big\}\, ,\\
\ttJ(\mtx{\theta};s) =  & \min_{\mtx{u}\in\reals^p}\Big\{ J(\mtx{\theta}-s\tscale\mtx{\Sigma}^{-1}\mtx{u})-s\<\mtx{\theta}-s\tscale\mtx{\Sigma}^{-1}\mtx{u},\mtx{u}\>+ s\sum_{i=1}^p
g(u_i,\theta_{0,i},(\mtx{\Sigma}^{-1})_{ii})\Big\}\, .
\end{align}
After the change of variable $\mtx{\theta}-s\tscale\mtx{\Sigma}^{-1}\mtx{u}\to\mtx{\theta}$, this reads
\begin{eqnarray}
\Free(s) &= &\frac{\delta}{2}(\zeta_0q-\zeta q_0) +\frac{\delta}{2}\, \frac{q_0+\sigma_0^2}{1+q}-\frac{\zeta_0}{2\zeta}\nonumber\\
&&+\lim_{p\to\infty}\frac{1}{p}
\E\min_{\mtx{\theta},\mtx{u}\in\reals^p}\Big\{\frac{\zeta}{2}\Big\|\mtx{\theta}-\mtx{\theta}_0-\frac{\sqrt{\zeta_0}}{\zeta}\mtx{\Sigma}^{-1/2}\mtx{z}+s\tscale\mtx{\Sigma}^{-1}\mtx{u}\Big\|_{\mtx{\Sigma}}^2+
\tJ(\mtx{\theta},\mtx{u};s)\Big\}\, ,\\
\tJ(\mtx{\theta},\mtx{u};s)    &=&J(\mtx{\theta})-s\<\mtx{\theta},\mtx{u}\>+ s\sum_{i=1}^p
g(u_i,\theta_{0,i},(\mtx{\Sigma}^{-1})_{ii})\, .
\end{eqnarray}
Finally, we must set $\zeta,\zeta_0$ and  $q,q_0$ to their saddle
point values. We start by using the stationarity conditions with
respect to $q$, $q_0$:
\begin{eqnarray}
\frac{\partial\Free}{\partial q}(s) &=&
\frac{\delta}{2}\zeta_0-\frac{\delta}{2}\,
\frac{q_0+\sigma_0^2}{(1+q)^2}\, ,\\
\frac{\partial\Free}{\partial q_0}(s) &=&-\frac{\delta}{2}\,\zeta
+\frac{\delta}{2}\,\frac{1}{1+q}\, .
\end{eqnarray}
We use these to eliminate $q$ and $q_0$. Renaming $\zeta_0 =
\zeta^2\tau^2$, we get our final expression for $\Free(s)$:
\begin{eqnarray}
\Free(s) &=
&-\frac{1}{2}(1-\delta)\zeta\tau^2-\frac{\delta}{2}\zeta^2\tau^2+\frac{\delta}{2}\sigma_0^2\zeta
\label{eq:FreeFinal}\nonumber\\
&&+\lim_{p\to\infty}\frac{1}{p}
\E\min_{\mtx{\theta},\mtx{u}\in\reals^p}\Big\{\frac{\zeta}{2}\Big\|\mtx{\theta}-\mtx{\theta}_0-\tau\mtx{\Sigma}^{-1/2}\mtx{z}+s\tscale\mtx{\Sigma}^{-1}\mtx{u}\Big\|_{\mtx{\Sigma}}^2+
\tJ(\mtx{\theta},\mtx{u};s)\Big\}\, ,\\
\tJ(\mtx{\theta},\mtx{u};s)    &=&J(\mtx{\theta})-s\<\mtx{\theta},\mtx{u}\>+ s\sum_{i=1}^p
g(u_i,\theta_{0,i},(\mtx{\Sigma}^{-1})_{ii})\, .
\end{eqnarray}
Here it is understood that $\zeta$ and $\tau^2$ are to be set to their
saddle point values.

We are interested in the derivative of $\Free(s)$ with respect to
$s$, cf. Eq.~(\ref{eq:LimEmp}).
Consider first the case $s=0$. Using the assumption $\Energy^{(p)}(a,b) \to\Energy(a,b)$,
cf. Eq.~(\ref{eq:ReplicaAssumption}), we get
\begin{eqnarray}
\Free(s=0) = -\frac{1}{2}(1-\delta)\zeta\tau^2-\frac{\delta}{2}\zeta^2\tau^2+\frac{\delta}{2}\sigma_0^2\zeta
+\Energy(\tau^2,\zeta)\,  .
\end{eqnarray}
The values of $\zeta$, $\tau^2$ are obtained by setting to zero the
partial derivatives
\begin{eqnarray}
\frac{\partial\Free}{\partial\zeta}(s=0) &=&
-\frac{1}{2}(1-\delta)\tau^2-\delta\zeta\tau^2+\frac{\delta}{2}\sigma_0^2
+\frac{\partial\Energy}{\partial\zeta}(\tau^2,\zeta)\, ,\label{eq:DFree1}\\
\frac{\partial\Free}{\partial\tau^2}(s=0)
&=&-\frac{1}{2}(1-\delta)\zeta-\frac{\delta}{2}\,\zeta^2 +
\frac{\partial\Energy}{\partial\tau^2}(\tau^2,\zeta)\, ,\label{eq:DFree2}
\end{eqnarray}
Define, as in the statement of the Replica Claim
\begin{align}
\Ena(a,b) 
&\equiv\lim_{p\to\infty}\frac{1}{p}\E\big\{\big\|\eta_{b}(\mtx{\theta}_0+\sqrt{a}\mtx{\Sigma}^{-1/2}\mtx{z})-\mtx{\theta}_0\big\|_{\mtx{\Sigma}}^2\big\}\,
,\\
\Enb(a,b)
&\equiv\lim_{p\to\infty}\frac{1}{p}\E\big\{\div\eta_{b}(\mtx{\theta}_0+\sqrt{a}\mtx{\Sigma}^{-1/2}\mtx{z})\big\}\nonumber\\
&=\lim_{p\to\infty}\frac{1}{p\tau}\E\big\{\<\eta_{b}(\mtx{\theta}_0+\sqrt{a}\mtx{\Sigma}^{-1/2}\mtx{z}),\mtx{\Sigma}^{1/2}\mtx{z}\>\big\}\, ,
\end{align}
where the last identity follows by integration by parts. These limits
exist by the assumption that
$\nabla\Energy^{(p)}(a,b)\to\nabla\Energy(a,b)$. In particular
\begin{align}
\frac{\partial\Energy}{\partial\zeta}(\tau^2,\zeta)& =
\frac{1}{2}\Ena (\tau^2,\zeta)-\tau^2\,\Enb (\tau^2,\zeta)+\frac{1}{2}\,\tau^2\, ,\\
\frac{\partial\Energy}{\partial\tau^2}(\tau^2,\zeta)& =
-\frac{\zeta}{2}\Enb (\tau^2,\zeta)+\frac{1}{2}\,\zeta\, .
\end{align}
Substituting these expressions in Eqs.~(\ref{eq:DFree1}),
(\ref{eq:DFree2}), and simplifying, we conclude that the derivatives
vanish if and only if $\zeta, \tau^2$ satisfy the following
equations
\begin{eqnarray}
\tau^2 &=& \sigma_0^2 +\frac{1}{\delta}\, \Ena(\tau^2,\zeta)\,
,\label{eq:LambdaTau1}\\
\zeta & = & 1-\frac{1}{\delta}\,\Enb(\tau^2,\zeta)\, . \label{eq:LambdaTau2}
\end{eqnarray}
The solution of these equations is expected to be unique for $J$ convex and $\sigma_0^2>0$.

Next consider the derivative of $\Free(s)$ with respect to $s$, which
is our main object of interest, cf. Eq.~(\ref{eq:LimEmp}).
By differentiating Eq.~(\ref{eq:FreeFinal}) and inverting the order of derivative
and limit, we get
\begin{eqnarray}
\frac{\de \Free}{\de s}(s=0) = \lim_{p\to\infty}
\frac{1}{p}\E\min_{\mtx{u}\in\reals^p}\Big\{\zeta\tscale\<\mtx{u},\mtx{\htheta}-\mtx{\theta}_0-\tau\mtx{\Sigma}^{-1/2}\mtx{z}\>-
\<\mtx{\htheta},\mtx{u}\>+\sum_{i=1}^pg(u_i,\theta_{0,i},(\mtx{\Sigma}^{-1})_{ii})\Big\}\, ,
\end{eqnarray}
where $\mtx{\htheta}$ is the minimizer at $s=0$, i.e.,
$\mtx{\htheta} = \eta_{\zeta}(\mtx{\theta}_0+\tau\mtx{\Sigma}^{-1/2}\mtx{z})$, and
$\zeta,\tau^2$ solve Eqs.~(\ref{eq:LambdaTau1}),
(\ref{eq:LambdaTau2}). At this point we choose $\tscale =
1/\zeta$. Minimizing over $\mtx{u}$ (recall that $\tg(x,y,z) =
\max_{u\in\reals}[ux-g(u,y,z)]$), we get
\begin{eqnarray}
\frac{\de \Free}{\de s}(s=0) = -\lim_{p\to\infty}
\frac{1}{p}\E\,
\tg(\theta_{0,i}+\tau(\mtx{\Sigma}^{-1/2}\mtx{z})_i,\theta_{0,i},(\mtx{\Sigma}^{-1})_{ii})\, .
\end{eqnarray}
Comparing with  Eq.~(\ref{eq:LimEmp}), this proves the claim that the
standard distributional limit does indeed hold. 

Notice that $\tau^2$ is given by Eq.~(\ref{eq:LambdaTau1}) that, for
$\scale=1/\zeta$ does indeed coincide with the claimed Eq.~(\ref{eq:ClaimedTau}).
Finally consider the scale parameter $\scale=\scale(p)$ defined by
Eq.~(\ref{eq:GeneralOnsager}). 
We claim that 
\begin{eqnarray}
\lim_{p\to\infty} \scale(p) = \tscale=\frac{1}{\zeta}\, .\label{eq:ClaimOnsager}
\end{eqnarray}
Consider, for the sake of simplicity, the case that $J$
is differentiable and strictly convex (the general case can be obtained as a limit). Then the minimum condition of the proximal operator
(\ref{eq:ProximalOperator}) reads
\begin{eqnarray}
\mtx{\theta} = \eta_b(\by) \;\;\; \Leftrightarrow\;\;\; b\mtx{\Sigma} (\by-\mtx{\theta}) = \nabla J(\mtx{\theta})\, .
\end{eqnarray}
Differentiating with respect to $\mtx{\theta}$, and denoting by $\D\eta_b$
the Jacobian of $\eta_b$, we get $\D\eta_b(\by) =
(\id + b^{-1}\mtx{\Sigma}^{-1} \nabla^2 J(\mtx{\theta}))^{-1}$ and hence 
\begin{eqnarray}
\Enb(a,b) & = & \lim_{p\to\infty}\frac{1}{p}\E\, \Tr \Big\{(1+b^{-1}
\mtx{\Sigma}^{-1/2}\nabla^2 J(\mtx{\htheta}) \mtx{\Sigma}^{-1/2})^{-1}\Big\} \, ,\label{eq:E2Tr}\\
\mtx{\htheta} &\equiv &\eta_b(\mtx{\theta}_0+\sqrt{a}\, \mtx{\Sigma}^{-1/2}\, \mtx{z})\, .
\end{eqnarray}
Hence, combining Eqs.~\eqref{eq:LambdaTau2} and~\eqref{eq:E2Tr} implies that $\tscale=\zeta^{-1}$ satisfies 
\begin{eqnarray}
1&=& \frac{1}{\tscale}+\lim_{p\to\infty}\frac{1}{n}\E \,\Tr \Big\{(1+\tscale
\mtx{\Sigma}^{-1/2}\nabla^2 J(\mtx{\htheta}) \mtx{\Sigma}^{-1/2})^{-1}\Big\} \, ,\\
\mtx{\htheta} &\equiv &\eta_{1/\tscale}(\mtx{\theta}_0+\tau\, \mtx{\Sigma}^{-1/2}\, \mtx{z})\,.
\end{eqnarray}
The claim (\ref{eq:ClaimOnsager}) follows by comparing this with
Eq.~(\ref{eq:GeneralOnsager}), and noting that, by the above $\mtx{\htheta}$
is indeed asymptotically distributed as the estimator (\ref{eqn:general_opt}).

%
%=============================================
%
\newpage
\section{Simulation results}\label{app:Simulation Results}
Consider the setup discussed in Section~\ref{sec:experiment-std}. We compute type I error and statistical power of \sdl\,, ridge-based regression~\cite{BuhlmannSignificance}, and $\ldpe$~\cite{ZhangZhangSignificance} for $10$ realizations of each configuration. The experiment results for the case of identity covariance ($\mtx{\Sigma} = \id_{p\times p}$) are summarized in Tables~\ref{tbl:iid_alpha05} and~\ref{tbl:iid_alpha025}. Table~\ref{tbl:iid_alpha05} and Table~\ref{tbl:iid_alpha025} respectively correspond to significance levels $\alpha = 0.05$ and $\alpha = 0.025$. The results are also compared with the asymptotic bound given in Theorem~\ref{thm:power}.

The results for the case of circulant covariance matrix are summarized in Tables~\ref{tbl:illus_alpha05} and~\ref{tbl:illus_alpha025}. Table~\ref{tbl:illus_alpha05} and Table~\ref{tbl:illus_alpha025} respectively correspond to significance levels $\alpha = 0.05$ and $\alpha = 0.025$. The results are also compared with the lower bound given in Theorem~\ref{thm:power2}. 

For each configuration, the tables contain the means and the standard deviations of type I errors and the powers  across 10 realizations.
A quadruple such as $(1000, 600, 50, 0.1)$ denotes the values of $p = 1000$, $n = 600$, $s_0 = 50$, $\mu = 0.1$. 
\vspace{1.3cm}

%\newpage
%========================================================
% Results for synthetic iid design model (alpha = 0.05)
%========================================================
\begin{table*}[]
\begin{center}
{\small
\begin{tabular}{c|cccc|}
%\cline{2-13}
%& \multicolumn{12}{|c|}{\bf Y} \\  \cline{2-13}

{\bf Method} & Type I err & Type I err & Avg. power & Avg. power \\ 
&(mean) & (std.) & (mean) & (std)
\\
\cline{1-5} \cline{2-5}
\multicolumn{1}{c|}{\bf $\sdl$ $(1000, 600, 50, 0.15)$} &  0.06189 & 0.01663 & 0.83600 & 0.04300
    \\ 
\multicolumn{1}{c|}{\bf Ridge-based regression $(1000, 600, 50, 0.15)$} &  \color{red} 0.00989 & \color{red}0.00239  & \color{red} 0.35000& \color{red} 0.07071
\\
  \multicolumn{1}{c|}{\bf $\ldpe\,$ $(1000, 600, 50, 0.15)$} & \color{darkblue}0.03925 & \color{darkblue} 0.00588& \color{darkblue} 0.55302& \color{darkblue} 0.07608
 \\
\multicolumn{1}{c|}{\bf Asymptotic Bound $(1000, 600, 50, 0.15)$} & \color{olivegreen}0.05 & \color{olivegreen}NA & \color{olivegreen}0.84721 & \color{olivegreen}NA
\\ 
\hline 
\multicolumn{1}{c|}{\bf $\sdl$ $(1000, 600, 25, 0.15)$} &  0.0572 & 0.0190 & 0.8840 & 0.0638
    \\ 
\multicolumn{1}{c|}{\bf Ridge-based regression $(1000, 600, 25, 0.15)$} &  \color{red} 0.0203& \color{red} 0.0052 & \color{red} 0.3680 & \color{red}0.1144  
\\
  \multicolumn{1}{c|}{\bf $\ldpe\,$ $(1000, 600, 25, 0.15)$} & \color{darkblue}0.04010 & \color{darkblue} 0.00917& \color{darkblue} 0.62313& \color{darkblue} 0.05408
 \\
\multicolumn{1}{c|}{\bf Asymptotic Bound $(1000, 600, 25, 0.15)$} & \color{olivegreen}0.05 & \color{olivegreen}NA & \color{olivegreen}0.9057 & \color{olivegreen}NA
\\ 
\hline 
\multicolumn{1}{c|}{\bf $\sdl$ $(1000, 300, 50, 0.15)$} &  0.05547 & 0.01554 & 0.45800 & 0.06957
    \\ 
\multicolumn{1}{c|}{\bf Ridge-based regression $(1000, 300, 50, 0.15)$} &  \color{red} 0.01084 & \color{red}  0.00306& \color{red} 0.19200 & \color{red}  0.04541
\\
  \multicolumn{1}{c|}{\bf $\ldpe\,$ $(1000, 300, 50, 0.15)$} & \color{darkblue}0.03022 & \color{darkblue} 0.00601& \color{darkblue} 0.23008& \color{darkblue} 0.08180
 \\
\multicolumn{1}{c|}{\bf Asymptotic Bound $(1000, 300, 50, 0.15)$} & \color{olivegreen}0.05 & \color{olivegreen}NA & \color{olivegreen}0.31224  & \color{olivegreen}NA
\\ 
\hline 
\multicolumn{1}{c|}{\bf $\sdl$ $(1000, 300, 25, 0.15)$} &   0.05149 & 0.01948 & 0.55600 & 0.11384
    \\ 
\multicolumn{1}{c|}{\bf Ridge-based regression $(1000, 300, 25, 0.15)$} &  \color{red} 0.00964 & \color{red} 0.00436 & \color{red} 0.32400 & \color{red}0.09324
\\
  \multicolumn{1}{c|}{\bf $\ldpe\,$ $(1000, 300, 25, 0.15)$} & \color{darkblue}0.04001 & \color{darkblue} 0.00531& \color{darkblue} 0.34091& \color{darkblue} 0.06408
 \\
\multicolumn{1}{c|}{\bf Asymptotic Bound $(1000, 300, 25, 0.15)$} & \color{olivegreen}0.05 & \color{olivegreen}NA & \color{olivegreen}0.51364  & \color{olivegreen}NA
\\ 
\hline 
\multicolumn{1}{c|}{\bf $\sdl$ $(2000, 600, 100, 0.1)$} &  0.05037 & 0.00874 & 0.44800 & 0.04940
    \\ 
\multicolumn{1}{c|}{\bf Ridge-based regression $(2000, 600, 100, 0.1)$} &  \color{red} 0.01232 & \color{red}  0.00265& \color{red} 0.21900 & \color{red}0.03143
\\
  \multicolumn{1}{c|}{\bf $\ldpe\,$ $(2000, 600, 100, 0.1)$} & \color{darkblue}0.03012 & \color{darkblue} 0.00862& \color{darkblue} 0.31003& \color{darkblue} 0.06338
 \\
\multicolumn{1}{c|}{\bf Asymptotic Bound $(2000, 600, 100, 0.1)$} & \color{olivegreen}0.05 & \color{olivegreen}NA & \color{olivegreen}0.28324   & \color{olivegreen}NA
\\ 
\hline 
\multicolumn{1}{c|}{\bf $\sdl$ $(2000, 600, 50, 0.1)$} &  0.05769 & 0.00725 & 0.52800 & 0.08548
    \\ 
\multicolumn{1}{c|}{\bf Ridge-based regression $(2000, 600, 50, 0.1)$} &  \color{red} 0.01451& \color{red} 0.00303 & \color{red}0.27000 & \color{red} 0.04137 
\\
  \multicolumn{1}{c|}{\bf $\ldpe\,$ $(2000, 600, 50, 0.1)$} & \color{darkblue}0.03221 & \color{darkblue} 0.01001& \color{darkblue} 0.35063& \color{darkblue} 0.05848
 \\
\multicolumn{1}{c|}{\bf Asymptotic Bound $(2000, 600, 50, 0.1)$} & \color{olivegreen}0.05 & \color{olivegreen}NA & \color{olivegreen}0.46818& \color{olivegreen}NA
\\ 
\hline 
\multicolumn{1}{c|}{\bf $\sdl$ $(2000, 600, 20, 0.1)$} &  0.05167 & 0.00814 & 0.58000 & 0.11595
    \\ 
\multicolumn{1}{c|}{\bf Ridge-based regression $(2000, 600, 20, 0.1)$} &  \color{red} 0.01879 & \color{red} 0.00402 & \color{red} 0.34500 & \color{red}0.09846  
\\
  \multicolumn{1}{c|}{\bf $\ldpe\,$ $(2000, 600, 20, 0.1)$} & \color{darkblue}0.04021 & \color{darkblue} 0.00608& \color{darkblue} 0.42048& \color{darkblue} 0.08331
 \\
\multicolumn{1}{c|}{\bf Asymptotic Bound $(2000, 600, 20, 0.1)$} & \color{olivegreen}0.05 & \color{olivegreen}NA & \color{olivegreen}0.58879 & \color{olivegreen}NA
\\ 
\hline 
\multicolumn{1}{c|}{\bf $\sdl$ $(2000, 600, 100, 0.15)$} &   0.05368 & 0.01004  & 0.64500 & 0.05104
    \\ 
\multicolumn{1}{c|}{\bf Ridge-based regression $(2000, 600, 100, 0.15)$} &  \color{red} 0.00921& \color{red} 0.00197& \color{red} 0.30700 & \color{red} 0.04877
\\
  \multicolumn{1}{c|}{\bf $\ldpe\,$ $(2000, 600, 100, 0.15)$} & \color{darkblue}0.02890 & \color{darkblue} 0.00493& \color{darkblue} 0.58003& \color{darkblue} 0.06338
 \\
\multicolumn{1}{c|}{\bf Asymptotic Bound $(2000, 600, 100, 0.15)$} & \color{olivegreen}0.05 & \color{olivegreen}NA & \color{olivegreen}0.54728 & \color{olivegreen}NA
\\ 
\hline 
\multicolumn{1}{c|}{\bf $\sdl$ $(2000, 600, 20, 0.15)$} &  0.04944 & 0.01142 & 0.89500 & 0.07619
    \\ 
\multicolumn{1}{c|}{\bf Ridge-based regression $(2000, 600, 20, 0.15)$} &  \color{red} 0.01763 & \color{red} 0.00329& \color{red} 0.64000& \color{red}0.08756
\\
  \multicolumn{1}{c|}{\bf $\ldpe\,$ $(2000, 600, 20, 0.15)$} & \color{darkblue}0.03554 & \color{darkblue} 0.005047& \color{darkblue} 0.73560& \color{darkblue} 0.04008
 \\
\multicolumn{1}{c|}{\bf Asymptotic Bound $(2000, 600, 20, 0.15)$} & \color{olivegreen}0.05 & \color{olivegreen}NA & \color{olivegreen}0.90608 & \color{olivegreen}NA
\\ 
%\hline 
\end{tabular}
}

\end{center}
\caption{ \small{Comparison between $\sdl$, ridge-based regression~\cite{BuhlmannSignificance}, $\ldpe$~\cite{ZhangZhangSignificance} and the asymptotic bound for $\sdl$ (cf. Theorem~\ref{thm:power}) on the setup described in Section~\ref{sec:experiment-std}. The significance level is $\alpha = 0.05$ and $\mtx{\Sigma}= \id_{p\times p}$ (standard Gaussian design).}}\label{tbl:iid_alpha05}
\end{table*}

 %========================================================
% Results for synthetic iid design model (alpha = 0.025)
%========================================================
%\newpage
\begin{table*}[]
\begin{center}
{\small
\begin{tabular}{c|cccc|}
%\cline{2-13}
%& \multicolumn{12}{|c|}{\bf Y} \\  \cline{2-13}

{\bf Method} & Type I err & Type I err & Avg. power & Avg. power \\ 
&(mean) & (std.) & (mean) & (std)
\\
\cline{1-5} \cline{2-5}
\multicolumn{1}{c|}{\bf $\sdl$ $(1000, 600, 50, 0.15)$} &  0.02874 & 0.00546 & 0.75600  &0.07706
    \\ 
\multicolumn{1}{c|}{\bf Ridge-based regression $(1000, 600, 50, 0.15)$} &  \color{red}0.00379 & \color{red}0.00282 & \color{red}0.22800 & \color{red}0.06052
\\
  \multicolumn{1}{c|}{\bf $\ldpe\,$ $(1000, 600, 100, 0.1)$} & \color{darkblue}0.01459 & \color{darkblue} 0.00605& \color{darkblue} 0.41503& \color{darkblue} 0.08482
 \\
\multicolumn{1}{c|}{\bf Asymptotic Bound $(1000, 600, 50, 0.15)$} & \color{olivegreen}0.025 & \color{olivegreen}NA &  \color{olivegreen}0.77107 &\color{olivegreen} NA
\\ 
\hline 
\multicolumn{1}{c|}{\bf $\sdl$ $(1000, 600, 25, 0.15)$} &  0.03262 &0.00925 &0.79200 &0.04131
    \\ 
\multicolumn{1}{c|}{\bf Ridge-based regression $(1000, 600, 25, 0.15)$} &  \color{red}0.00759 & \color{red}0.00223 & \color{red}0.28800 & \color{red}0.07729   
\\
  \multicolumn{1}{c|}{\bf $\ldpe\,$ $(1000, 600, 25, 0.15)$} & \color{darkblue}0.01032 & \color{darkblue} 0.00490& \color{darkblue} 0.55032& \color{darkblue} 0.07428
 \\
\multicolumn{1}{c|}{\bf Asymptotic Bound $(1000, 600, 25, 0.15)$} & \color{olivegreen}0.025 & \color{olivegreen}NA & \color{olivegreen}0.84912 & \color{olivegreen}NA
\\ 
\hline 
\multicolumn{1}{c|}{\bf $\sdl$ $(1000, 300, 50, 0.15)$} &  0.02916 & 0.00924 & 0.36000 & 0.08380
    \\ 
\multicolumn{1}{c|}{\bf Ridge-based regression $(1000, 300, 50, 0.15)$} &  \color{red}0.00400 & \color{red}0.00257& \color{red} 0.10800& \color{red} 0.05432
\\
  \multicolumn{1}{c|}{\bf $\ldpe\,$ $(1000, 300, 50, 0.15)$} & \color{darkblue}0.01520 & \color{darkblue} 0.00652& \color{darkblue} 0.25332& \color{darkblue} 0.06285
 \\
\multicolumn{1}{c|}{\bf Asymptotic Bound $(1000, 300, 50, 0.15)$} & \color{olivegreen}0.025 & \color{olivegreen}NA &  \color{olivegreen}0.22001 &\color{olivegreen} NA
\\ 
\hline 
\multicolumn{1}{c|}{\bf $\sdl$ $(1000, 300, 25, 0.15)$} &    0.03005 &0.00894 &0.42400 &0.08884
    \\ 
\multicolumn{1}{c|}{\bf Ridge-based regression $(1000, 300, 25, 0.15)$} &  \color{red}0.00492 &\color{red}0.00226 &\color{red}0.21600 & \color{red}0.06310
\\
  \multicolumn{1}{c|}{\bf $\ldpe\,$ $(1000, 300, 25, 0.15)$} & \color{darkblue}0.00881 & \color{darkblue} 0.00377& \color{darkblue} 0.31305& \color{darkblue} 0.05218
 \\
\multicolumn{1}{c|}{\bf Asymptotic Bound $(1000, 300, 25, 0.15)$} & \color{olivegreen}0.025 & \color{olivegreen}NA & \color{olivegreen}0.40207 &\color{olivegreen} NA
\\ 
\hline 
\multicolumn{1}{c|}{\bf $\sdl$ $(2000, 600, 100, 0.1)$} &   0.03079& 0.00663& 0.33000& 0.05033
    \\ 
\multicolumn{1}{c|}{\bf Ridge-based regression $(2000, 600, 100, 0.1)$} &  \color{red}0.00484 & \color{red}0.00179& \color{red} 0.11200&\color{red} 0.03615
\\
  \multicolumn{1}{c|}{\bf $\ldpe\,$ $(2000, 600, 100, 0.1)$} & \color{darkblue}0.01403 & \color{darkblue} 0.00970& \color{darkblue} 0.24308& \color{darkblue} 0.06041
 \\
\multicolumn{1}{c|}{\bf Asymptotic Bound $(2000, 600, 100, 0.1)$} & \color{olivegreen}0.025 & \color{olivegreen}NA &  \color{olivegreen}0.19598  &\color{olivegreen} NA
\\ 
\hline 
\multicolumn{1}{c|}{\bf $\sdl$ $(2000, 600, 50, 0.1)$} &  0.02585 &0.00481 &0.41200& 0.06197
    \\ 
\multicolumn{1}{c|}{\bf Ridge-based regression $(2000, 600, 50, 0.1)$} &  \color{red}0.00662 & \color{red}0.00098 & \color{red}0.20600& \color{red} 0.03406 
\\
  \multicolumn{1}{c|}{\bf $\ldpe\,$ $(2000, 600, 50, 0.1)$} & \color{darkblue}0.01601 & \color{darkblue} 0.00440& \color{darkblue} 0.27031& \color{darkblue} 0.03248
 \\
\multicolumn{1}{c|}{\bf Asymptotic Bound $(2000, 600, 50, 0.1)$} & \color{olivegreen}0.025 & \color{olivegreen}NA & \color{olivegreen}0.35865& \color{olivegreen}NA
\\ 
\hline 
\multicolumn{1}{c|}{\bf $\sdl$ $(2000, 600, 20, 0.1)$} &  0.02626 &0.00510 &0.47500 &0.10607
    \\ 
\multicolumn{1}{c|}{\bf Ridge-based regression $(2000, 600, 20, 0.1)$} &  \color{red}0.00838 & \color{red}0.00232 & \color{red}0.23500 & \color{red}0.08182 
\\
  \multicolumn{1}{c|}{\bf $\ldpe\,$ $(2000, 600, 20, 0.1)$} & \color{darkblue}0.02012 & \color{darkblue} 0.00628& \color{darkblue} 0.34553& \color{darkblue} 0.09848
 \\
\multicolumn{1}{c|}{\bf Asymptotic Bound $(2000, 600, 20, 0.1)$} & \color{olivegreen}0.025 & \color{olivegreen}NA & \color{olivegreen}0.47698& \color{olivegreen}NA
\\ 
\hline 
\multicolumn{1}{c|}{\bf $\sdl$ $(2000, 600, 100, 0.15)$} &   0.02484 &0.00691 &0.52700 &0.09522
 \\ 
\multicolumn{1}{c|}{\bf Ridge-based regression $(2000, 600, 100, 0.15)$} &  \color{red}0.00311 & \color{red}0.00154 & \color{red}0.22500 & \color{red}0.04007 
\\
  \multicolumn{1}{c|}{\bf $\ldpe\,$ $(2000, 600, 100, 0.15)$} & \color{darkblue}0.01482 & \color{darkblue} 0.00717& \color{darkblue} 0.38405& \color{darkblue} 0.03248
 \\
\multicolumn{1}{c|}{\bf Asymptotic Bound $(2000, 600, 100, 0.15)$} & \color{olivegreen}0.025 & \color{olivegreen}NA & \color{olivegreen}0.43511& \color{olivegreen}NA
\\ 
\hline 
\multicolumn{1}{c|}{\bf $\sdl$ $(2000, 600, 20, 0.15)$} &  0.03116 &0.01304& 0.81500 &0.09443
    \\ 
\multicolumn{1}{c|}{\bf Ridge-based regression $(2000, 600, 20, 0.15)$} &  \color{red}0.00727 & \color{red}0.00131 & \color{red}0.54500& \color{red} 0.09560
\\
  \multicolumn{1}{c|}{\bf $\ldpe\,$ $(2000, 600, 20, 0.15)$} & \color{darkblue}0.01801 & \color{darkblue} 0.00399& \color{darkblue} 0.68101& \color{darkblue} 0.06255
 \\
\multicolumn{1}{c|}{\bf Asymptotic Bound $(2000, 600, 20, 0.15)$} & \color{olivegreen}0.025 & \color{olivegreen}NA & \color{olivegreen}0.84963 & \color{olivegreen}NA
\\ 
%\hline 
\end{tabular}
}

\end{center}
\caption{ \small{Comparison between $\sdl$, ridge-based regression~\cite{BuhlmannSignificance}, $\ldpe$~\cite{ZhangZhangSignificance} and the asymptotic bound for $\sdl$ (cf. Theorem~\ref{thm:power}) on the setup described in Section~\ref{sec:experiment-std}. The significance level is $\alpha = 0.025$ and $\mtx{\Sigma} =\id_{p\times p}$ (standard Gaussian design). }}\label{tbl:iid_alpha025}
\end{table*}

%========================================================
% Results for synthetic correlated design model (alpha = 0.05)
%========================================================
%\newpage
\begin{table*}[]
\begin{center}
{\small
\begin{tabular}{c|cccc|}
%\cline{2-13}
%& \multicolumn{12}{|c|}{\bf Y} \\  \cline{2-13}

{\bf Method} & Type I err & Type I err & Avg. power & Avg. power \\ 
&(mean) & (std.) & (mean) & (std)
\\
\cline{1-5} \cline{2-5}
\multicolumn{1}{c|}{\bf \sdl\, $(1000, 600, 50, 0.15)$} &   0.05179 & 0.01262 & 0.81400 & 0.07604 
    \\ 
\multicolumn{1}{c|}{\bf Ridge-based regression $(1000, 600, 50, 0.15)$} &  \color{red}0.01095 & \color{red} 0.00352 & \color{red} 0.34000 & \color{red} 0.05735 
\\ 
  \multicolumn{1}{c|}{\bf $\ldpe\,$ $(1000, 600, 50, 0.15)$} & \color{darkblue}0.02653 & \color{darkblue} 0.00574& \color{darkblue} 0.66800& \color{darkblue} 0.07823
 \\
 \multicolumn{1}{c|}{\bf Lower bound $(1000, 600, 50, 0.15)$} &\color{olivegreen} 0.05& \color{olivegreen} NA& \color{olivegreen}0.84013& \color{olivegreen}0.03810
\\
\hline 
\multicolumn{1}{c|}{\bf \sdl\, $(1000, 600, 25, 0.15)$} &  0.04937 & 0.01840 & 0.85600 & 0.06310 
    \\ 
\multicolumn{1}{c|}{\bf Ridge-based regression $(1000, 600, 25, 0.15)$} &  \color{red}0.01969& \color{red}0.00358 & \color{red}0.46800 & \color{red} 0.08011 
\\ 
  \multicolumn{1}{c|}{\bf $\ldpe\,$ $(1000, 600, 25, 0.15)$} & \color{darkblue}0.01374 & \color{darkblue} 0.00709& \color{darkblue} 0.63200& \color{darkblue} 0.07155
 \\
 \multicolumn{1}{c|}{\bf Lower bound $(1000, 600, 25, 0.15)$} &\color{olivegreen} 0.05& \color{olivegreen} NA& \color{olivegreen}0.86362& \color{olivegreen}0.02227
\\ 
\hline 
\multicolumn{1}{c|}{\bf \sdl\, $(1000, 300, 50, 0.15)$} &  0.05111 & 0.01947 & 0.43800 & 0.09402
    \\ 
\multicolumn{1}{c|}{\bf Ridge-based regression $(1000, 300, 50, 0.15)$} &  \color{red}0.01011 & \color{red}0.00362 & \color{red}0.20200 & \color{red} 0.05029 
\\ 
  \multicolumn{1}{c|}{\bf $\ldpe\,$ $(1000, 300, 50, 0.15)$} & \color{darkblue}0.03621 & \color{darkblue} 0.00701& \color{darkblue} 0.37600& \color{darkblue} 0.07127
 \\
 \multicolumn{1}{c|}{\bf Lower bound $(1000, 300, 50, 0.15)$} &\color{olivegreen} 0.05& \color{olivegreen} NA& \color{olivegreen}0.43435& \color{olivegreen}0.03983
\\ 
\hline 
\multicolumn{1}{c|}{\bf \sdl\, $(1000, 300, 25, 0.15)$} &  0.05262 & 0.01854 &0.53600 & 0.08044 
    \\ 
\multicolumn{1}{c|}{\bf Ridge-based regression $(1000, 300, 25, 0.15)$} &  \color{red}0.01344 & \color{red}0.00258 & \color{red}0.33200 & \color{red}0.08230
\\ 
  \multicolumn{1}{c|}{\bf $\ldpe\,$ $(1000, 300, 25, 0.15)$} & \color{darkblue}0.01682 & \color{darkblue} 0.00352& \color{darkblue} 0.36800& \color{darkblue} 0.10354
 \\
 \multicolumn{1}{c|}{\bf Lower bound $(1000, 300, 25, 0.15)$} &\color{olivegreen} 0.05& \color{olivegreen} NA& \color{olivegreen}0.50198& \color{olivegreen}0.05738
\\ 
\hline 
\multicolumn{1}{c|}{\bf \sdl\, $(2000, 600, 100, 0.1)$} &  0.05268 & 0.01105  & 0.43900 & 0.04383 
    \\ 
\multicolumn{1}{c|}{\bf Ridge-based regression $(2000, 600, 100, 0.1)$} &  \color{red}0.01205 & \color{red}0.00284 & \color{red}0.21200 & \color{red}0.04392
\\ 
  \multicolumn{1}{c|}{\bf $\ldpe\,$ $(2000, 600, 100, 0.1)$} & \color{darkblue}0.028102 & \color{darkblue} 0.00720& \color{darkblue} 0.33419& \color{darkblue} 0.04837
 \\
 \multicolumn{1}{c|}{\bf Lower bound $(2000, 600, 100, 0.1)$} &\color{olivegreen} 0.05& \color{olivegreen} NA& \color{olivegreen}0.41398& \color{olivegreen}0.03424
\\ 
\hline 
\multicolumn{1}{c|}{\bf \sdl\, $(2000, 600, 50, 0.1)$} &  0.05856 & 0.00531 & 0.50800 & 0.05350
    \\ 
\multicolumn{1}{c|}{\bf Ridge-based regression $(2000, 600, 50, 0.1)$} &  \color{red} 0.01344& \color{red} 0.00225& \color{red} 0.26000& \color{red}0.03771  
\\ 
  \multicolumn{1}{c|}{\bf $\ldpe\,$ $(2000, 600, 50, 0.1)$} & \color{darkblue}0.03029 & \color{darkblue} 0.00602& \color{darkblue} 0.37305& \color{darkblue} 0.07281
 \\
 \multicolumn{1}{c|}{\bf Lower bound $(2000, 600, 50, 0.1)$} &\color{olivegreen} 0.05& \color{olivegreen} NA& \color{olivegreen}0.49026& \color{olivegreen}0.02625
\\ 
\hline 
\multicolumn{1}{c|}{\bf \sdl\, $(2000, 600, 20, 0.1)$} &  0.04955 & 0.00824 & 0.57500 & 0.13385 
    \\ 
\multicolumn{1}{c|}{\bf Ridge-based regression $(2000, 600, 20, 0.1)$} &  \color{red}0.01672 & \color{red}0.00282 & \color{red}0.35500 & \color{red} 0.08960 
\\ 
  \multicolumn{1}{c|}{\bf $\ldpe\,$ $(2000, 600, 20, 0.1)$} & \color{darkblue}0.03099 & \color{darkblue} 0.00805& \color{darkblue} 0.31350& \color{darkblue} 0.04482
 \\
 \multicolumn{1}{c|}{\bf Lower bound $(2000, 600, 20, 0.1)$} &\color{olivegreen} 0.05& \color{olivegreen} NA& \color{olivegreen}0.58947& \color{olivegreen}0.04472
\\ 
\hline 
\multicolumn{1}{c|}{\bf \sdl\, $(2000, 600, 100, 0.15)$} &  0.05284& 0.00949 & 0.61600 &0.06802 
    \\ 
\multicolumn{1}{c|}{\bf Ridge-based regression $(2000, 600, 100, 0.15)$} &  \color{red}0.00895 & \color{red} 0.00272 & \color{red} 0.31800 & \color{red} 0.04131
\\ 
  \multicolumn{1}{c|}{\bf $\ldpe\,$ $(2000, 600, 100, 0.15)$} & \color{darkblue}0.01022 & \color{darkblue} 0.00570& \color{darkblue} 0.35904& \color{darkblue} 0.05205
 \\
 \multicolumn{1}{c|}{\bf Lower bound $(2000, 600, 100, 0.15)$} &\color{olivegreen} 0.05& \color{olivegreen} NA& \color{olivegreen}0.64924& \color{olivegreen}0.05312
\\ 
\hline 
\multicolumn{1}{c|}{\bf \sdl\, $(2000, 600, 20, 0.15)$} &  0.05318 & 0.00871 & 0.85500 & 0.11891
    \\ 
\multicolumn{1}{c|}{\bf Ridge-based regression $(2000, 600, 20, 0.15)$} &  \color{red}0.01838  & \color{red}0.00305 & \color{red}0.68000 & \color{red}0.12517
\\ 
  \multicolumn{1}{c|}{\bf $\ldpe\,$ $(2000, 600, 20, 0.15)$} & \color{darkblue}0.02512 & \color{darkblue} 0.00817& \color{darkblue} 0.36434& \color{darkblue} 0.05824
 \\
 \multicolumn{1}{c|}{\bf Lower bound $(2000, 600, 20, 0.15)$} &\color{olivegreen} 0.05& \color{olivegreen} NA& \color{olivegreen}0.87988& \color{olivegreen}0.03708
\\ 
%\hline 
\end{tabular}
}

\end{center}
\caption{ 
\small{Comparison between \sdl, ridge-based regression~\cite{BuhlmannSignificance}, $\ldpe$~\cite{ZhangZhangSignificance} and the lower bound for the statistical power of \sdl\, (cf. Theorem~\ref{thm:power2}) on the setup described in Section~\ref{sec:experiment-nstd}. The significance level is $\alpha = 0.05$ and $\mtx{\Sigma}$ is the described circulant matrix (nonstandard Gaussian design).}}
\label{tbl:illus_alpha05}
\end{table*}

%========================================================
% Results for synthetic correlated design model (alpha = 0.025)
%========================================================
%\newpage
\begin{table*}[]
\begin{center}
{\small
\begin{tabular}{c|cccc|}
%\cline{2-13}
%& \multicolumn{12}{|c|}{\bf Y} \\  \cline{2-13}

{\bf Method} & Type I err & Type I err & Avg. power & Avg. power \\ 
&(mean) & (std.) & (mean) & (std)
\\
\cline{1-5} \cline{2-5}
\multicolumn{1}{c|}{\bf \sdl\, $(1000, 600, 50, 0.15)$} &   0.02579 & 0.00967& 0.71800 &0.03824
    \\ 
\multicolumn{1}{c|}{\bf Ridge-based regression $(1000, 600, 50, 0.15)$} &  \color{red} 0.00326& \color{red} 0.00274& \color{red}  0.21000& \color{red} 0.05437
\\ 
  \multicolumn{1}{c|}{\bf $\ldpe\,$ $(1000, 600, 50, 0.15)$} & \color{darkblue}0.01245 & \color{darkblue} 0.00391& \color{darkblue} 0.64807& \color{darkblue} 0.065020
 \\
 \multicolumn{1}{c|}{\bf Lower bound $(1000, 600, 50, 0.15)$} &\color{olivegreen} 0.025& \color{olivegreen} NA& \color{olivegreen}0.75676& \color{olivegreen}0.05937
\\ 
\hline 
\multicolumn{1}{c|}{\bf \sdl\, $(1000, 600, 25, 0.15)$} &   0.02462 & 0.00866 & 0.75600 & 0.12429
    \\ 
\multicolumn{1}{c|}{\bf Ridge-based regression $(1000, 600, 25, 0.15)$} &  \color{red}0.01077& \color{red} 0.00346& \color{red} 0.30400& \color{red}  0.08262
\\ 
  \multicolumn{1}{c|}{\bf $\ldpe\,$ $(1000, 600, 25, 0.15)$} & \color{darkblue}0.00931 & \color{darkblue} 0.00183& \color{darkblue} 0.68503& \color{darkblue} 0.17889
 \\
 \multicolumn{1}{c|}{\bf Lower bound $(1000, 600, 25, 0.15)$} &\color{olivegreen} 0.025& \color{olivegreen} NA& \color{olivegreen}0.80044& \color{olivegreen}0.05435
\\ 
\hline 
\multicolumn{1}{c|}{\bf \sdl\, $(1000, 300, 50, 0.15)$} &   0.02646 & 0.01473 & 0.39200 & 0.11478
    \\ 
\multicolumn{1}{c|}{\bf Ridge-based regression $(1000, 300, 50, 0.15)$} &  \color{red} 0.00368& \color{red} 0.00239& \color{red} 0.15000& \color{red}  0.04137
\\ 
  \multicolumn{1}{c|}{\bf $\ldpe\,$ $(1000, 300, 50, 0.15)$} & \color{darkblue}0.01200 & \color{darkblue} 0.00425& \color{darkblue} 0.28800& \color{darkblue} 0.09654
 \\
 \multicolumn{1}{c|}{\bf Lower bound $(1000, 300, 50, 0.15)$} &\color{olivegreen} 0.025& \color{olivegreen} NA& \color{olivegreen}0.36084& \color{olivegreen}0.04315
\\ 
\hline 
\multicolumn{1}{c|}{\bf \sdl\, $(1000, 300, 25, 0.15)$} &   0.02400 & 0.00892 & 0.42400 & 0.09834
    \\ 
\multicolumn{1}{c|}{\bf Ridge-based regression $(1000, 300, 25, 0.15)$} &  \color{red} 0.00513& \color{red}0.00118 & \color{red}0.18800 & \color{red} 0.07786
\\ 
  \multicolumn{1}{c|}{\bf $\ldpe\,$ $(1000, 300, 25, 0.15)$} & \color{darkblue}0.00492 & \color{darkblue} 0.00169& \color{darkblue} 0.24500& \color{darkblue} 0.07483
 \\
 \multicolumn{1}{c|}{\bf Lower bound $(1000, 300, 25, 0.15)$} &\color{olivegreen} 0.025& \color{olivegreen} NA& \color{olivegreen}0.42709& \color{olivegreen}0.03217
\\ 
\hline 
\multicolumn{1}{c|}{\bf \sdl\, $(2000, 600, 100, 0.1)$} &  0.03268 & 0.00607& 0.32600& 0.07412
    \\ 
\multicolumn{1}{c|}{\bf Ridge-based regression $(2000, 600, 100, 0.1)$} &  \color{red} 0.00432& \color{red} 0.00179& \color{red} 0.14100& \color{red} 0.05065
\\ 
  \multicolumn{1}{c|}{\bf $\ldpe\,$ $(2000, 600, 100, 0.1)$} & \color{darkblue}0.01240 & \color{darkblue} 0.00572& \color{darkblue} 0.20503& \color{darkblue} 0.09280
 \\
 \multicolumn{1}{c|}{\bf Lower bound $(2000, 600, 100, 0.1)$} &\color{olivegreen} 0.025& \color{olivegreen} NA& \color{olivegreen}0.32958& \color{olivegreen}0.03179
\\ 
\hline 
\multicolumn{1}{c|}{\bf \sdl\, $(2000, 600, 50, 0.1)$} &  0.03108 &0.00745 &0.41800 &0.04662
    \\ 
\multicolumn{1}{c|}{\bf Ridge-based regression $(2000, 600, 50, 0.1)$} &  \color{red} 0.00687& \color{red} 0.00170& \color{red} 0.18800& \color{red}  0.06680
\\ 
  \multicolumn{1}{c|}{\bf $\ldpe\,$ $(2000, 600, 50, 0.1)$} & \color{darkblue}0.014005 & \color{darkblue} 0.00740& \color{darkblue} 0.25331& \color{darkblue} 0.04247
 \\
 \multicolumn{1}{c|}{\bf Lower bound $(2000, 600, 50, 0.1)$} &\color{olivegreen} 0.025& \color{olivegreen} NA& \color{olivegreen}0.40404& \color{olivegreen}0.06553
\\ 
\hline 
\multicolumn{1}{c|}{\bf \sdl\, $(2000, 600, 20, 0.1)$} &    0.02965 & 0.00844 & 0.38500 & 0.07091
    \\ 
\multicolumn{1}{c|}{\bf Ridge-based regression $(2000, 600, 20, 0.1)$} &  \color{red} 0.00864& \color{red}0.00219& \color{red} 0.22500& \color{red}0.07906
\\ 
  \multicolumn{1}{c|}{\bf $\ldpe\,$ $(2000, 600, 20, 0.1)$} & \color{darkblue}0.01912 & \color{darkblue} 0.00837& \color{darkblue} 0.31551& \color{darkblue} 0.06288
 \\
 \multicolumn{1}{c|}{\bf Lower bound $(2000, 600, 20, 0.1)$} &\color{olivegreen} 0.025& \color{olivegreen} NA& \color{olivegreen}0.47549& \color{olivegreen}0.06233
\\ 
\hline 
\multicolumn{1}{c|}{\bf \sdl\, $(2000, 600, 100, 0.15)$} &   0.026737& 0.009541 & 0.528000 & 0.062681
    \\ 
\multicolumn{1}{c|}{\bf Ridge-based regression $(2000, 600, 100, 0.15)$} &  \color{red}0.002947& \color{red} 0.000867  & \color{red}0.236000  & \color{red} 0.035653
\\ 
  \multicolumn{1}{c|}{\bf $\ldpe\,$ $(2000, 600, 100, 0.15)$} & \color{darkblue}0.01012 & \color{darkblue} 0.00417& \color{darkblue} 0.36503& \color{darkblue} 0.05823
 \\
 \multicolumn{1}{c|}{\bf Lower bound $(2000, 600, 100, 0.15)$} &\color{olivegreen} 0.025& \color{olivegreen} NA& \color{olivegreen}0.54512& \color{olivegreen}0.05511
\\ 
\hline 
\multicolumn{1}{c|}{\bf \sdl\, $(2000, 600, 20, 0.15)$} &  0.03298 & 0.00771  & 0.79000 & 0.12202
    \\ 
\multicolumn{1}{c|}{\bf Ridge-based regression $(2000, 600, 20, 0.15)$} &  \color{red} 0.00732 & \color{red} 0.00195& \color{red}0.53500& \color{red}0.07091
\\ 
 \multicolumn{1}{c|}{\bf $\ldpe\,$ $(2000, 600, 20, 0.15)$} & \color{darkblue}0.01302 & \color{darkblue} 0.00711& \color{darkblue} 0.60033& \color{darkblue} 0.03441
 \\
 \multicolumn{1}{c|}{\bf Lower bound $(2000, 600, 20, 0.15)$} &\color{olivegreen} 0.025& \color{olivegreen} NA& \color{olivegreen}0.81899& \color{olivegreen}0.03012
\\
%\hline 
\end{tabular}
}

\end{center}
\caption{ 
\small{Comparison between \sdl, ridge-based regression~\cite{BuhlmannSignificance}, $\ldpe$~\cite{ZhangZhangSignificance} and the lower bound for the statistical power of \sdl\, (cf. Theorem~\ref{thm:power2})  on the setup described in Section~\ref{sec:experiment-nstd}. The significance level is $\alpha = 0.025$ and $\mtx{\Sigma}$ is the described circulant matrix (nonstandard Gaussian design).}} \label{tbl:illus_alpha025}
%Comparison between \Proced 1, B\"uhlmann's method~\cite{BuhlmannSignificance} and the lower bound for \Proced 1 (established in Theorem~\ref{thm:power2}) on the illustration setup in Section~\ref{sec:illustration}. The significance level is $\alpha = 0.025$. The means and the standard deviations of the powers and type I errors in the table are obtained by testing over $10$ realizations of the corresponding configuration. Here a quadruple such as $(1000, 600, 50, 0.1)$ denotes the values of $p = 1000$, $n = 600$, $s_0 = 50$, $b = 0.1$. }\label{tbl:illus_alpha025}
\end{table*}

 \newpage
\section{Alternative hypothesis testing procedure}
\label{app:CovarianceFree}
\sdl, described in Table~\ref{tbl:SDL}, needs to compute an estimate
of the covariance matrix $\mtx{\Sigma}$. Here, we discuss another hypothesis
testing procedure which leverages on a slightly different form of the standard
distributional limit, cf. Definition~\ref{def:SDL}. This procedure
only requires bounds on $\mtx{\Sigma}$ that can be estimated from the
data. Furthermore, we establish a connection with the hypothesis testing
procedure of~\cite{BuhlmannSignificance}. We will describe this
alternative procedure synthetically since it is not the main focus of
the paper.

By Definition~\ref{def:SDL}, if a sequence of instances $\mathcal{S} =
\{(\mtx{\Sigma}(p), \mtx{\theta}(p), n(p), \sigma(p))\}_{p \in \naturals}$ has
standard distributional limit, then with probability one the empirical
distribution of $\{(\htheta_i^u -
\theta_i)/[(\mtx{\Sigma}^{-1})_{ii}]^{1/2}\}_{i=1}^p$ converges weakly to
$\normal(0,\tau^2)$. 
We make a somewhat different assumption that is also supported by the
statistical physics arguments of
Appendix \ref{app:replica}. The two assumptions coincide in the case
of standard Gaussian designs.

In order to motivate the new assumption, notice that the standard
distributional limit is consistent with $\mtx{\htheta}^u - \mtx{\theta}_0$ being
approximately $\normal(0,\tau^2 \mtx{\Sigma}^{-1})$. If this holds, then
\begin{eqnarray}
\mtx{\Sigma} (\mtx{\htheta}^u - \mtx{\theta}_0) = \mtx{\Sigma} (\mtx{\htheta} -\mtx{\theta}_0) +
\frac{\scale}{n} \bX^{\sT} (\by- \bX \mtx{\htheta}) \approx \normal(0,\tau^2\mtx{\Sigma}).
\end{eqnarray}
This motivates the definition of $\tilde{\theta}_i =
\tau^{-1}(\Sigma_{ii})^{-1/2} [\mtx{\Sigma} (\mtx{\htheta}^u - \mtx{\theta}_0)]_i$.
We then assume that the empirical distribution of
$\{(\theta_{0,i},\tilde{\theta}_i)\}_{i\in [p]}$
converges weakly to $(\Theta_0,Z)$, with $Z\sim\normal(0,1)$
independent of $\Theta_0$. 

Under the null-hypothesis $H_{0,i}$, we get
\begin{align}
\tilde{\theta}_i &= \tau^{-1}(\Sigma_{ii})^{-1/2} [\mtx{\Sigma} (\mtx{\htheta}^u - \mtx{\theta}_0)]_i\\
&= \tau^{-1}(\Sigma_{ii})^{-1/2} [\mtx{\Sigma} (\mtx{\htheta} - \mtx{\theta}_0) + \frac{\scale}{n} \bX^{\sT} (\by- \bX \mtx{\htheta})]_i\\
&= \tau^{-1} (\Sigma_{ii})^{1/2} \htheta_i + \tau^{-1} (\Sigma_{ii})^{-1/2} [\frac{\scale}{n}\bX^\sT(\by-\bX\mtx{\htheta})]_i +\tau^{-1} (\Sigma_{ii})^{-1/2} \mtx{\Sigma}_{i,\sim i} (\mtx{\htheta}_{\sim i} - \mtx{\theta}_{0,\sim i}),
\end{align} 
where $\mtx{\Sigma}_{i,\sim i}$ denotes the vector $(\Sigma_{ij})_{j \ne i}$. Similarly $\mtx{\htheta}_{\sim i}$ and $\mtx{\theta}_{0,\sim i}$ respectively denote the vectors $(\htheta_j)_{j \neq i}$ and $(\theta_{0,j})_{j \neq i}$.
Therefore,
\begin{eqnarray}
\tau^{-1} (\Sigma_{ii})^{1/2} \htheta_i + \tau^{-1} (\Sigma_{ii})^{-1/2} [\frac{\scale}{n}\bX^\sT(\by-\bX\mtx{\htheta})]_i = \tilde{\theta}_i - \tau^{-1} (\Sigma_{ii})^{-1/2} \mtx{\Sigma}_{i,\sim i} (\mtx{\htheta}_{\sim i} - \mtx{\theta}_{0,\sim i}).
\end{eqnarray}
Following the philosophy of ~\cite{BuhlmannSignificance}, 
the key step in obtaining a p-value for testing $H_{0,i}$ is to find constants $\Delta_i$, such that asymptotically
\begin{eqnarray}\label{eqn:xi}
\xi_i\equiv \tau^{-1} (\Sigma_{ii})^{1/2} \htheta_i + \tau^{-1} (\Sigma_{ii})^{-1/2} [\frac{\scale}{n}\bX^\sT(\by-\bX\mtx{\htheta})]_i   \preceq |Z| + \Delta_i,
\end{eqnarray}
where $Z\sim \normal(0,1)$, and $\preceq$ denotes ``stochastically smaller than or equal to''. 
%Denote the left hand side of Eq.~\eqref{eqn:xi} by $\xi_i$. 
Then, we can define the p-value for the two-sided alternative as
\begin{eqnarray}
P_i = 2(1-\Phi((|\xi_i|  -\Delta_i)+)).
\end{eqnarray}
Control of type I errors then follows immediately from the construction of p-values:
\begin{eqnarray}
\underset{p \to \infty}{\lim \sup}\, \P(P_i \le \alpha) \le \alpha, \quad \quad \text{if } H_{0,i} \text{ holds}.
\end{eqnarray}
In order to define the constant $\Delta_i$, we use analogous argument to the one in~\cite{BuhlmannSignificance}:
\begin{eqnarray}
|\tau^{-1} (\Sigma_{ii})^{-1/2} \mtx{\Sigma}_{i,\sim i} (\mtx{\htheta}_{\sim i} - \mtx{\theta}_{0,\sim i})|
\le \tau^{-1}(\Sigma_{ii})^{-1/2}\, \Big(\max_{j \neq i} |{\Sigma}_{i,j}|\Big) \|\mtx{\htheta}- \mtx{\theta}_0\|_1.
\end{eqnarray}
Recall that $\mtx{\htheta} = \mtx{\htheta}(\lambda)$ is the solution of the Lasso with regularization parameter $\lambda$. Due to the result of~\cite{BickelEtAl,BuhlmannVanDeGeer}, using $\lambda = 4\sigma\sqrt{(t^2+2\log(p))/n}$, the following holds with probability at least $1- 2e^{-t^2/2}$:
\begin{eqnarray}
\|\mtx{\htheta} - \mtx{\theta}_0\|_1 \le 4\lambda s_0/\phi_0^2,
\end{eqnarray}
where $s_0$ is the sparsity (number of active parameters) and $\phi_0$
is the compatibility constant.  Assuming for simplicity
${\Sigma}_{i,i}=1$ (which can be ensured by normalizing the columns of
$\bX$), we can define
\begin{align}
\Delta_i \equiv \frac{4\lambda s_0}{\tau \phi_0^2}\max_{j \neq i} |{\Sigma}_{ij}|\,. 
\end{align}
Therefore, this procedure only requires to bound the off-diagonal
entries of $\mtx{\Sigma}$, i.e., $\max_{j\neq i}|{\Sigma}_{ij}|$. 
It is straightforward to bound this quantity using the empirical
covariance, $\hSigma = (1/n) \bX^\sT \bX$. 
\begin{claim}
Consider Gaussian design matrix $\bX\in \reals^{n\times p}$, whose rows are drawn independently
from $\normal(0,\mtx{\Sigma})$. Without loss of generality assume  $\Sigma_{ii} = 1$, for $i\in [p]$. 
For any fixed $i\in[p]$, the following holds true   
with probability at least $1-2p^{-1}$ 
\begin{eqnarray}\label{eqn:empSigma}
\max_{j\neq i}|{\Sigma}_{i,j}| \le \max_{j\neq i}|\hSigma_{i,j}| + 40 \sqrt{\frac{\log p}{n}}\,.
\end{eqnarray}
\end{claim}
\begin{proof}
Let $\mtx{Z} = \mtx{\hSigma} -\mtx{\Sigma}$. Fix $i,j \in [p]$ and for $\ell \in [n]$, let $v_\ell = X_{\ell,i} X_{\ell j} - \Sigma_{i j}$.
Then $Z_{ij}  = \frac{1}{n}\sum_{\ell=1}^n v_\ell$. Notice that the random variables $v_\ell$ are independent and $\E(v_{\ell}) = 0$. Further $v_\ell$ is sub-exponential. More specifically, letting $\|\cdot\|_{\psi_1}$ and $\|\cdot\|_{\psi_2}$ denote
the sub-exponential and sub-gaussian norms respectively, we have
\begin{align}
\|v_\ell\|_{\psi_1} \le 2\|X_{\ell i} X_{\ell j}\|_{\psi_1} \le 4\|X_{\ell i}\|_{\psi_2} \|X_{\ell j}\|_{\psi_2} = 4\,,
\end{align}
where the first step follows from~\cite[Remark 5.18]{Vershynin-CS} and the second step follows from definition of sub-exponential and sub-gaussian norms and using the assumption $\Sigma_{ii} = 1$. 

Now, by applying Bernstein-type inequality for centered sub-exponential random variables~\cite{Vershynin-CS}, we get
\begin{align}
\prob\Big\{\frac{1}{n} \Big|\sum_{\ell=1}^n v_\ell \Big| \ge \eps \Big\} \le 2 \exp \Big[ -\frac{n}{6} \min\Big((\frac{\eps}{4e})^2, \frac{\eps}{4e}\Big) \Big]\,.
\end{align}
Choosing $\eps = 40 \sqrt{(\log p)/n}$, and 
assuming $n \ge (100/e) \log p$, we arrive at
\begin{align}
\prob\bigg \{\frac{1}{n} \Big|\sum_{\ell=1}^n v^{(ij)}_\ell \Big| \ge 40 \sqrt{\frac{\log p}{n}}  \bigg\} 
\le 2 p^{-{100}/(6e^2)} < 2p^{-2}\,.
\end{align}
 Using union bound for $j\in [p]$, $j\neq i$, we get
\begin{eqnarray}
 \prob\Big(\max_{j\neq i} |\hSigma_{i,j} - {\Sigma}_{i,j}|\le 40\sqrt{\frac{\log p}{n}}\Big) \ge1 - 2p^{-1}\, .
 \end{eqnarray}
 The result follows from the inequality $\max_{j\neq i} |{\Sigma}_{i,j}| - \max_{j\neq i} |\hSigma_{i,j}| \le \max_{j\neq i} |\hSigma_{i,j} - {\Sigma}_{i,j}|$.
\end{proof}
%
%***********************************************
%
\section{Proof of Lemma \ref{lemma:MatrixBound}}
\label{app:MatrixBound}

Let $\bOmega\equiv\bSigma^{-1}$, $\bR \equiv (\bOmega\hbSigma - \id)_{A,B}$ and define
$\cF_1\equiv\{\bu\in S^{p-1}:\,\supp(\bu)\subseteq[A] \}$, 
$\cF_2\equiv\{\bv\in S^{p-1}:\,\supp(\bv)\subseteq[B] \}$, with $S^{p-1}
\equiv\{\bv\in\reals^p:\, \|\bv\|_2 = 1\}$.
We have
\begin{align}
\|\bR\|_2 &= \underset{\substack{\bu,\bv\\ \|\bu\|, \|\bv\|\le 1}}{\sup} \<\bu,\bR \bv\> \nonumber \\
&=\underset{\substack{\bu,\bv\\ \|\bu\|, \|\bv\|\le 1}}{\sup} \Big(\<\bu,\frac{1}{n} \sum_{i=1}^n (\bOmega \bx_i)_A (\bx_i^\sT)_B v\> - \<\bu_A,\bv_B\> \Big) \nonumber \\
&\le \underset{\bu\in \cF_1, \bv\in \cF_2}{\sup} \,\frac{1}{n} \sum_{i=1}^n
\Big( \<\bu, \bOmega \bx_i\> \<\bx_i,\bv\> - \<\bu,\bv\>\Big)\,.\label{eq:R-1}
\end{align}

Fix $\bu\in \cF_1$ and $\bv \in \cF_2$. 
Let $\xi_i \equiv \<u, \bOmega \bx_i\> \<\bx_i,\bv\> - \<\bu,\bv\>$.
The variables $\xi_{i}$ are independent and it is easy to see that $\E(\xi_i) = 0$. 
Throughout, let $\|\cdot\|_{\psi_1}$ and $\|\cdot\|_{\psi_2}$ respectively denote
the sub-exponential and sub-gaussian norms.
By~\cite[Remark 5.18]{Vershynin-CS}, we have
$$ \|\xi_i\|_{\psi_1} \le 2\|\<\bu, \bOmega \bx_i\> \<\bx_i,\bv\>\|_{\psi_1}\,.$$
Moreover, recalling that for any two random variables $X, Y$,
$\|XY\|_{\psi_1}\le 2\|X\|_{\psi_2}\|Y\|_{\psi_2}$
\cite{Vershynin-CS}, we have
\begin{align*}
\|\<\bu, \bOmega \bx_i\> \<\bx_i,\bv\>\|_{\psi_1} &\le 2\|\<\bu, \bOmega\bx_i\>\|_{\psi_2} \|\<\bx_i,\bv\>\|_{\psi_2}\\
 &= 2\|\bOmega^{1/2}\bu\|_2 \|\bOmega^{-1/2}\bv\|_2 \|\bOmega^{1/2}\bx_i\|_{\psi_2}^2\\
 &\le 2\sqrt{\sigma_{\max}(\bSigma)/\sigma_{\min}(\bSigma)} \|\bOmega^{1/2}\bx_i\|_{\psi_2}^2\,.
\end{align*}
Since $\bOmega^{1/2}\bx_i \sim \normal(0,\id)$, we have $ \|\bOmega^{1/2}\bx_i\|_{\psi_2} =1$, and thus $\max_{i\in[n]}\|\xi_i\|_{\psi_1} \le C$, for some constant $C = C(c_{\rm min}, c_{\rm max})$. Now, by applying Bernstein inequality for centered sub-exponential random
variables~\cite{Vershynin-CS}, for every $t\ge 0$, we have
\[
\prob\bigg(\frac{1}{n} \sum_{i=1}^n \xi_i  \ge t \Big) \le 2 \exp\bigg[-cn\min\Big(\frac{t^2}{C^2}, \frac{t}{C}\Big) \bigg]\,,
\]
where $c>0$ is an absolute constant. Therefore, for any constant
$c_1>0$, since $n = \omega(s_0\log p)$, we have
\begin{eqnarray}\label{eq:R-2}
\prob\bigg(\frac{1}{n} \sum_{i=1}^n \xi_i \ge C\sqrt{\frac{c_1 s_0\log p}{cn}} \bigg) \le  p^{-c_1s_0}\,.
\end{eqnarray}
In order to bound the right hand side of Eq.~\eqref{eq:R-1}, we use a \emph{$\eps$-net argument}.
Clearly, $\cF_1\cong S^{|A|-1}$ and $\cF_2 \cong S^{|B|-1}$ where $\cong$ denotes that the two objects are isometric.
By~\cite[Lemma 5.2]{Vershynin-CS}, there exists a $\frac{1}{2}$-net $\cN_1$ of $S^{|A|-1}$
(and hence of $\cF_1$) with size at most $5^{|A|}$. Similarly there exists a $\frac{1}{2}$-net $\cN_2$ of $\cF_2$
of size at most $5^{|B|}$. Hence, using Eq.~\eqref{eq:R-2} and taking union bound over all vectors in ${\cal N}_1$ and ${\cal N}_2$ , we obtain
\begin{align}
\underset{{\bu\in {\cal N}_1}, {\bv \in {\cal N}_2}}{\sup}  \frac{1}{n} \sum_{i=1}^n \<\bu, (\bOmega \bx_i \bx_i^\sT - \id)\bv\> 
%\quad \quad \quad \quad \,
 \le C\sqrt{\frac{c_1 s_0\log p}{cn}}\,, \label{eq:R-3}
\end{align}
with probability at least $1- 5^{|A|+|B|} p^{-c_1s_0}$.

The last part of the argument is based on the following lemma, whose
proof is standard
(see e.g. \cite{Vershynin-CS} or \cite[Appendix D]{JavanMon-OptSample}).
\begin{lemma}\label{lem:net}
Let $\mtx{M} \in \reals^{p\times p}$. Then,
$$ \underset{\bu\in \cF_1, \bv\in \cF_2}{\sup} \<\bu,\mtx{M}\bv\> \le 4 \underset{u\in {\cal N}_1, v\in {\cal N}_2}{\sup} \<\bu,\mtx{M}\bv\>\,.$$
\end{lemma}
Employing Lemma~\ref{lem:net} and bound~\eqref{eq:R-3} in Eq.~\eqref{eq:R-1}, we arrive at
\begin{align}\label{eq:R-4}
\|\bR\|_2\le 4C\sqrt{\frac{c_1 s_0\log p}{cn}}\,,
\end{align}
with probability at least $1- 5^{|A|+|B|} p^{-c_1s_0}$.

Finally, note that there are  less than $p^{2c_0s_0}$ pairs of subsets $A,B$,
with $|A|$, $|B|\le c_0s_0$.
Taking union bound over all these sets, we obtain that with high probability,
$$\|(\bOmega\hbSigma-\id)_{A,B}\|_2 \le  K \sqrt{s_0 \log p /n}\,,$$
for all such sets $A, B$, where $K = K(c_0, c_{\min},c_{\max})$
is a constant.

%% that's all folks

\newpage

\bibliographystyle{ieeetr}
\bibliography{all-bibliography}

\begin{thebibliography}{10}

\bibitem{MeinshausenBuhlmann}
N.~Meinshausen and P.~B{\"u}hlmann, ``High-dimensional graphs and variable
  selection with the lasso,'' {\em Ann.~Statist.}, vol.~34, pp.~1436--1462,
  2006.

\bibitem{ren2013asymptotic}
Z.~Ren, T.~Sun, C.-H. Zhang, and H.~H. Zhou, ``Asymptotic normality and
  optimalities in estimation of large gaussian graphical model,'' {\em {\sf
  arXiv:1309.6024}}, 2013.

\bibitem{Do}
D.~L. Donoho, ``{High-dimensional centrally symmetric polytopes with
  neighborliness proportional to dimension},'' {\em Discrete Comput. Geometry},
  vol.~35, pp.~617--652, 2006.

\bibitem{DoTa05}
D.~L. Donoho and J.~Tanner, ``Neighborliness of randomly-projected simplices in
  high dimensions,'' {\em Proceedings of the National Academy of Sciences},
  vol.~102, no.~27, pp.~9452--9457, 2005.

\bibitem{DoTa08}
D.~L. Donoho and J.~Tanner, ``Counting faces of randomly projected polytopes
  when the projection radically lowers dimension,'' {\em Journal of American
  Mathematical Society}, vol.~22, pp.~1--53, 2009.

\bibitem{DoTa10}
D.~L. Donoho and J.~Tanner, ``Precise undersampling theorems,'' {\em
  Proceedings of the IEEE}, vol.~98, pp.~913--924, 2010.

\bibitem{lehmann2005testing}
E.~Lehmann and J.~Romano, {\em Testing statistical hypotheses}.
\newblock Springer, 2005.

\bibitem{CandesRechtSimple}
E.~Cand{\`e}s and B.~Recht, ``Simple bounds for recovering low-complexity
  models,'' {\em Mathematical Programming}, pp.~1--13, 2012.

\bibitem{CandesPlanRIPless}
E.~J. Candes and Y.~Plan, ``A probabilistic and ripless theory of compressed
  sensing,'' {\em Information Theory, IEEE Transactions on}, vol.~57, no.~11,
  pp.~7235--7254, 2011.

\bibitem{BayatiMontanariLASSO}
M.~Bayati and A.~Montanari, ``{The LASSO risk for gaussian matrices},'' {\em
  IEEE Trans. on Inform. Theory}, vol.~58, pp.~1997--2017, 2012.

\bibitem{DMM-NSPT-11}
D.~Donoho, A.~Maleki, and A.~Montanari, ``{The Noise Sensitivity Phase
  Transition in Compressed Sensing},'' {\em IEEE Trans. on Inform. Theory},
  vol.~57, pp.~6920--6941, 2011.

\bibitem{BP95}
S.~Chen and D.~Donoho, ``{Examples of basis pursuit},'' in {\em Proceedings of
  Wavelet Applications in Signal and Image Processing III}, (San Diego, CA),
  1995.

\bibitem{Tibs96}
R.~Tibshirani, ``{Regression shrinkage and selection with the Lasso},'' {\em J.
  Royal. Statist. Soc B}, vol.~58, pp.~267--288, 1996.

\bibitem{IndykLower}
K.~D. Ba, P.~Indyk, E.~Price, and D.~P. Woodruff, ``Lower bounds for sparse
  recovery,'' in {\em Proceedings of the Twenty-First Annual ACM-SIAM Symposium
  on Discrete Algorithms}, SODA '10, pp.~1190--1197, 2010.

\bibitem{CandesDavenport}
E.~J. Cand{\`e}s and M.~A. Davenport, ``How well can we estimate a sparse
  vector?,'' {\em Applied and Computational Harmonic Analysis}, 2012.

\bibitem{ZhangZhangSignificance}
C.-H. Zhang and S.~Zhang, ``{Confidence Intervals for Low-Dimensional
  Parameters in High-Dimensional Linear Models},'' {\em Journal of the Royal
  Statistical Society: Series B (Statistical Methodology)}, 2013.

\bibitem{BuhlmannSignificance}
P.~B{\"u}hlmann, ``{Statistical significance in high-dimensional linear
  models}.'' {\sf arXiv:1202.1377}, 2012.

\bibitem{BickelEtAl}
P.~J. Bickel, Y.~Ritov, and A.~B. Tsybakov, ``{Simultaneous analysis of Lasso
  and Dantzig selector},'' {\em Amer. J. of Mathematics}, vol.~37,
  pp.~1705--1732, 2009.

\bibitem{GBR-hypothesis}
S.~van~de Geer, P.~B{\"u}hlmann, and Y.~Ritov, ``{On asymptotically optimal
  confidence regions and tests for high-dimensional models}.'' {\sf
  arXiv:1303.0518}, 2013.

\bibitem{confidenceJM}
A.~Javanmard and A.~Montanari, ``{Confidence Intervals and Hypothesis Testing
  for High-Dimensional Regression}.'' {\sf arXiv:1306.3171}, 2013.

\bibitem{GreenshteinRitov}
E.~Greenshtein and Y.~Ritov, ``Persistence in high-dimensional predictor
  selection and the virtue of over-parametrization,'' {\em Bernoulli}, vol.~10,
  pp.~971--988, 2004.

\bibitem{Dantzig}
E.~Cand\'es and T.~Tao, ``{The Dantzig selector: statistical estimation when p
  is much larger than n},'' {\em Annals of Statistics}, vol.~35,
  pp.~2313--2351, 2007.

\bibitem{WainwrightEllP}
G.~Raskutti, M.~J. Wainwright, and B.~Yu, ``{Minimax rates of estimation for
  high-dimensional linear regression over $\ell_q$-balls},'' in {\em 47th
  Annual Allerton Conf.}, (Monticello, IL), Sept. 2009.

\bibitem{zhao}
P.~Zhao and B.~Yu, ``{On model selection consistency of Lasso},'' {\em The
  Journal of Machine Learning Research}, vol.~7, pp.~2541--2563, 2006.

\bibitem{Wainwright2009LASSO}
M.~Wainwright, ``Sharp thresholds for high-dimensional and noisy sparsity
  recovery using $\ell_1$-constrained quadratic programming,'' {\em IEEE Trans.
  on Inform. Theory}, vol.~55, pp.~2183--2202, 2009.

\bibitem{BuhlmannVanDeGeer}
S.~van~de Geer and P.~B{\"u}hlmann, ``On the conditions used to prove oracle
  results for the lasso,'' {\em Electron. J. Statist.}, vol.~3, pp.~1360--1392,
  2009.

\bibitem{Wainwright2010Gaussian}
G.~Raskutti, M.~Wainwright, and B.~Yu, ``Restricted eigenvalue properties for
  correlated gaussian designs,'' {\em Journal of Machine Learning Research},
  vol.~11, pp.~2241--2259, 2010.

\bibitem{MeinshausenBuhlmannStability}
N.~Meinshausen and P.~B{\"u}hlmann, ``Stability selection,'' {\em J. R.
  Statist. Soc. B}, vol.~72, pp.~417--473, 2010.

\bibitem{Report-Semi}
C.-H. Zhang, ``Statistical inference for high-dimensional data,'' in {\em
  Workshop on Very High Dimensional Semiparametric Models, Report No. 48/2011},
  pp.~2772--2775, Mathematisches Forschungsinstitut Oberwolfach, Oct 2011.

\bibitem{SZ-scaledLasso}
T.~Sun and C.-H. Zhang, ``Scaled sparse linear regression,'' {\em Biometrika},
  vol.~99, no.~4, pp.~879--898, 2012.

\bibitem{HuberBook}
P.~Huber and E.~Ronchetti, {\em Robust Statistics (second edition)}.
\newblock J. Wiley and Sons, 2009.

\bibitem{DMM09}
D.~L. Donoho, A.~Maleki, and A.~Montanari, ``{Message Passing Algorithms for
  Compressed Sensing},'' {\em Proceedings of the National Academy of Sciences},
  vol.~106, pp.~18914--18919, 2009.

\bibitem{JavanMon-OptSample}
A.~Javanmard and A.~Montanari, ``{Nearly Optimal Sample Size in Hypothesis
  Testing for High-Dimensional Regression},'' in {\em 52nd Annual Allerton
  Conference}, (Monticello, IL), pp.~798 -- 805, Sept. 2013.
\newblock {\sf arXiv:1311.0274}.

\bibitem{BM-Universality}
M.~Bayati, M.~Lelarge, and A.~Montanari, ``{Universality in polytope phase
  transitions and message passing algorithms}.'' {\sf arXiv:1207.7321}, 2012.

\bibitem{oymak2013squared}
S.~Oymak, C.~Thrampoulidis, and B.~Hassibi, ``The squared-error of generalized
  lasso: A precise analysis,'' {\em {\sf arXiv:1311.0830}}, 2013.

\bibitem{TalagrandVolI}
M.~Talagrand, {\em Mean Field Models for Spin Glasses: Volume I}.
\newblock Berlin: Springer-Verlag, 2010.

\bibitem{panchenko2013sherrington}
D.~Panchenko, {\em {The Sherrington-Kirkpatrick model}}.
\newblock Springer, 2013.

\bibitem{guerra2003broken}
F.~Guerra, ``Broken replica symmetry bounds in the mean field spin glass
  model,'' {\em Communications in mathematical physics}, vol.~233, no.~1,
  pp.~1--12, 2003.

\bibitem{aizenman2003extended}
M.~Aizenman, R.~Sims, and S.~L. Starr, ``Extended variational principle for the
  sherrington-kirkpatrick spin-glass model,'' {\em Physical Review B}, vol.~68,
  no.~21, p.~214403, 2003.

\bibitem{TanakaCDMA}
T.~Tanaka, ``{A Statistical-Mechanics Approach to Large-System Analysis of CDMA
  Multiuser Detectors},'' {\em IEEE Trans. on Inform.~Theory}, vol.~48,
  pp.~2888--2910, 2002.

\bibitem{GuoVerdu}
D.~Guo and S.~Verdu, ``{Randomly Spread CDMA: Asymptotics via Statistical
  Physics},'' {\em IEEE Trans. on Inform. Theory}, vol.~51, pp.~1982--2010,
  2005.

\bibitem{tanaka2005approximate}
T.~Tanaka and M.~Okada, ``Approximate belief propagation, density evolution,
  and statistical neurodynamics for cdma multiuser detection,'' {\em
  Information Theory, IEEE Transactions on}, vol.~51, no.~2, pp.~700--706,
  2005.

\bibitem{campo2011large}
A.~T. Campo, A.~Guillen~i Fabregas, and E.~Biglieri, ``Large-system analysis of
  multiuser detection with an unknown number of users: A high-snr approach,''
  {\em Information Theory, IEEE Transactions on}, vol.~57, no.~6,
  pp.~3416--3428, 2011.

\bibitem{wu2012optimal}
Y.~Wu and S.~Verd{\'u}, ``Optimal phase transitions in compressed sensing,''
  {\em Information Theory, IEEE Transactions on}, vol.~58, no.~10, p.~6241,
  2012.

\bibitem{RanganFletcherGoyal}
S.~Rangan, A.~K. Fletcher, and V.~K. Goyal, ``{Asymptotic Analysis of MAP
  Estimation via the Replica Method and Applications to Compressed Sensing},''
  in {\em {NIPS}}, (Vancouver), 2009.

\bibitem{KabashimaTanaka}
Y.~Kabashima, T.~Wadayama, and T.~Tanaka, ``{A typical reconstruction limit for
  compressed sensing based on $L_p$-norm minimization},'' {\em J.Stat. Mech.},
  p.~L09003, 2009.

\bibitem{BaronGuoShamai}
D.~Guo, D.~Baron, and S.~Shamai, ``{A Single-letter Characterization of Optimal
  Noisy Compressed Sensing},'' in {\em 47th Annual Allerton Conference},
  (Monticello, IL), Sept. 2009.

\bibitem{TakedaKabashima}
K.~Takeda and Y.~Kabashima, ``Statistical mechanical analysis of compressed
  sensing utilizing correlated compression matrix,'' in {\em IEEE Intl. Symp.
  on Inform. Theory}, june 2010.

\bibitem{tulino2011support}
A.~Tulino, G.~Caire, S.~Shamai, and S.~Verd{\'u}, ``Support recovery with
  sparsely sampled free random matrices,'' in {\em Information Theory
  Proceedings (ISIT), 2011 IEEE International Symposium on}, pp.~2328--2332,
  IEEE, 2011.

\bibitem{KabashimaOrthogonal}
Y.~Kabashima and S.~C. M.~Vehkapera, ``{Typical $l_1$-recovery limit of sparse
  vectors represented by concatenations of random orthogonal matrices},'' {\em
  J. Stat. Mech.}, p.~P12003, 2012.

\bibitem{MoT06}
A.~Montanari and D.~Tse, ``{Analysis of belief propagation for non-linear
  problems: the example of CDMA (or: how to prove Tanaka's formula)},'' in {\em
  Proceedings of IEEE Inform. Theory Workshop}, (Punta de l'Este, Uruguay),
  2006.

\bibitem{chapelle2006semi}
O.~Chapelle, B.~Sch{\"o}lkopf, A.~Zien, {\em et~al.}, {\em Semi-supervised
  learning}.
\newblock Cambridge: MIT Press, 2006.

\bibitem{Zhao-Cov}
Z.~Ren, T.~Sun, C.-H. Zhang, and H.~H. Zhou, ``{Asymptotic Normality and
  Optimalities in Estimation of Large Gaussian Graphical Model}.'' {\sf
  arXiv:1309.6024}, 2013.

\bibitem{CLIME}
T.~Cai, W.~Liu, and X.~Luo, ``A constrained $\ell_1$ minimization approach to
  sparse precision matrix estimation,'' {\em Journal of the American
  Statistical Association}, vol.~106, no.~494, pp.~594--607, 2011.

\bibitem{bickel2008regularized}
P.~J. Bickel and E.~Levina, ``Regularized estimation of large covariance
  matrices,'' {\em The Annals of Statistics}, pp.~199--227, 2008.

\bibitem{le1956asymptotic}
L.~Le~Cam, ``On the asymptotic theory of estimation and testing hypotheses,''
  in {\em Proceedings of the Third Berkeley Symposium on Mathematical
  Statistics and Probability}, vol.~1, pp.~129--156, University of California
  Press Berkeley, CA, 1956.

\bibitem{van2000asymptotic}
A.~W. Van~der Vaart, {\em Asymptotic statistics}.
\newblock Cambridge University Press, 2000.

\bibitem{FrankAsuncion2010}
A.~Frank and A.~Asuncion, ``{UCI} machine learning repository (communities and
  crime data set).'' {\sf http://archive.ics.uci.edu/ml}, 2010.
\newblock {University of California, Irvine, School of Information and Computer
  Sciences}.

\bibitem{rudelson2013reconstruction}
M.~Rudelson and S.~Zhou, ``Reconstruction from anisotropic random
  measurements,'' {\em Information Theory, IEEE Transactions on}, vol.~59,
  no.~6, pp.~3434--3447, 2013.

\bibitem{Ledoux}
M.~Ledoux, ``{The concentration of measure phenomenon},'' in {\em Mathematical
  Surveys and Monographs}, vol.~89, {American Mathematical Society, Providence,
  RI}, 2001.

\bibitem{SpinGlass}
M.~M\'ezard, G.~Parisi, and M.~A. Virasoro, {\em Spin Glass Theory and Beyond}.
\newblock World Scientific, 1987.

\bibitem{MezardMontanari}
M.~M{\'e}zard and A.~Montanari, {\em {Information, Physics and Computation}}.
\newblock Oxford, 2009.

\bibitem{Vershynin-CS}
R.~Vershynin, ``Introduction to the non-asymptotic analysis of random
  matrices,'' in {\em Compressed Sensing: Theory and Applications} (Y.~Eldar
  and G.~Kutyniok, eds.), pp.~210--268, Cambridge University Press, 2012.

\end{thebibliography}

\end{document}